%% file: CCCPJournal.tex
\documentclass{LMCS}

\def\dOi{11(1:18)2015}
\lmcsheading%
{\dOi}
{1--59}
{}
{}
{Jan.~13, 2013}
{Mar.~31, 2015}
{}

\ACMCCS{[{\bf Theory of computation}]: Models of computation---Concurrency---Process Calculi}


\usepackage[british]{babel}
\usepackage{timestamp}
\usepackage{tikz}
\usetikzlibrary{arrows,patterns}
\usepackage{networks}
\newenvironment{definition}{\begin{defi}\rm}{ \hfill\qed \end{defi}}
\newenvironment{theorem}{\begin{thm} \rm}{\end{thm}}
\newenvironment{proposition}{\begin{prop} \rm}{\end{prop}}
\newenvironment{lemma}{\begin{lem} \rm}{\end{lem}}
\newenvironment{corollary}{\begin{cor}\rm}{\end{cor}}
\newenvironment{example}{\begin{exa}\rm}{ \hfill\qed\end{exa}}
\newenvironment{remark}{\begin{rem}\rm}{ \hfill\qed\end{rem}}

\usepackage{prooftree,proofs}
\usepackage{amsmath} 
\usepackage{amssymb,stmaryrd}
\usepackage{float} 
\usepackage{graphicx}

\floatstyle{ruled}
\restylefloat{table}
\restylefloat{figure}
\usepackage{basic}
\usepackage{ambients}
\usepackage{hyperref}
\usepackage{xspace}
\usepackage{txfonts}
\usepackage{macros}

\definecolor{greenish}{rgb}{.24,.5,.26} % Only RGB, CMYK and BW are predefined.
\newcommand{\MHx}{\color{greenish}} % M's for Matthew
\newcommand{\MHc}[1]{ {\MHx   #1}}

\renewcommand{\MHc}[1]{{#1}}
\newcommand{\MHf}[1]{\marginpar{\MHx~\bf\fbox{\footnotemark}~}\footnotetext{ \color{greenish}   #1}}
\renewcommand{\MHf}[1]{}

\newcommand{\ffalse}{\ensuremath{\text{false}}}
\newcommand{\rcvno}[1]{\ensuremath{\mathop{\#\textsf{Rcv}({#1})}}}

\newcommand{\pre}{\ensuremath{\text{pre}}}
\newcommand{\post}{\ensuremath{\text{post}}}

\title[Modelling  MAC-layer communications in wireless systems]{Modelling  MAC-layer communications in wireless systems\rsuper*}
\titlecomment{{\lsuper*}An extended abstract appeared in the proceeding of the \emph{5th International Conference on Coordination Models and Languages} (COORDINATION 2013), volume 7890 of \emph{Lecture Notes in Computer Science\/}, pp.\ 16-30, Springer, 2013.}
\author[A. Cerone]{Andrea Cerone\rsuper a}
\address{{\lsuper a}IMDEA Software Institute, Spain}
\email{andrea.cerone@imdea.org}
\author[M. Hennessy]{Matthew Hennessy\rsuper b}
\address{{\lsuper b}School of Statistics and Computer Science, Trinity College Dublin, Ireland}
\email{Matthew.Hennessy@cs.tcd.ie}
\thanks{{\lsuper{a,b}}The first and the second authors are supported by SFI project SFI 06 IN.1 1898.}
\author[M. Merro]{Massimo Merro\rsuper c}
\address{{\lsuper c}Dipartimento di Informatica, Universit\`a degli Studi di Verona, Italy}
\email{massimo.merro@univr.it}
\thanks{{\lsuper c}The third author was partially supported by the PRIN 2010-2011 national project ``Security Horizons''.}
\begin{document} 

\begin{abstract}
  We present a timed process calculus for modelling wireless networks
  in which individual stations broadcast and receive messages;
  moreover the broadcasts are subject to collisions.  Based on a
  reduction semantics for the calculus we define a contextual
  equivalence to compare the external behaviour of such wireless
  networks.  Further, we construct an extensional 
\MHc{LTS (labelled transition system) which models the
  activities of stations that can be directly observed by the external
  environment. Standard bisimulations in this LTS provide a sound proof method
  for proving systems contextually equivalence. We illustrate the usefulness of the
  proof methodology by a series of examples. Finally we  show
  that this proof method is also complete, for a large class of systems.}
  %\footnote{stamp: \timestamp}

\end{abstract}

\maketitle

%%%%%%%%%%%%%%%%%%%%%% INTRO %%%%%%%%%%%%%%%%%%%%%%%%%%%%%%%%%%%%%%%

%\tableofcontents\MHf{Table of contents not to appear in final document}

\section{Introduction}
\label{sec:intro}

Wireless networks are becoming increasingly pervasive with
applications across many domains, \cite{rappaport,wirelesssurvey}.
They are also becoming increasingly complex, with their behaviour
depending on ever more sophisticated protocols. There are different
levels of abstraction at which these can be defined and implemented,
from the very basic level in which the communication primitives
consist of sending and receiving electromagnetic signals, to the
higher level where the basic primitives allow the 
initiation of connections between nodes in a wireless system and the exchange of
data between them \cite{tanenbaum}.
%set up of  connections and exchange of data between two nodes 
%in a wireless system 

Assuring the correctness of the behaviour of a wireless network 
has always been difficult. \MHc{Several} approaches have been 
proposed to address this issue for networks described at a high level 
\cite{nanz,Merro07,Godskesen07,GhasWanFok10,SRS10,KouzapasP11,Borgstrom_etal11,Cerone_Hennessy2013};
\MHc{these typically}
allow the formal description of protocols at the  
\emph{network layer} of the \emph{TCP/IP} reference model 
\cite{tanenbaum}.  \MHc{However}
there are few frameworks in the literature which consider networks 
described at the MAC-Sublayer of the \emph{TCP/IP} reference model \cite{LaneseS10,MerroBS11,BGMR12,Wang13}. 
This is the topic of the current paper.  
We propose a process calculus for describing and verifying wireless networks at the 
\emph{MAC-Sublayer} of the \emph{TCP/IP} reference model. 

This calculus, called the Calculus of Collision-prone Communicating Processes (CCCP), 
has been largely inspired by TCWS \cite{MerroBS11};
in particular CCCP  inherits its communication features but simplifies considerably
the syntax, the reduction semantics, the notion of observation, and as we will see the behavioural theory. 
%, and it can be seen as a simplification 
%of the latter from which CCCP 
%inherits its communication features . 
%Specifically, 
In CCCP a wireless system is considered to be a collection of wireless stations which 
transmit and receive messages.
The transmission of messages is \emph{broadcast\/}, and it is 
\emph{time-consuming\/}; 
the transmission of a message $v$ can require several time slots (or instants). 
In addition, wireless stations in our calculus are sensitive to \emph{collisions\/}; 
if two different stations are transmitting a value over a channel $c$ at the same time slot 
then a collision occurs; as a result, the content of the messages originally being transmitted is lost.

More specifically, in CCCP a state of a wireless network (or simply network, or system) 
will be described by a \emph{configuration}
of the form 
\begin{math}
  \conf{\Gamma}{ W}
\end{math}
where $W$ describes the code running at individual wireless stations
and $\Gamma$ represents the communication state of channels.  
At any given point of time there may be \emph{exposed}
communication channels, that is channels containing messages (or values) in
transmission; this information will be recorded in $\Gamma$. 

Such systems evolve by the broadcast of  messages between stations, the passage 
of time, or some other internal activity, such as the occurrence of collisions 
and their consequences. One of the topics of the paper is to capture
formally these complex evolutions, by defining  a \emph{reduction
  semantics}, whose judgements take the form 
$\conf{\Gamma_1}{ W_1} \red \conf{\Gamma_2}{ W_2 }$. 
We show that the reduction semantics we propose satisfies 
some desirable time properties such as \emph{time determinism\/}, \emph{maximal progress\/} and \emph{patience\/}~\cite{Sifakis94,HennessyR95,Yi91}. 

However the main aim of the paper is to develop a behavioural theory
of wireless networks with time-consuming communications. To this end we need a formal notion of when two
such systems are indistinguishable from the point of view of
users. Having a reduction semantics it is now straightforward to adapt
a standard notion of \emph{contextual equivalence}: 
$\conf{\Gamma_1}{ W_1} \simeq \conf{\Gamma_2}{ W_2 }$.
Intuitively this means that either system, $\conf{\Gamma_1}{ W_1}$ or
$\conf{\Gamma_2}{ W_2}$, can be replaced by the other in a larger
system without changing the observable behaviour of the overall
system. Formally, we use the approach of 
\cite{hy,EnvBisim07}, 
often called
\emph{reduction barbed congruence}, 
rather
than that of \cite{MiSa92}\footnote{See page 106 of \cite{SaWabook} for a brief
discussion of the difference.}.  
The only parameter in the definition of our contextual equivalence is the choice of
primitive observation or \emph{barb};  our choice is natural for wireless systems:
the ability  to transmit on
an idle 
(or unexposed) 
channel, that is a channel with 
no active transmissions.

As explained in papers such as \cite{RS08-dbtm,dpibook}, contextual
equivalences are determined by so-called \emph{extensional actions},
that is the set of minimal observable interactions which a system can
have with its external environment. For CCCP determining these actions
is non-trivial. Although values can be transmitted and received on
channels, the presence of collisions means that these are not
necessarily observable. 
%It turns out that 
In fact the important point is not
the transmission of a value, but its successful delivery.  Also,
although the basic notion of observation on systems does not involve
the recording of the passage of time, 
this has to be taken into account extensionally in order  to gain a proper extensional account of
systems. 

The extensional semantics determines an LTS (labelled transition
  system) over \MHc{configurations}, which in turn gives rise to the standard
  notion of (weak) bisimulation equivalence between configurations. This
  gives a powerful co-inductive proof technique: to show that two
  systems are behaviourally equivalent it is sufficient to exhibit a
  witness bisimulation which contains them. 
  
  One result of this paper is that weak bisimulation in the extensional LTS
is  sound with respect to the touchstone contextual equivalence:
if two systems are related by some bisimulation in the extensional
LTS then they are contextually equivalent.  In order to show the effectiveness
 of our bisimulation proof method we 
prove a number of non-obvious system equalities. 
However, the main contribution 
of the current paper 
is that completeness holds 
for a large class of networks, called \emph{well-formed}.
If two such 
networks are contextually equivalent
then there is some bisimulation, based on our novel extensional actions, 
which contains them. 

To the best of our knowledge, this is the first result of 
full abstraction for weak barbed congruence, for a calculus 
of wireless systems where communication is subject to collisions. 
Also, the only other result in the field of which we are aware is 
the one illustrated in \cite{MerroBS11}. Here a sound but not complete bisimulation based
proof method is developed for (a different form of) reduction barbed congruence. 
In this paper, both soundness and completeness are achieved by simplifying 
the calculus and isolating novel extensional actions.

We end this introduction with an outline of the paper. 
In Section \ref{sec:calculus} we present the calculus CCCP. 
More precisely, 
Section \ref{sec:syntax} contains the syntax of our language; 
Section \ref{sec:intsem} introduces the intensional 
semantics; here the adjective \emph{intensional} is used to 
stress the fact that the actions of this semantics correspond 
to those activities which can be performed by a network. 
Section \ref{sec:redsem}
provides the reduction 
semantics, which models the intra-actions that can be performed by a network 
when isolated from the external environment.; Section \ref{sec:cxt} defines
 our touchstone contextually-defined behavioural equivalence for comparing 
wireless networks. 

In Section \ref{sec:extsem} we address the problem of 
defining the minimal observable activities of systems. 
These are defined as actions of an extensional 
semantics in Section \ref{sec:ea}, while in Section \ref{sec:bisim} 
we consider the bisimulation principle induced by such actions. 
Here the adjective \emph{extensional} is used to stress the fact that 
the actions of such a semantics correspond to those activities 
which can be observed by the external environment of a network.

In Section \ref{sec:fullabstraction} we present the main results of the paper. 
First we prove that our bisimulation proof technique is sound with respect to the 
contextual equivalence, Section \ref{sec:soundness}. 
In Section \ref{sec:completeness} we prove that, for 
a large class of configurations,  called well-formed, 
 our proof technique is also complete. 

The usefulness of our bisimulation proof technique is shown in Section \ref{sec:apps}, 
where we consider simple case studies which model common features of 
wireless networks at the Mac-Layer. 

Section \ref{sec:conclusions} concludes the paper with a comparison with the related work.

\section{The calculus}
\label{sec:calculus}
As already discusses 
a wireless system will be represented in our calculus as a \emph{configuration} of the form 
\begin{math}
  \conf{\Gamma}{ W},
\end{math}
 where $W$ describes the code running at individual wireless stations and $\Gamma$ 
is a channel environment containing the transmission information for channels. A possible evolution of a system will then be given by a sequence
of computation steps:
\begin{align}\label{eq:steps}
  &\conf{\Gamma_1}{ W_1} \red \conf{\Gamma_2}{ W_2 } \red \ldots \ldots \red 
\conf{\Gamma_k}{W_k} \ldots \red \ldots 
\end{align}
where intuitively each step corresponds to either the passage of time, a broadcast  from 
a station, or some unspecified internal computation; the code running at stations evolves as a 
computation proceeds, but so also does the state of the underlying 
channel environment. In the following we will use the meta-variable $\confC$ to 
range over configurations.

\subsection{Syntax}
\label{sec:syntax}
\begin{table}[!t]
  \begin{center}
    \begin{math}
 \arraycolsep 0pt
      \begin{array}{@{\hspace*{4mm}}rcl@{\hspace*{40mm}}l@{\hspace*{12mm}}rcl@{\hspace*{7mm}}l}
          \\[2pt]
          W & \Bdf & P & \mbox{station code}\\[2pt]
          & \Bor & \arcv c x P & \mbox{active receiver}\\[2pt]
          & \Bor & W_1 | W_2 & \mbox{parallel composition}\\[2pt]
          & \Bor & \crest{c}{(n,v)}.W & \mbox{channel restriction}
         \\%[4pt]
\\%[6pt]
          P, Q
          & \Bdf & 
           \bcastc  c e  P
          & \mbox{broadcast} \\[2pt]
                    & \Bor & 
           \rcvtimec c x P Q & \mbox{receiver with timeout}\\[2pt]
          &\Bor  &\delay \sigma P
          & \mbox{delay} \\[2pt]
           & \Bor & \tau.P & \mbox{internal activity}\\[2pt]
          & \Bor & 
         P + Q  & \mbox{choice}\\[2pt]
  & \Bor &
          \matchb{b} P Q
          & \mbox{matching}\\[2pt]
          & \Bor & X & \mbox{process variable}\\[2pt]
          & \Bor &   \nil           & \mbox{termination} \\[2pt]
          & \Bor & 
           \fix X P
          & \mbox{recursion} 
         \\%[4pt]
\\[2pt]
\multicolumn{4}{l}{\textit{Channel Environment:}
\Q \Gamma : \chanset \rightarrow \mathbb{N} \times \valset}\\[6pt]
         \end{array}
    \end{math}
   \end{center}
\caption{CCCP: Syntax}
\label{tab:syntax}
\end{table}
Formally we assume a set of channels $\chanset$, ranged over by
  $c,d,\cdots$,  and a set of values $\valset$, which contains 
a set of data-variables, ranged over by $x,y,\cdots$ and a special value
$\err$; this value will be used to denote  faulty transmissions. 
The set of \emph{closed values}, that is those not containing occurrences of 
variables, are ranged over by $v,w,\cdots$. 
We
also assume that every closed value $v \in \valset$ has an
  associated strictly positive integer $\delta_v$, which denotes the
number of time slots needed by a wireless station to transmit
$v$.
Finally, we assume a language of expressions ${\mathbf{Exp}}$ 
which can be built from values in $\valset$; we also assume a 
function $\interpr{\cdot}$, for evaluating expressions 
with no occurrences of data-variables into closed values.

A channel environment is a mapping $\Gamma : \chanset \rightarrow
\mathbb{N} \times \valset$. In a configuration $\conf{\Gamma}{W}$
where $\Gamma(c) = (n,v)$ for some channel $c$, there is a wireless station
which is
currently transmitting the value $v$ for the next $n$ time slots.  We
will use some suggestive notation for channel environments:
$\Gamma \vdash_{\mathrm{t}} c:n$ in place of
$\Gamma(c) = (n,w)$ for some $w$, $\Gamma \vdash_{\mathrm{v}} c:w$ in place of
$\Gamma(c) = (n,w)$ for some $n$.  If $\Gamma \vdash_{\mathrm{t}} c : 0$ we say
that channel $c$ is idle in $\Gamma$, and we denote it with $\Gamma
\vdash c: \cfree$. Otherwise we say that $c$ is exposed in $\Gamma$,
denoted by $\Gamma \vdash c: \cbusy$.  The channel environment $\Gamma$ 
such that $\Gamma \vdash c: \cfree$ for every channel $c$ 
is said to be \emph{stable}. 
Often we will compare channel environments according to the amount of time 
instants for which channels will be exposed; we say that 
$\Gamma \leq \Gamma'$ if, for any channel $c$, $\Gamma \vdash_{\mathrm{t}} c : n$ 
implies $\Gamma' \vdash_{\mathrm{t}} c: m$, for some $m$ such that  $n \leq m$.

The syntax for system terms $W$ is given in Table~\ref{tab:syntax},
\MHc{where $P$} ranges over code for  programming individual
stations, which is also illustrated in Table~\ref{tab:syntax}. 
A system term $W$ is a collection of individual threads running in parallel, with
possibly some channels restricted. As we will see in Section \ref{sec:apps}, 
channel restriction can be used to model non-flat network topologies.

Each thread may be either an inactive piece of code $P$ or
an active code of the form $\arcv c x P$. This latter term represents a wireless station
which is receiving a value from the channel $c$; when the
value is eventually received the variable $x$ will be replaced with the 
received value in the 
code $P$. 

The syntax for station code is based on standard process
calculus constructs. 
The main constructs are 
time-dependent reception from a channel $\rcvtimec c x P Q$, 
explicit time delay $\sigma.P$, and broadcast along a channel $ \bcastc c {e}
  P$; here the value being broadcast is the one obtained by evaluating $e$ via 
  the function $\interpr{\cdot}$, provided that $e$ does not contain any occurrence 
  of data-variables.
%Here $e$ denotes a term from some language of expressions constructed from a 
%set of data-variables $x,y,\cdots$ and a set of closed values $\valset$. Elements of $\valset$ are ranged over by $v,w,\cdots$. 
%Expressions with no occurrences of data-variables are said to be closed; 
%we assume an interpretation function $\interpr{e}$, which maps closed expressions with no 
%into (closed) values.
Of the remaining
standard constructs the most notable is matching, $ \matchb{b} P Q$
which branches to $P$ or $Q$, depending on the value of the Boolean
expression $b$. Such boolean expressions can be either 
equality tests of the form $e_1 = e_2$, or terms of the form 
$\expsd{c}$, which will be used to check whether channel $c$ 
is exposed, that is it is being used for transmission. 
%We leave the language of Boolean expressions
%unspecified, other than saying that it should contain equality tests
%for values, $u_1= u_2$.   More importantly, it should also contain the
%expression $\expsd{c}$ for checking if in the current configuration
%the channel $c$ is exposed, that is it is being used for
%transmission.

In the construct $\fix X P$ occurrences of the recursion variable $X$ in $P$ are bound; 
similarly in
the terms $\rcvtimec c x P Q$ and   $\arcv c x P$  
 the data-variable $x$ is 
bound in $P$.
This gives
rise to the standard notions of free and bound variables, $\alpha$-conversion
and capture-avoiding substitution; 
In a configuration of the form $\conf{\Gamma}{W}$, we 
assume that $W$ is closed, meaning that all its occurrences of 
both data-variables and process variables are bound. 
In general, we always assume that a system term $W$ is closed, 
unless otherwise stated. Sometimes we will 
need to consider system terms with free occurrences of process variables, 
we will explicitly say that they are open system terms.
System terms, both open and closed, are identified up to $\alpha$-conversion. 
We assume that all occurrences of recursion variables are  \emph{guarded}; 
they must occur within either a broadcast, input residual, timeout branch, time delay prefix, or within an 
execution branch of a matching construct. 
This ensures that  recursive calls cannot be used to build up infinite loops 
within a time slot
%As we will see when discussing the operational semantics 
%of CCCP, this limitation ensures that nested recursive calls cannot appear in the same time unit.
% throughout the paper we use the metavariable 
%$\confC$ to range over configurations.

\begin{example}
\label{ex:collision.syntax}
Consider the configuration 
$$\confC_1 = \conf{\Gamma}{S_1 | S_2 | R_1}$$ where 
\begin{eqnarray*}
S_1 &=& \bcastc{c}{v_0}{\nil}\\
S_2 &=& \sigma.\bcastc{c}{v_1}{\nil}\\
R_1 &=& \rcvtimec{c}{x}{P}{\nil}
\end{eqnarray*}
and $\Gamma$ is the stable channel environment. Further, 
we assume that $\delta_{v_0} = 2$ and $ \delta_{v_1} = 1$. 
This configuration contains two sender stations, running the code $S_1$ and 
$S_2$, respectively,  and a receiving station, running the code $R_1$. 
In the first time slot, the station running the code $S_1$ broadcasts the 
value $v_0$ along channel $c$. 
The station running the code $R_1$ starts receiving such a value and 
it will be busy in 
receiving it for the next two time slots. 
In the first time slot the station running the code $S_2$ is idle. 
It is only in the second time slot that this station will broadcast a value along 
channel $c$. At this point the receiving station will be exposed to two transmissions; 
the transmission of value $v_0$, which is still in progress, and the transmission of 
value $v_1$. As a result, a collision happens, and the value received by the receiver 
will be at the end error value $\err$. 

The formal behaviour of the configuration $\confC_1$ will be explained in Example 
\ref{ex:collision}.
\end{example}

We use a number of notational conventions. $\prod_{i \in I}W_i$ means
the parallel composition of all stations $W_i$, for $i \in I$.  
We identify $\prod_{i \in I}W_i$ \MHc{with $\nil$} if $I = \emptyset$.
We will omit trailing occurrences of $\nil$, render
$\crest{c}{(n,v)}.W$ as $\nu c.W$ when the values $(n,v)$ are not
relevant to the discussion, and use $\nu {\tilde c}.W$ as an abbreviation for a
sequence $\tilde c$ of such restrictions. We write  
%%%$\matchb b P \nil$ with $[b]P$ and 
$\rcvtimec c x P {}$ for $\rcvtimec c x P \nil$. Finally, 
we abbreviate the recursive process $\fix X {\rcvtimec c x P X}$ with 
$\rcvc c x P$; \MHc{as we will see}  this is a persistent listener at channel 
$c$ waiting for an incoming message. 

\subsection{Intensional semantics}
 \label{sec:intsem}
Our first goal is to formally define  computation steps among configurations
 of the form $\conf
{\Gamma_1} {W_1} \red \conf {\Gamma_2} {W_2}$. 
In order to do that, we first 
define the evolution of system terms with respect to a channel
environment $\Gamma$ via a set of SOS rules whose judgements take the
form $\conf{\Gamma}{W_1} \trans{\lambda} W_2$, were $\lambda$ is 
an intensional action taking one of the following forms: 
\begin{enumerate}%[(1)]
\item $\sndac c v$, denoting a station starting broadcasting 
value $v$ along channel $c$
\item $\sigma$, denoting the passage of one time slot, or time instant
\item $\tau$, denoting an internal action
\item $\rcva c v$, denoting a station in the external environment starting 
broadcasting value $v$ on channel $c$. 
\end{enumerate}
These actions $\lambda$ will have an
effect also on the channel environment, which we  describe
by means of a
functional $\gupd{\lambda}{\cdot} : {\mathbf{Env}} \rightarrow
{\mathbf{Env}}$, where ${\mathbf{Env}}$ is the set of channel
environments.

\begin{definition}\label{def:conf.after}[Channel Environment update]
 Let $\Gamma \in \mathbf{Env}$ be an arbitrary channel environment and 
 $c \in \chanset$ an arbitrary channel. Let $t_c$ 
 and $v_c$ be the exposure time and the value transmitted along 
channel $c$ in $\Gamma$, respectively, that is $\Gamma \vdash_{\mathrm{t}} c: t_c$ 
and $\Gamma \vdash_{\mathrm{v}} c: v_c$. 
For any intensional action $\lambda$, we let 
$\gupd \lambda \Gamma$ 
be the unique channel environment determined by the following 
definitions:\footnote{For convenience
                 we assume $0-1$ to be $0$.}
\begin{enumerate}%[(1)]
\item $\tupd \Gamma \vdash_{\mathrm{t}} c : t_c - 1$ and $\tupd \Gamma \vdash_{\mathrm{v}} c : v_c$; 
\item
for any value $v \in \valset$, let $\gupd{c!v}{\Gamma}$ be the channel environment such that 
  \begin{align*}
       &\gupd{c!v} {\Gamma} \vdash_{\mathrm{t}} c :
       \begin{cases}
                        \deltav  &\, \text{if}\,\,  \Gamma \vdash c:\cfree \\
                       \max(\delta_v, t_c)& \, \text{if}\,\,  \Gamma \vdash c:\cbusy\\        
      \end{cases}
       & \gupd{c!v} {\Gamma}  \vdash_{\mathrm{v}} c :  
      \begin{cases}
                        v  &\, \text{if}\,\,  \Gamma \vdash c:\cfree \\
                       \err & \, \text{if}\,\,  \Gamma \vdash c:\cbusy\\        
      \end{cases}
  \end{align*}
  and for any channel $d$, $d \neq c$, let  $\gupd{c!v}{\Gamma} \vdash_{\mathrm{t}} d: t_d$ and  
  $\gupd{c!v}{\Gamma} \vdash_{\mathrm{v}} d : v_d$;  
\item for any value $v$, $\gupd{c?v}{\Gamma} = \gupd{c!v}{\Gamma}$; 
\item $\gupd{\tau}{\Gamma} = \Gamma$.
\end{enumerate}
\end{definition} 

Let us describe the intuitive meaning of this definition. When
time passes, the time of exposure of each channel decreases by one
time unit.  
The predicates $\gupd{c!v}{\Gamma}$ and 
$\gupd{c?v}{\Gamma}$ model how collisions are handled in our
calculus. When a station begins broadcasting a value $v$ over an 
idle channel $c$ this channel becomes exposed for the amount of time
required to transmit $v$, that is $\delta_v$.  If the channel is not
idle a collision happens. As a consequence, the value that will be
received by a receiving station, when all transmissions over channel
$c$ terminate, is the error value $\err$, and the exposure time is
adjusted accordingly.
Finally the definition of $\gupd{\tau}{\Gamma}$ reflects the
intuition that internal activities do not affect the exposure state
of channels.

Let us turn our attention to the  intensional semantics of system terms.
For the sake of clarity, the inference rules for the evolution of system terms, 
$\conf{\Gamma}{W_1} \trans{\lambda} W_2$, 
are split in four tables, each one focusing 
on a particular form of activity.

\begin{table}[!t]
\[
  \begin{array}{llll}
    &\Txiom{Snd}{\interpr{e} = v 
}
         {\conf \Gamma { \bcastc c e P}
          \trans{\sndac c v} {\asnd \sigma \deltav P}}
&&
\Txiom{Rcv}{  
\Gamma \vdash c:\cfree}
          {\conf \Gamma {\rcvtimec c x P Q}
          \trans{\rcva c v} {\arcv c x P}
         }
\\\\
%&\Txiom{RcvFail}{
%\Gamma \vdash c:\cbusy}
%          {\conf \Gamma {W}
%          \trans{c?v} {\conf \Gamma {W}
%         }}
%&&
&
\Txiom{RcvIgn}{\neg\isrcv{\conf{\Gamma}{W},c}}
{\conf \Gamma W \trans{c?v} {W}}
%\\\\
&&
\Txiom{Sync}{\conf \Gamma W_1 \trans{c!v}  {W_1'} \Q 
 \conf \Gamma W_2 \trans{c?v}  {W_2'}}
{\conf \Gamma W_1 | W_2 \trans{c!v}  {W_1' | W_2'}}
\\\\
&\multicolumn{3}{c}{
\Txiom{RcvPar}{\conf \Gamma W_1 \trans{\rcva c v}  {W_1'} \Q 
\conf \Gamma W_2 \trans{\rcva c v}  {W_2'}}
{\conf \Gamma W_1 | W_2 \trans{\rcva c v}  {W_1' | W_2'}}
}
  \end{array}
 \]
     \caption{Intensional semantics:  transmission}
    \label{tab:proc}
  \end{table} 

Table~\ref{tab:proc} contains the rules governing transmission.
Rule (Snd) models a non-blocking  broadcast of a message  
along channel $c$.  The value $v$ sent by process $\bcastc c e P$ is
the one obtained by evaluating an expression $e$; note that here we are 
assuming that $e$ is closed, hence we can evaluate it to a closed value 
via the function $\interpr{\cdot}$.
A transmission can fire at any time,  independently
on the state of the network; the
notation $\sigma^{\deltav}$ represents the time delay operator
$\sigma$ iterated $\deltav$ times.  So when the process $\bcastc c v
P$ broadcasts, it has to wait $\deltav$ time units (the time required to transmit $v$) before the residual
$P$ is activated.  On the other hand, reception of a message by a
time-guarded listener $\rcvtimec c x P Q$ depends on the state of the
channel environment. If the channel $c$ is free then rule (Rcv) indicates that
reception can start and the listener evolves into the active receiver
$\arcv c x P$. 

%If the channel is already exposed then by rule (RcvFail) the
%transmission is ignored and the reception is doomed to fail.  This
%rule reflects the fact that, in general, collisions can be detected
%only at the end of a transmission.  Although there are some protocols
%which allow a station to discover prematurely if a transmission is
%colliding with another one, \MHc{for engineering reasons} this is
%rarely done in wireless networks; see \cite{tanenbaum}, page 186 for a
%discussion.
%
%If the channel is already exposed, then by Rule \rulename{RcvFail} 
%the transmission is ignored. Note that, if 
%$W$ is a receiver process of the form $\rcvtimec c x P Q$, 
%then at least intuitively the reception of the message should start and 
%$W$ should evolve into an active receiver. Instead, the evolution of a receiver along an exposed channel into an active receiver will be modelled as an internal action. See Rule \rulename{RcvLate} Table \ref{tab:net3}, Example \ref{ex:intsem3} and Remark \ref{ex:rcvfail}.
%Intuitively, Rule \rulename{RcvFail}
%has been introduced to ensure that any configuration can perform 
%an input action, which is a desirable property in broadcast calculi; 
%see Lemma \ref{lem:rcv-enabling} and Remark \ref{rem:input}.  
Rule \rulename{RcvIgn} states that 
if a system term $W$ is not waiting for a message  
along a channel $c$, or if $c$ is already exposed, then any broadcast 
along $c$ is ignored by the configuration $\conf{\Gamma}{W}$. 
Here $\isrcv{\conf{\Gamma}{W}, c}$ is a predicate which evaluates to true 
in the case that in $\conf{\Gamma}{W}$ channel $c$ is not exposed, 
and
%
%
%\MHc{The rule (RcvIgn)
%says that if a system can not receive on the channel $c$ then any transmission
%along it is ignored.}
%Intuitively,
%the predicate $\isrcv{W,c}$ means that 
$W$ contains among its parallel 
\MHc{components
at least one} non-guarded 
 receiver of the form $\rcvtimec c x {P} {Q}$ which is actively 
\MHc{awaiting  a message}. 
Formally, we first define a predicate $\isrcv{W,c}$ for open 
terms, which is then lifted to configurations. For open terms we have
 $\isrcv{\conf{\Gamma}{W},c}$ is defined inductively as 
\begin{eqnarray*}
\isrcv{P,c} = \ffalse&\mbox{provided }& P = \bcastc c e Q, P = \tau.Q \mbox{ or } P = X\\
\isrcv{\rcvtimec d x P Q, c} =\ttrue&\text{if and only if}& d=c\\
\isrcv{P+Q,c} = \ttrue &\mbox{if and only if}& \isrcv{P,c} = \ttrue \mbox{ or } \isrcv{Q,c} = \ttrue \\
\isrcv{\fix X P,c} = \ttrue &\mbox{if and only if}& \isrcv{P,c} = \ttrue\\
&&\\
\isrcv{\arcv c x P, d} = \ffalse &\mbox{always}\\
\isrcv{W_1 | W_2,c} =\ttrue &\mbox{if and only if}& \isrcv{W_1, c} = \ttrue \mbox{ or } \isrcv{W_2,c} = \ttrue\\
\isrcv{\nu d.W,c} = \ttrue &\mbox{if and only if}& \isrcv{W,c} = \ttrue \mbox{, where we assume } d \neq c
\label{page:isrcv}
\end{eqnarray*}
Then, for any configuration $\conf{\Gamma}{W}$, we let $\isrcv{\conf{\Gamma}{W}, c} = \ttrue$ if
and only if $\Gamma \vdash c: \cfree$ and $\isrcv{W,c} = \ttrue$.

The remaining two rules in Table~\ref{tab:proc} (Sync) and (RcvPar) 
serve to synchronise parallel stations on the same transmission
 \cite{HeRa98,Sifakis94,CBS}.

\begin{example}[Transmission]
\label{ex:intsem1}
Let $\confC_0 = \conf{\Gamma_0}{W_0}$, where 
 $W_0 =
  \bcastzeroc{c}{v_0} \,|\,$ 
$\rcvtimec d x \nil {(\rcvtimec c x Q {})} \,|\, \rcvtimec c x P {}$, 
%\end{align*}
with $\delta_{v_0} = 2$, and $\Gamma_0$ a stable environment. 
%Consider 
%the configuration $\confC_0 = \conf{\Gamma_0}{W_0}$. 

Using rule \rulename{Snd} we can infer 
$\conf{\Gamma_0}{\bcastzeroc c {v_0}} \trans{\snda c {v_0}} \sigma^2$; 
this station starts transmitting the value $v_0$ along channel $c$.
Rule \rulename{RcvIgn} can be used to derive 
the transition $\conf{\Gamma_0}{\rcvtimec d x {\nil} {(\rcvtimec c x Q {})}} \trans{\rcva {c} {v_0}} 
\rcvtimec d x {\nil} {(\rcvtimec c x Q {})}$, in which  the broadcast of value $v_0$ along 
channel $c$ is ignored.
On the other hand, Rule \rulename{RcvIgn} cannot be applied 
to the configuration $\conf{\Gamma_0}{\rcvtimec c x P {}}$, 
since this station is waiting to receive a value on channel $c$; however we can 
derive the transition $\conf{\Gamma_0}{\rcvtimec c x P {}} \trans{\rcva c {v_0}} 
\arcv c x P$ using Rule \rulename{Rcv}. 

We can put together the three transitions above  using the rule 
(Sync), leading to the transition $\confC_0 \trans{\snda c v} W_1$, 
where $W_1 = \sigma^2 | \rcvtimec d x {\nil} {(\rcvtimec c x Q {})} | \arcv c x P$.
\end{example}

\begin{example}[Ignored Receptions]
\label{ex:rcvign}
Consider the configuration $\confC = \conf{\Gamma}{\bcastzeroc c v  \;|\; \rcvtimec c x P Q}$, 
where $\deltav = 1$ and $\Gamma$ is such that $\Gamma \vdash c:
\cbusy$, say $\Gamma \vdash_{\mathrm{t}} c : 1$. 
Using the rules introduced so far we can derive
\begin{align}
  \label{eq:1}
  \confC \trans{c!v} 
\conf{\Gamma}{\sigma \;|\; \rcvtimec c x P Q}
\end{align}
describing the unblocked sending of the value $v$ along the channel
$c$. This can be inferred using Rule \rulename{Sync} from
$\conf{\Gamma}{\bcastzeroc c v} \trans{\sndac c v} \sigma$, which can
be inferred using Rule \rulename{Snd}, and the judgement
$\conf{\Gamma}{\rcvtimec c x P Q} \trans{\rcva c v} \rcvtimec c x P
Q$.  This latter can be inferred using Rule \rulename{RcvIgn}, because
$\Gamma \vdash c: \cbusy$ means that $\isrcv{\conf{\Gamma}{\rcvtimec c
    x P Q}, c} = \text{false}$.

In the transition (\ref{eq:1}) above the receiver  $\rcvtimec c x P Q$
ignores the transmission of $v$ along $c$. One might have expected 
it to accept this value. However the channel is already exposed, $\Gamma \vdash c:
\cbusy$, and thus the receptor can not properly synchronise 
properly with the sender.  We will see later, in Example~\ref{ex:rcvfail},
that a transmission errors actually occurs. 
\end{example}

\begin{table}[!t]
\[
  \begin{array}{ll}
   \Txiom{TimeNil}{}
{\conf \Gamma \nil \trans{\sigma}  \nil}
&
\Txiom{Sleep}{}{\conf \Gamma {\delay \sigma P} 
\trans{\sigma}  P}
\\\\
\Txiom{ActRcv}{\Gamma \vdash_{\mathrm{t}} c: n,\, n > 1 }
{\conf{\Gamma}{\arcv c x  P} \trans{\sigma}
 {{\arcv c x  P}}}
&
\Txiom{EndRcv}{\Gamma \vdash_{\mathrm{t}} c: 1,\,\, \Gamma \vdash_{\mathrm{v}} c = w  }
{\conf \Gamma {\arcv c x  P} \trans{\sigma}
 {{{\subst w x}P}}}    
\\\\
\multicolumn{2}{c}{\Txiom{Timeout}{\Gamma \vdash c: \cfree }
{\conf \Gamma {\rcvtimec c x P Q} \trans{\sigma}  {Q}}}
% & &&
% \Txiom{Timerestr}{\conf{\Gamma[c\mapsto(n,v)]}{W} \trans{\sigma}{W'}}
% {\conf{\Gamma}{\crest{c}{(n,v)}W  \trans{\sigma} {\crest{c}{\tupd{\Gamma}(c)}.W'}}}
  \end{array}
\]
    \caption{Intensional semantics: timed transitions}
    \label{tab:net2}
  \end{table}   

The transitions for modelling 
the passage of time, $\conf \Gamma W \trans{\sigma} W'$,
are given in Table~\ref{tab:net2}. Rules (TimeNil) and (Sleep) are 
straightforward. In 
rules (ActRcv) and (EndRcv) we see that the active receiver $\arcv c x
P$ continues to wait for the transmitted value to make its way through
the network; when the allocated transmission time elapses the value is
then delivered and the receiver evolves to ${\subst w x}P$.  
%%The rule
%%(SumTime) is necessary to ensure \emph{time determinism} (see Proposition~\ref{prop:time-determinism}). 
 Finally, Rule 
(Timeout) implements the idea that ${\rcvtimec c x P Q}$ is a
time-guarded receptor; when time passes it evolves into the
alternative $Q$. However this only happens if the channel $c$ is not
exposed. What happens if it is exposed is explained in 
Table~\ref{tab:net3}.
%% Finally, Rule (TimePar) models how 
%%$\sigma$-actions are derived for collections of threads.

\begin{example}[Passage of Time]
\label{ex:intsem2}
Let $\confC_1 = \conf{\Gamma_1}{W_1}$, where $\Gamma_1(c) = (2,v_0)$,  $ \Gamma_1 
\vdash d: \cfree$ 
and $W_1 = \sigma^2 | \rcvtimec d x {\nil} {\rcvtimec c x Q {}} | \arcv c x P$ is the system term 
derived in Example \ref{ex:intsem1}.
We show how a $\sigma$-action can be derived for this configuration. 
First note that $\conf{\Gamma_1}{\sigma^2} \trans{\sigma} 
\sigma$; this transition can be derived using Rule \rulename{Sleep}. 
Since $d$ is idle in $\Gamma_1$, we can apply Rule \rulename{TimeOut} 
to infer the transition $\conf{\Gamma_1}{\rcvtimec d x {\nil}{(\rcvtimec c x {Q}{})}} 
\trans{\sigma} \rcvtimec c x {Q}{}$; time passed before a value could 
be broadcast along channel $d$, causing a timeout in the station waiting 
to receive a value along $d$.
Finally, since $\Gamma_1 \vdash_{\mathrm{v}} c: 2$, we can use Rule 
\rulename{ActRcv} to derive $\conf{\Gamma_1}{\arcv c x P} 
\trans{\sigma} \arcv c x P$.
 
At this point we can use twice Rule \rulename{TimePar} (which 
is given in Table \ref{tab:struct})  to 
infer a $\sigma$-action performed by $\confC_1$. This leads 
to the transition $\confC_1 \trans{\sigma} W_2$, where 
$W_2 = \sigma | \rcvtimec c x Q {} | \arcv c x P$.
\end{example}
 \begin{table}[!t]
\begin{align*}
&\Txiom{RcvLate}{
\Gamma \vdash c:\cbusy }
          {
\conf \Gamma {\rcvtimec c x P Q}
         \trans{\tau} {\arcv c x {{\subst {\err} x}P}}
         }
&&
\Txiom{Tau}{}{\conf \Gamma {\tau.P} \trans{\tau}  P}
\\\\
&\Txiom{Then}
{\interpr{b}_\Gamma = \mbox{true}}
{ \conf {\Gamma}{\matchb b {P} {Q}} \trans{\tau} 
   {\sigma.P}} 
&&
\Txiom{Else}
{\interpr{b}_\Gamma = \mbox{false}}
{ \conf {\Gamma}{\matchb b {P} {Q}} \trans{\tau} 
{\sigma.Q}}
\end{align*}
\caption{Intensional semantics: - internal activity}
\label{tab:net3}
\end{table}

Table \ref{tab:net3} is devoted to internal transitions $\conf \Gamma W \trans{\tau} W'$. 
\MHc{Let us first  explain} rule (RcvLate). Intuitively the process
$\rcvtimec c x P Q$ is ready to start receiving a value on channel $c$. 
However if $c$ is exposed this means that a transmission is already taking
 place. Since the process has therefore missed the start of the transmission 
it will receive an  error value.
\MHc{Thus the rule (RcvLate) reflects the fact that in wireless systems a collision takes place
if there is a misalignment between the transmission and reception of a message.}
The remaining rules are straightforward. 
Note that in the matching construct 
we use a channel environment dependent
evaluation function for Boolean expressions $\interpr{b}_\Gamma$ (note that this has not to 
be confused with the function $\interpr{\cdot}$, used to evaluate closed expressions), 
\MHc{\MHf{Andrea: please check}because of the presence of the exposure predicate $\expsd{c}$
in the Boolean language. Formally we have that $\interpr{e_1 = e_2}_{\Gamma} = \ttrue$ evaluates to 
true if and only if $\interpr{e_1} = \interpr{e_2}$, and $\interpr{\expsd c}_{\Gamma} = \ttrue$ if 
and only if $\Gamma \vdash c: \cbusy$. 
We remark that checking for the exposure of a channel 
amounts to listening on the channel for a value. But in wireless systems
it is not possible to both listen and transmit within the same time unit, as
communication is half-duplex, \cite{rappaport}.
As a consequence in our intensional 
semantics, in the rules (Then) and (Else),  the execution of both branches is 
delayed of one time unit.}

\begin{example}
\label{ex:intsem3}
Let $\Gamma_2$ be a channel environment such that $\Gamma_2(c) = (1,v)$, 
and consider the configuration $\confC_2 = \conf{\Gamma_2}{W_2}$, 
where $W_2 = \sigma \, | \, \rcvtimec c x Q {} \, | \, \arcv c x P$ 
has been defined in Example \ref{ex:intsem2}. 

Note that this configuration contains both a receiver process and 
an active receiver 
along the exposed channel $c$. We can think of the receiver 
$\rcvtimec c x Q {}$
as a process which missed the synchronisation with a 
broadcast which has been previously performed along channel $c$; 
as a consequence this process is doomed to receive an error value. 

This situation is modelled by Rule \rulename{RcvLate}, 
which allows us to infer the transition $\conf{\Gamma_2}{\rcvtimec c x Q}{} 
\trans{\tau} \arcv c x {\{\err/x \}Q}$. 
As we will see, Rule 
\rulename{TauPar}  which we introduce in Table \ref{tab:struct}, ensures that 
$\tau$-actions are contextual. 
This means that the transition derived 
above allows us to infer the transition $\confC_2 \trans{\tau} 
W_3$, where $W_3 = \sigma \, | \, \arcv c x {\{\err/x\} Q} \, | \, \arcv c x P$.
\end{example}

\begin{example}[On rules \rulename{RcvIgn} and \rulename{RcvLate}]
\label{ex:rcvfail}
Consider again the configuration $\confC$ of Example \ref{ex:rcvign}. 
Recall that $\confC = \conf{\Gamma}{\bcastzeroc c v} | \rcvtimec c x P Q$, where 
$\Gamma \vdash_{\mathrm{v}} c : 1$ and $\delta_v = 1$. In Example \ref{ex:rcvign} we have shown 
that $\confC \trans{\sndac c v} \sigma | \rcvtimec c x P Q$, where the proof of the 
transition contains an application of Rule \rulename{RcvIgn}. 
This transition represents the unblocked transmission of the value 
$v$ along the channel $c$, which also changes the channel environment
from $\Gamma$ to $\gupd{\sndac c v}{\Gamma}$. Now consider the 
resulting configuration 
$\confC' = \conf{ \gupd{\sndac c v}{\Gamma} }{\sigma | \rcvtimec c x P
  Q}$.
As $\gupd{\sndac c v}{\Gamma} \vdash c:\cbusy$ we can use 
Rule \rulename{RcvLate}\footnote{An application of 
Rule \rulename{TauPar} from  Table \ref{tab:struct} is also required.},  to infer the transition 
$\confC' \trans{\tau} \sigma | \arcv c x {\{\err/x\}P}$,
modelling the expected error in transmission along channel $c$ 
due to a collision.

% As we already noticed, this behaviour of $\confC$ is somewhat counter-intuitive; 
% at least intuitively the broadcast of value $v$ along channel $c$ 
% should cause the receiving component $\rcvtimec c x P Q$ to start the reception of 
% the message.
% However, consider the configuration $\confC' = \conf{\Gamma'}{\sigma | \rcvtimec c x P Q}$, where 
% $\Gamma' = \gupd{\sndac c v}{\Gamma}$.  As $\Gamma' \vdash c: \cbusy$, 
% we can infer the transition $\confC' \trans{\tau} \sigma | \arcv c x {\{\err/x\}P}$, 
% using Rule \rulename{RcvLate}\footnote{The proof of this transition also requires an application of 
% Rule \rulename{TauPar}, which is presented in Table \ref{tab:struct}.}, 
% to model the incorrect reception at channel $c$. 

Note also that we could have applied Rule \rulename{RcvLate} directly to the initial configuration 
$\confC = \conf{\Gamma}{\bcastzeroc c v} | \rcvtimec c x P Q$,
leading to the  transition $\confC \trans{\tau} 
c!\langle v \rangle  | \arcv c x {\{\err/x\}P}$, again reflecting an
error in transmission along the channel $c$ due to the fact that it is
already exposed. In fact we have the transition 
$\conf{\Gamma}{W} | \rcvtimec c x P Q \trans{\tau} 
W  | \arcv c x {\{\err/x\}P}$, regardless of the form of $W$.
This emphasises the fact that the inability of the receiver to receive
correctly the value being transmitted is because the channel is
already exposed and not because another station is 
willing to broadcast along it. 
\end{example}

\begin{remark}
The previous example together with 
Example \ref{ex:rcvign} shows that there is a delicate interplay
between the rules \rulename{RcvIgn} and \rulename{RcvLate},
particularly when modelling the effect of an external broadcast on 
receivers in the presence of exposed channels. The overall goal of our
intensional semantics is to ensure that it has certain natural
properties, such as \emph{input-enabledness}. This ensures that for 
any configuration $\conf{\Gamma}{W}$ and any $c?v$ there exists some transition
$\conf{\Gamma}{W}    \trans{c?v} W'$. Here $W'$ records the effect of an external
broadcast of $v$ along $c$ has on the configuration; if the
broadcast is actually ignored by all stations in the configuration
then $W'$ will coincide with $W$. \emph{Input-enabledness} also helps
us in ensuring that broadcasts are independent of their environment. 
For example, we require the configuration $(\conf{\Gamma}{\bcastzeroc
  c v | W})$
to be able to perform the broadcast of value $v$ along channel $c$, regardless of 
the structure of $W$, even if $c$ is exposed in $\Gamma$.
Such a transition can only be inferred from Rule \rulename{Sync} 
if we we match the output action along channel 
$c$ performed by the configuration $\conf{\Gamma}{\bcastzeroc c v}$ with an input 
action  performed by $\conf{\Gamma}{W}$. \emph{Input-enabledness} will
ensure that the latter input action is always possible.  

In Section \ref{sec:intsem.properties} we will 
show that our intensional semantics in fact satisfies a number of 
natural properties, including \emph{input-enabledness}; see Lemma
\ref{lem:rcv-enabling}. 
This would obviously be not true if, by omitting Rule
\rulename{Rcvlgn}, we were to forbid inputs over
exposed channels.
\end{remark}

\begin{table}[t]
%%  \begin{align*}
%%& 
\[
\begin{array}{ll}
\Txiom{TimePar}
{\conf {\Gamma} W_1 \trans{\sigma}  {W'_1} \Q 
\conf {\Gamma} W_2 \trans{\sigma}  W'_2}
{\conf {\Gamma} {W_1|W_2} \trans{\sigma}  {W'_1 | W'_2}
}
&
\Txiom{TauPar}{\conf \Gamma {W_1} \trans{\tau}  {W'_1}}
{\conf \Gamma {W_1|W_2} \trans{\tau} {W'_1|W_2}} 
\\\\
%&
\Txiom{Rec}
{   \conf{\Gamma}{ \{ \fix X P/X\}P} \trans{\lambda} W }
{\conf {\Gamma}{\fix X P}      \trans{\lambda} W}
& %&
\Txiom{Sum}{\conf \Gamma {P \trans{\lambda}  W} 
\q\: \lambda \in\{ 
\tau,  c!v\} 
}
{\conf \Gamma {P+Q} \trans{\lambda} W}
\\\\
\Txiom{SumTime}
{\conf{\Gamma}{P} \trans{\sigma}{P'} \Q \conf{\Gamma}{Q} \trans{\sigma} {Q'}}
{\conf{\Gamma}{P+Q}\trans{\sigma} P' + Q'}
& 
\Txiom{SumRcv}
{\conf \Gamma P \trans{\rcva c v} W \Q \isrcv{\conf{\Gamma}{P},c }}
 {\conf \Gamma {P+Q} \trans{\rcva c v} {W}}
\\\\
%&
\Txiom{ResI}{\conf{\Gamma[c\mapsto(n,v)]}{W} \trans{c!v} {W'} }
{\conf \Gamma {\crest{c}{(n,v)}.W} \trans{\tau} 
{\crest{c}{\gupd{c!v}{\Gamma}(c)}.W'}}  
&%%&
\Txiom{ResV}{\conf{\Gamma[c\mapsto(n,v)]}{W} \trans{\lambda} {W'},\,\, c \not\in \lambda}
{\conf \Gamma {\crest{c}{(n,v)}.W} \trans{\lambda} 
{\crest{c}{(n,v)}.W}}  
%%  \end{align*}
\end{array}
\]
\caption{Intensional semantics: - structural rules}
\label{tab:struct}
\end{table}

The final set of rules, in Table~\ref{tab:struct}, are structural. 
Rule (TimePar) models how 
$\sigma$-actions are derived for collections of threads. Rules
(TauPar), (Rec) and (Sum) are standard. 
Rule
(SumTime) is necessary to ensure \emph{time determinism} (see Proposition~\ref{prop:time-determinism}). 
Rule (SumRcv) guaranteed that only effective
receptions can decide in a choice process. 
 Finally 
 Rules (ResI) and (ResV) show how restricted channels are
handled. Intuitively moves from the configuration $\conf \Gamma
{\crest{c}{(n,v)}.W}$ are inherited from the configuration
$\conf{\Gamma[c\mapsto(n,v)]}{W} $; here the channel environment 
$\Gamma[c\mapsto(n,v)]$ is the same as $\Gamma$ except that $c$ has
associated with it (temporarily) the information $(n,v)$. However if
this move mentions the restricted channel $c$ then the inherited move
is rendered as an internal action $\tau$, (ResI). Moreover the
information associated with the restricted channel in the residual is
updated, using the function $\gupd{\sndac c v}{\cdot}$ previously defined. 
Rules (TauPar), (Sum) and  (SumRcv)  have their 
symmetric counterparts. 

%\subsection{Properties of the Intensional Semantics}
\label{sec:intsem.properties}
%%In this section we illustrate some of the main properties enjoyed by 
%%the intensional semantics illustrated in Section \ref{sec:intsem}. 
%%The contents of this section are purely technical and 
%%needed only for the proofs of the results illustrated later in the 
%%paper. Therefore, this section may be safely skipped by the  reader 
%not interested in details.

In the remainder of this section we illustrate some of the main properties enjoyed by 
the intensional semantics illustrated in Section \ref{sec:intsem}. 
The contents of this part are purely technical and 
needed only for the proofs of the results illustrated later in the 
paper: they may be safely skipped by the  reader 
not interested in details.

In broadcast process calculi transmission of a value is usually modelled 
as a non-blocking action \cite{CBS,MerroBS11,Cerone_Hennessy2013}, 
meaning that all configurations should always be able to receive
  an arbitrary value along an arbitrary channel. This is a derived
  property of our calculus: 
\begin{lemma}[Input enabledness]
\label{lem:rcv-enabling}
Let $\conf \Gamma W$ be a configuration. Then for any channel $c$ and value 
$v$ we have that 
$\conf{\Gamma}{W} \trans{\rcva c v} W'$ for some $W'$; further 
\begin{enumerate}
\item 
\label{rcv-enabling1}
$\neg\isrcv{\conf{\Gamma}{W},c}$ implies $W'=W$ %% channel busy or no receivers along channel c
%\item 
%\label{rcv-enabling3}
%$\Gamma \vdash c : \cbusy$ implies 
%$W'=W$. 
\item 
\label{rcv-enabling2}
$\isrcv{\conf{\Gamma}{W},c}$ implies $W' \neq W$, and %%channel free and receiver along channel c
\MHc{\MHf{Andrea: please check}for every value $w$,  $\conf{\Gamma}{W} \trans{c?w} W'$.  } 
\end{enumerate}
\end{lemma}
\begin{proof}
See the Appendix, Page \pageref{proof:rcv-enabling}. 
\end{proof}

Our model of  time also conforms to a well-established
approach in the literature; see for example 
\cite{Sifakis94,Yi91}: 
  \begin{proposition}[Time Determinism]
\label{prop:time-determinism}
    Suppose $\confC \trans{\sigma} W_1$ and  $\confC \trans{\sigma} W_2$.
          Then $W_1 = W_2$. 

\end{proposition}

\begin{proof}
By induction on the proof of the transition $\confC \trans{\sigma} W_1$. 
See the Appendix, Page \pageref{proof:time-determinism} for details.
\end{proof}

\begin{proposition}[Maximal Progress]
\label{prop:maximal-progress}
Suppose $\confC \trans{\sigma} W_1$. If $\lambda \in \{ \tau, \, c!v \}$, for some $c$ and $v$, then there is no
$W_2$ such that $\confC \trans{\lambda} W_2$. 
\end{proposition}

\begin{proof}
    By induction on the proof of the derivation $\confC \trans{\sigma} W_1$. 
    See the Appendix, Page \pageref{proof:maximal-progress} for details. 
  \end{proof}

Another important property concerns the exposure state of channel 
environments. This property states that non-timed transitions are identified up-to 
channel environments which share the same set of idle channels.

\begin{proposition}[Exposure Consistency]
\label{prop:exposure-consistency}
Let $\Gamma_1,\Gamma_2$ be two channel environments such that
$\Gamma_1 \vdash c: \cbusy$ if and only if $\Gamma_2 \vdash c:
\cbusy$ for every channel $c$. 
Then for any system term $W$ and action $\lambda \neq \sigma$, 
$\conf{\Gamma_1}{W} \trans{\lambda} W'$ implies $\conf{\Gamma_2}{W} 
\trans{\lambda} W'$.
\end{proposition}

\begin{proof}
By Induction on the proof of the derivation $\conf{\Gamma_1}{W} \trans{\lambda} W'$. 
See the Appendix, Page \pageref{proof:exposure-consistency} for details.
\end{proof}

We end our discussion on the intensional semantics \MHc{with a technical result on the 
interaction between stations in systems; this will be useful in later developments.}
\begin{proposition}[Parallel components]
\label{prop:parallel-components}
Let  $\conf \Gamma {W_1| W_2}$ be a configuration. 
\label{prop:par-comp}
\begin{enumerate}
\item $\conf \Gamma {W_1 | W_2} \trans{\tau} W$ if and only if
\begin{itemize}
\item either there is $W'_1$ such that $\conf {\Gamma}{W_1} \trans{\tau} 
{W'_1}$ with $W=W'_1 | W_2$
\item or there is $W'_2$ such that $\conf {\Gamma}{W_2} \trans{\tau} 
{W'_2}$ with $W=W_1 | W'_2$. 
\end{itemize} 
\item 
\label{par2}
$\conf \Gamma {W_1 | W_2} \trans{\rcva c v} W$ if and only if there are $W'_1$ and 
$W'_2$ such that $\conf \Gamma {W_1} \trans{\rcva c v}  {W'_1}$, 
$\conf \Gamma {W_2} \trans{\rcva c v}  {W'_2}$ and $W = W'_1 | W'_2$. 
\item 
\label{par3}
$\conf \Gamma {W_1 | W_2} \trans{c!v} W$ if and only if there are $W'_1$ and 
$W'_2$ such that 
\begin{itemize}
\item $\conf \Gamma {W_1} \trans{c!v} {W'_1}$, 
$\conf \Gamma {W_2} \trans{\rcva c v} {W'_2}$ and $W = W'_1 | W'_2$
\item 
or $\conf \Gamma {W_1} \trans{\rcva c v}  {W'_1}$, 
$\conf \Gamma {W_2} \trans{c!v} {W'_2}$ and $W = W'_1 | W'_2$. 
\end{itemize}
\item 
\label{prop:sigma}
$\conf \Gamma {W_1 | W_2} \trans{\sigma} W$ if and only if there are $W'_1$ and 
$W'_2$ such that 
$\conf \Gamma {W_1} \trans{\sigma}  {W'_1}$, 
$\conf \Gamma {W_2} \trans{\sigma}  {W'_2}$ and $W = W'_1 | W'_2$.
\hfill\qed
\end{enumerate} 
\end{proposition}

\begin{proof}
Details for (3) are given in the Appendix; see Page \pageref{proof:parallel-components}. The other three statements can 
be proved similarly.
\end{proof}

\subsection{Reduction semantics}
\label{sec:redsem}
We are now in a position to formally define the individual computation steps
for wireless systems, alluded to informally in (\ref{eq:steps}) above. 

\begin{definition}[Reduction]\label{def:step}
  We write 
$\conf{\Gamma}{W} \red \conf{\Gamma'}{W'}$ if
\begin{enumerate}[label=(\roman*)] 
\item (Transmission) $\conf{\Gamma}{W} \trans{c!v} W'$ for some channel $c$ and value $v$,
where  $\Gamma' = \gupd{c!v}{\Gamma}$

\item (Time) $\conf{\Gamma}{W} \trans{\sigma} W'$ and $\Gamma' = \tupd{\Gamma}$

\item (Internal)  $\conf{\Gamma}{W} \trans{\tau} W'$ and $\Gamma' = \gupd{\tau}{\Gamma}$. 
\end{enumerate}
The intuition here should be obvious; computation proceeds either by the transmission of values 
between  stations,
the passage of time, or internal activity; further, 
the exposure state of channels is updated according to 
the performed transition.
\end{definition}
Sometimes it will be useful to distinguish between 
instantaneous reductions and timed reductions; instantaneous 
reductions, $\conf{\Gamma_1}{W_1} \red_i \conf{\Gamma_2}{W_2}$, 
are those derived via clauses \MHc{(i) or  (iii) above;}  timed 
reductions are denoted with the symbol $\red_{\sigma}$ and 
coincide with reductions derived using clause \MHc{(ii)}.
We use the notation $\conf{\Gamma}{W} \red_i$ ($\conf{\Gamma}{W} \red_{\sigma}$) 
if there exists $\conf{\Gamma'}{W'}$ such that $\conf{\Gamma}{W} \red_{i} \conf{\Gamma'}{W'}$ 
($\conf{\Gamma}{W} \red_{\sigma} \conf{\Gamma'}{W'}$), and $\conf{\Gamma}{W} \not\red_i$ 
($\conf{\Gamma}{W} \not\red_{\sigma}$) to stress that there is no configuration $\conf{\Gamma'}{W'}$ 
such that $\conf{\Gamma}{W} \red_{i} \conf{\Gamma'}{W'}$ ($\conf{\Gamma}{W} \red_{i} \conf{\Gamma'}{W'}$).

\begin{example}
\label{ex:redsem1}
We  show how the transitions we have inferred 
in the Examples~\ref{ex:intsem1}, \ref{ex:intsem2} and \ref{ex:intsem3} 
can be combined together to derive a computation fragment 
for the configuration $\confC_0$ considered in Example 
\ref{ex:intsem1}. 

Let $\confC_i = \conf{ \Gamma_i}{ W_i}$, $i\in 0,{..},2$, be as defined in 
the examples mentioned above. Note that $\Gamma_1 = \gupd{\sndac c {v_0}}{\Gamma_0}$ and 
 $\Gamma_2 = \gupd{\sigma}{\Gamma_1}$. 
We have already shown that $\confC_0 \trans{\sndac c {v_0}} W_1$; 
this transition, together with the equality $\Gamma_1 = \gupd{\sndac c {v_0}}{\Gamma_0}$, 
can be used to infer the reduction $\confC_0 \red_i \confC_1$. 
A similar argument shows that $\confC_1 \red_{\sigma} \confC_2$.
Finally,  if we let $\confC_3$ denote $\conf{\Gamma_2}{W_3}$ we also have
$\confC_2 \red_{i} \confC_3$ since $\confC_2 \trans{\tau} W_3$ and 
$\Gamma_2 = \gupd{\tau}{\Gamma_2}$. 
\end{example}

\begin{example}[Time-consuming transmission]
Consider a wireless system with two stations, that is  
a configuration $\confC_1$ of the form $ \conf {\Gamma_1} {P_1 | Q_1} $. Let us suppose
  \begin{align*}
    &P_1 \,\text{ is }\, \bcastc c w R,   \qquad Q_1 \,\text{ is }\, \rcvtimec c x {S} T_1
  \end{align*}
where $\Gamma_1$ is a stable channel environment and $\delta_w=2$. Then 
\begin{align}\label{eq:step1}
  &{\confC_1} \red \confC_2
\end{align}
where $\confC_2$ has the form $ \conf {\Gamma_2} {P_2 | Q_2}$ and 
\begin{align*}
  &P_2 \,\text{ is }\, \asnd \sigma 2 R  &&Q_2 \,\,\text{ is }\, \arcv c x S  && \Gamma_2 \vdash_{\mathrm{t}} c:2
  &&\Gamma_2 \vdash_{\mathrm{v}} c: w 
\end{align*}
The move from $P_1$ to $P_2$ is via an application of the rule (Snd), from $Q_1$ to $Q_2$ relies on
(Rcv) and they are combined together using (Sync) to obtain
\begin{math}
  \conf {\Gamma_1} {P_1 | Q_1} \trans{\sndac c w} P_2 | Q_2.
\end{math}
The final step (\ref{eq:step1}) results from (Transmission) in Definition~\ref{def:step}.  

The next step $\confC_2 \red \confC_3 = \conf{\Gamma_3}{\asnd \sigma{} R | Q_2}$ is via (Time) in Definition~\ref{def:step}; 
here the only change to the channel environment is
that $\Gamma_3 \vdash_{\mathrm{t}} c:1$. The inference of the transition 
\begin{align*}
 &\conf {\Gamma_2}  {P_2 | Q_2} \trans{\sigma} \asnd \sigma {} R  \,|\, Q_2
\end{align*}
uses the rules (Sleep),  (ActRcv) and (TimePar). 

The final move we consider, $\confC_3 \red \confC_4 = \conf {\Gamma}
{R | \subst w x {S}}$, is another instance of (Time). However here the
delay action is inferred using (Sleep), (EndRcv) and (TimePar). Thus
in three reduction steps the value $w$ has been transmitted from the
first station to the second one along the channel $c$, in two units of
time.

Now suppose we change $P_1$ to $P'_1 = \sigma.P_1$, obtaining thus the
configuration $\confC_1' = \conf {\Gamma_1} {P'_1 | Q_1}$.  Then the
first step, $\confC'_1 \red \confC'_2$ is a (Time) step, with
$\confC'_2 = \conf{\Gamma_1}{P_1 | T_1}$. Here an instance of the rule
(Timeout) is used in the transition from $Q_1$ to $T_1$. In
$\confC'_2$ the station $P_1$ is now ready to transmit on channel $c$,
but the second station has stopped listening. The next step 
depends on the exact form of $T_1$; if for example 
$\isrcv{T_1,c}$ is false then by an application of rule (RcvIgn) we can derive
$\confC'_2 \red \confC'_3 = \conf {\Gamma_2} {P_2 | T_1}$. Here the transmission of
$w$ along $c$  started but nobody was listening. 

Finally, suppose $T_1$ is a delayed listener on channel $c$, say $\delay \sigma T_2$ where $T_2$ is
$ \rcvtimec c y {S_2} U_2$. Then we have the (Time) step 
$\confC'_3 \red \confC'_4 = \conf {\Gamma_3}  {\asnd \sigma {} R  \,|\, T_2}$ and now the second
station, $T_2$,  is ready to listen. However,  as $\Gamma_3 \vdash c: \cbusy$, 
station $T_2$ is joining the
transmission too late. 
\MHc{To reflect this we can derive}
the (Internal) step 
\begin{align*}
\confC'_4 \red \confC'_5 = \conf {\Gamma_3} {\asnd \sigma {} R  \,|\, \arcv c y   {{\subst \err y}\MHc{S_2}}} 
\end{align*}
\MHc{using the rules (RcvLate) and
(TauPar), among others.} 
At the end of the transmission, in one more time step, the second station will therefore 
end up with an error in reception. 

In the revised system $\confC'_1 = \conf {\Gamma_1} {\sigma.{P_1'} | Q_1}
$ the second station missed the delayed transmission from $P_1'$. 
\MHc{However we}
can change the code at the second station to accommodate this
delay, \MHc{by replacing $Q_1$ with the persistent listener} $Q'_1 = \rcvc c x S$.
 We leave the reader
to check that starting from the  configuration $\conf
{\Gamma_1} {\sigma.{P_1'} | Q'_1}$ the value $w$ will be successfully
transmitted between the stations in four reduction steps.
\end{example}

\begin{example}[Collisions]
\label{ex:collision}
Let us now consider again the configuration $\confC_1 
= \conf{\Gamma}{S_1 | S_2 | R_1}$ of 
Example \ref{ex:collision.syntax}.
In this configuration the station $S_1$ can perform a 
broadcast, leading to the reduction 
$\confC_1 \red \confC_2 = \conf{\Gamma_1}{\sigma^2 | S_2 
| \arcv c x P}$, the derivation of which 
requires an instance of the rule \rulename{RcvIgn}, 
$\conf{\Gamma}{S_1} \trans{\rcva c {v_1}} \MHc{S_1}$; 
here the channel environment $\Gamma_1$ is defined as 
$\gupd{\sndac c {v_0}}{\Gamma}$, leading to 
$\Gamma_1(c) = (2,v_0)$. 
We can now derive the reduction 
$\confC_2 \red \confC_3 = \conf{\Gamma_2}{ 
\sigma | \bcastzeroc c {v_1} | \arcv c x P}$, where 
$\Gamma_2 = \gupd{\sigma}{\Gamma_1}$ meaning 
that $\Gamma_2 \vdash_{\mathrm{t}} c : 1$. 

In this configuration the second station is ready to 
 broadcast value $v_1$ along channel $c$. 
Since there 
is already a value being transmitted along this channel, 
we expect this second broadcast to cause a \MHc{collision; further}, 
since the amount of time required for transmitting value $v_1$ 
is equal to the time needed to end the transmission of 
value $v_0$, we expect that the broadcast performed 
by the first station does not affect the amount of 
time for which \MHc{the channel $c$ is exposed.}
%%%See Figure \ref{fig:collisions}.
 
\MHc{Formally this is reflected in the reduction} 
$\confC_3 
\red \confC_3' = \conf{\Gamma_2'}{\sigma | \sigma | \arcv c x P}$.
Here the reduction of the system term uses the  sub-inferences 
$\conf{\Gamma_2}{\sigma} \trans{\rcva c {v_1}} \sigma$, 
$\conf{\Gamma_2}{\bcastzeroc c {v_1}} \trans{\sndac c {v_1}} \sigma$ and 
$\conf{\Gamma_2}{\arcv c x P} \trans{\rcva c {v_1}} 
\arcv c x P$; the first and the third of these transitions can be derived using Rule  
\rulename{RcvIgn}, while the second one can be derived using Rule \rulename{Bcast}.
Consequently $\Gamma_2' = 
\gupd{\sndac c {v_1}}{\Gamma_2}$, and  since 
$\Gamma_2 \vdash c: \cbusy$ we obtain 
$\Gamma_2'(c) = (1,\err)$; this represents the fact that
a collision 
has occurred, and thus the special  value $\err$
will eventually be delivered on $c$.

At this point we can derive the reductions 
$\confC_3' \red_\sigma \confC_4 = \conf{\Gamma}{\nil | \nil | \{\err/x\}P}$, 
meaning that the transmission along channel $c$ terminates in one time instant, 
leading the receiving station to detect a collision. 
The reduction above can be obtained from the transitions 
$\conf{\Gamma_2'}{\sigma} \trans{\sigma} \nil$ and 
$\conf{\Gamma_2'}{\arcv c x P} \trans{\sigma} \{\err/x\}P$, obtained 
via rules (TimeNil) and (EndRcv) presented in Table \ref{tab:net2}. 

Now, suppose  we change the amount of time required 
to transmit value $v_1$ from $1$ to $2$, and consider 
again the configuration $\confC_3$ above. 
In this case the transmission of value $v_1$ will 
also cause a collision; however, in this case 
the transmission of value $v_1$ is long enough 
to continue after that of value $v_0$ has finished; 
as a consequence, we expect that the time 
required for channel $c$ to be released rises 
when the broadcast of $v_1$ happens. 
%%%; see Figure \ref{fig:collisions}. 

In fact, in this case we have the reduction 
$\confC_3 \red \confC_3'' = \conf{\Gamma_2''}{\sigma | \sigma^2 | 
\arcv c x P}$, where $\Gamma_2'' = \gupd{\sndac c {v_1}}{\Gamma_2}$ 
and specifically $\Gamma_2''(c) = (2, \err)$. Now, two time 
instants 
are needed for the transmission along channel $c$ to end, 
leading to the sequence of (timed) reductions $\confC_3'' 
\red_\sigma \red_\sigma \confC_4.$
\end{example}

\subsection{Behavioural Equivalence}
\label{sec:cxt}
In this section we propose a notion of timed behavioural equivalence
for our wireless networks. Our touchstone system equality is 
\emph{reduction 
barbed congruence} \cite{hy,pibook,MiSa92,JeffreyRathke05hopi}, 
a standard contextually defined
\MHc{process}  equivalence. Intuitively, two terms are reduction barbed congruent  
if they have the same \emph{basic observables}, in all \MHc{parallel} contexts, 
under all possible \emph{computations}.
The \MHc{formal definition relies on two crucial 
concepts,  a reduction semantics to describe how  systems
evolve, which we have already defined, and  a notion of basic observable which says what
the environment can observe directly of  a system.
There is some choice as to what to take as a basic observation, or \emph{barb}, of a wireless system.
In standard process calculi this is usually taken to be the ability of the
environment to receive a value along a channel. But the series of 
examples we have just seen demonstrates that this is problematic, in the presence of possible collisions
and the passage of time. Instead we choose a more appropriate notion for wireless systems,
one which is already present in our language for station code}: 
\emph{channel exposure\/}.
\begin{definition}[Barbs]
\label{def:observability}  
\MHc{
We say the configuration $\conf \Gamma W$  has a \emph{strong barb on $c$},  
written $\conf \Gamma W \downarrow_{c}$, if 
$\Gamma \vdash c:\cbusy$. 
We write $\conf \Gamma W \Downarrow_{c}$, a \emph{weak barb},  if there exists a configuration 
$\confC'$ such that $\conf \Gamma W \red^* \confC'$ 
and $\confC' \downarrow_{c}$. Note that we allow the passage of time in the definition of weak barb.
}
\end{definition}

\begin{definition}
  Let $\RR$ be a relation over configurations. 
  \begin{enumerate}%[(1)]
  \item $\RR$  is said to be \emph{barb preserving} if 
$\conf {\Gamma_1} {W_1} \MHc{\Downarrow_{c}}$ 
implies $\conf {\Gamma_2} {W_2} \Downarrow_c$, 
whenever 
$(\conf {\Gamma_1} {W_1}) \RRr (\conf {\Gamma_2} W_2)$. 

\item It is \emph{reduction-closed} if  
$(\conf {\Gamma_1} {W_1}) \RRr (\conf {\Gamma_2}{W_2})$ and 
%%\begin{itemize}
%%\item 
$\conf {\Gamma_1} {W_1} \red \conf {\Gamma_1'} {W'_1}$ imply there is some
$\conf {\Gamma'_2}{W'_2}$ such
that $\conf {\Gamma_2} {W_2} \red^* \conf {\Gamma'_2} {W'_2}$ and  
$(\conf {\Gamma'_1} {W'_1}) \RRr (\conf {\Gamma'_2}{W'_2})$. 

\item It is \emph{contextual} if $\conf {\Gamma_1} {W_1} \, \RRr \,
  \conf {\Gamma_2}{W_2}$, implies $ \conf {\Gamma_1} {(W_1 | W)} \;
  \RRr \; \conf{\Gamma_2}{(W_2 | W)} $ for all processes \MHc{$W$.}
%  such
%   that both $\conf {\Gamma_1} {W_1 | W}$ and $ \conf {\Gamma_2}{W_2 |
%     W }$ are well-formed. 
  \end{enumerate}
\end{definition}
With these concepts we now have everything in place for a standard
definition of contextual equivalence between systems:
\begin{definition}%%[Reduction barbed congruence] 
\label{def:rbc}
\MHc{[Reduction barbed congruence]}, 
written $\simeq$, is the largest symmetric  relation over 
\MHc{ configurations}  which is barb preserving,
reduction-closed  and contextual. 
\end{definition}

In the remainder of this section we explore via examples the
implications of Definition~\ref{def:rbc}. The notion of a fresh channel will be
important; we say that $c$ is \emph{fresh} for the configuration
$\conf \Gamma W$ if it does not occur free in $W$ and $\Gamma \vdash
c:\cfree$. Note that \MHc{we can always pick a fresh channel for an arbitrary
configuration.}

\begin{example}\label{ex:transmit}
  Let us assume that $\Gamma \vdash c: \cfree$. Then it is easy to see that
\begin{align}\label{eq:barb.ex1}
\conf \Gamma \bcastc c {v_0} P  \,\not\simeq\, \conf \Gamma \bcastc c {v_1} P
\end{align}
under the assumption that $v_0$ and $v_1$ are different values. For let $T$ be the testing
context 
\begin{displaymath}
  \rcvtimec c x   {\matchb {x=v_0}{\bcastzeroc   {\eureka} \arb } \nil} {}
\end{displaymath} 
where \eureka is fresh, and $\arb$ is some arbitrary value. 
Then $\conf \Gamma {\bcastc c {v_0} P | T}$ has a weak barb on \eureka which is not the case for 
 $\conf \Gamma {\bcastc c {v_1} P | T}$. Since $\simeq$ is contextual and barb preserving, the statement 
(\ref{eq:barb.ex1}) above follows. 

However such tests will not distinguish between $\conf{\Gamma}{Q_1}$ and 
$\conf{\Gamma}{Q_2}$, where 
\begin{align*}
  & Q_1 =   \bcastzeroc c {v_0}  |  \bcastc c {v_1} P    
&& \text{and}\;&Q_2  =  \bcastzeroc c {v_1} |\bcastc c {v_0} P 
\end{align*}
assuming that $\delta_{v_0} = \delta_{v_1}$.
In both configurations $\conf \Gamma Q_1$ and $\conf \Gamma Q_2$ a
collision will occur at channel $c$ and a receiving station, such as $T$, 
will receive the error value $\err$ at the end of the transmission. So there is reason to
hope that
\begin{math}
  \conf \Gamma Q_1  \;\simeq\,\conf \Gamma Q_2.
\end{math}
However we must wait for for the proof techniques of the next section to 
establish this equivalence; \MHc{see Example~\ref{ex:transmit.bisim}.} 
\end{example}

The above example suggests that transmitted values can be observed only at the 
end of a transmission; \MHc{so} if a collision happens, there is 
no possibility \MHc{of determining} the value that was originally broadcast. 
This concept is stressed even more in the following example.

\begin{example}[Equating values]
\label{ex:equators}
Let $\Gamma$ be a stable channel environment, 
$W_0 = \bcastzeroc c {v_0}, W_1 = \bcastzeroc c {v_1}$ and consider 
the configurations $\conf{\Gamma}{W_0}, 
\conf{\Gamma}{W_1}$; here we assume that 
$v_0$ and $v_1$ are two different values with 
possibly different \MHc{transmission times}. 

We already argued in Example \ref{ex:transmit} that 
these two configurations can be distinguished by the 
context 
\begin{displaymath}
  \rcvtimec c x   { \matchb {x=v_0} {\bcastzeroc   {\eureka} \arb } \nil } {}
\end{displaymath}

However, the two configurations above can be made 
indistinguishable if we add to each of 
them a parallel component 
that causes a collision on channel $c$. To this end, 
let 
\[Eq(v_0,v_1) = \sigma^h.\bcastzeroc c {\arb}\]
 for 
some positive integer $h$ and value $\arb$ 
such that 
$h < \min{(\delta_{v_0}, \delta_{v_1})}$ 
and $\delta_\arb \geq \max{(\delta_{v_0}, \delta_{v_1})} - h$. Now, 
 consider the configurations $\confC_0 = \conf{\Gamma}{W_0 | Eq(v_0,v_1)}$, 
$\confC_1 = \conf{\Gamma}{W_1 | Eq(v_0,v_1)}$. 

One could hope that there exists a context which is 
able to distinguish these two configurations. 
However, before the transmission of $v_0$ ends in 
$\confC_0$, a second broadcast along channel $c$ will 
fire, causing a collision; the same happens before 
the end of transmission of value $v_1$ in $\confC_1$. 
Further, the total amount of time for which channel 
$c$ will be exposed is the same for both configurations, 
so that one can  argue that it is impossible to 
provide a context which is able to distinguish 
$\confC_0$ from $\confC_1$. In order to prove this 
to be formally true, we have to wait until the next 
section.  
\end{example}

Collisions can also be used to merge two different 
transmissions on the same channel in a 
single corrupted transmission.

\begin{example}[Merging Transmissions]
\label{ex:merge}
Let  $\Gamma$ be a stable channel environment,
$W_0 = \bcastc c {v_0} {\bcastzeroc c {v_1}}$, 
$W_1 = \bcastc c {v_1} {\bcastzeroc c {v_0}}$. 
In $\conf{\Gamma}{W_0}$ a broadcast of value $v_0$ 
along channel $c$ can fire; when the transmission of 
$v_0$ is finished, a second broadcast of value 
$v_1$ along the same channel can also fire. 
The behaviour of $\conf{\Gamma}{W_1}$ is 
similar, though the order of the two values to 
be broadcast is swapped. 
 \MHc{Note that it is possible 
to distinguish the two configurations 
$\conf{\Gamma}{W_0}$ and $\conf{\Gamma}{W_1}$ 
using  the test
\begin{displaymath}
  \rcvtimec c x   {\matchb {x=v_0} {\bcastzeroc   {\eureka} \arb }  \nil} {}
\end{displaymath}
we have already seen in the previous example. 

However} suppose now that we add a parallel component to both configurations 
which broadcasts another value along channel $c$ before 
the transmission of value $v_0$ ($v_1$) has finished, and which 
terminates after the broadcast of value $v_1$ ($v_0$) has 
begun. More formally, let 
\[ Mrg(v_0,v_1) = \sigma^{h}.\bcastzeroc c \arb \]
%%$W_1 = \sigma^{\delta_{v_1} - 1}.\bcastzeroc c \arb$; 
where $h = \operatorname{min}(\delta_{v_0}, \delta_{v_1}) -1$ and 
$\delta_{\arb} = {|} \delta_{v_0} - \delta_{v_1}{|} +2$.

Consider the configurations $\conf{\Gamma}{W_0 | Mrg(v_0,v_1)}$, 
$\conf{\Gamma}{W_1 | Mrg(v_0,v_1)}$. In both configurations 
a collision occurs; further, once the transmission 
of value $v_0$ has begun in the former configuration, 
channel $c$ will remain exposed until the transmission 
of value $v_1$ has finished. A similar behaviour can be observed on 
the second configuration. This leads to the intuition \MHf{Where is this proved?}
that $\conf{\Gamma}{W_0 | Mrg(v_0,v_1)} \simeq \conf{\Gamma'} {W_1 | 
Mrg(v_0,v_1)}$; we prove this in Example \ref{ex:merging.bisimilar}, for a particular instance 
of transmission values for $v_0, v_1$. 
\end{example}

A priori reductions ignore the passage of time, and therefore one might suspect that reduction
barbed congruence is impervious to the precise timing of activities. But the next example demonstrates that 
this is not the case. 
\begin{example}[Observing the passage of time]
\label{ex:time}
  Consider the two processes
 $Q_1 = \bcastzeroc c {v_0}$ and  $Q_2 = \sigma.{Q_1}$, 
and again let us assume that $\Gamma \vdash c:\cfree$. There is very little difference between
the behaviours of $\conf \Gamma {Q_1}$ and  $\conf \Gamma {Q_2}$; both will transmit (successfully)
the value $v_0$, although the latter is a little slower. 
However this slight difference can be observed.
Consider the test $T$ defined by 
\begin{displaymath}
  \matchb{\expsd{c}} {\bcastzeroc \eureka \arb} {  \nil    }
\end{displaymath}
In fact, $\conf{\Gamma}{(Q_1 | T)}$ can start a transmission 
along channel $c$, after which the predicate $\expsd{c}$ will be evaluated in 
the system term $T$. The resulting configuration is given by 
$\conf{\Gamma'}{\sigma^{\delta_{v_0}} | \sigma. \bcastzeroc \eureka \arb}$; 
at this point, it is not difficult to note that the configuration has 
a weak barb on $\eureka$.

%%\leaveout{However the unique reduction from $\confC_1 = \conf \Gamma (Q_2 | T)$
%%is of the form $\confC_1 \red \confC'_1 = \conf {\Gamma_1}
%%%\delta^k.\nil | T$ in which the transmission along $c$ is
%%%initiated. This is inferred with the aid of the rules (Snd) and
%%(RcvIgn); the transmission is ignored by the testing station. But now, crucially,
%%%$\Gamma_1 \vdash c: \cbusy$; therefore when time passes and the Boolean expression
%%is evaluated the result is false and the barb is not possible.}

On the other hand, 
 the \emph{unique} reduction from $\confC_2 = \conf \Gamma {(Q_2 | T)}$ 
leads to the evaluation of the exposure predicate $\expsd{c}$; since 
$\Gamma \vdash c: \cfree$ the only possibility for the 
resulting configuration is given by $\confC_2' = \conf{\Gamma}{Q_2 | \sigma}$. 
Since $\eureka$ is a fresh channel, it is now immediate to note that 
$\confC_2' \not\Downarrow_{\eureka}$ and hence also $\confC_2 \not\Downarrow_{\eureka}$.
For the test to work correctly it is essential 
%%\MHf{Maybe}
that $\Gamma \vdash c:
\cfree$.  \MHf{Where, later on?}
%% and $\Gamma \vdash_{\mathrm{t}} c: n$ for some 
%%$n > \delta_{v_0}$. 
Here we would like to point out that 
using the proof methodology developed in Section \ref{sec:bisim} we 
are able to show that if $\Gamma' \vdash_{\mathrm{t}} c: n$ 
and $n > \delta_{v_0}$ 
then $\conf {\Gamma'} Q_1 \simeq \conf {\Gamma'} Q_2$. 
\end{example}

Behind this example is the general principle that reduction barbed
congruence is actually sensitive to the passage of time; \MHc{this is proved
formally in}
Proposition~\ref{prop:delay.preservation} of
Section~\ref{sec:completeness}.

\begin{example}
  As a final example we illustrate the use of channel restriction. 
Assume that $v_1$ and $v_2$ are some kind of values which 
can be compared via a (total) order relation $\preccurlyeq$. 
Consider the configuration\\
\begin{math}
 \conf  \Gamma {  \nu c:(0,\cdot).({\bcastzeroc c {v_1} } | P_e |R)}
\end{math}
where  the station code is given by
\begin{align*}
  P_e &=   \sigma . \fix X {(\matchb {\expsd{c}} {X}{\bcastzeroc c {v_2}  })}  \\
   R &= \rcvc c x {R_1}\\
   R_1 &= \rcvc c y { \matchb{y \preccurlyeq x}{ \bcastzeroc d x}{ \bcastzeroc d y}  }
\end{align*}
Intuitively the receiver $R$ waits indefinitely for two values along
the restricted channel $c$ and broadcasts the largest on channel
$d$. Intuitively the use of channel restriction here shelters $c$ from
external interference.
Assuming $\Gamma \vdash d: \cfree$ we will be able to show that 
\begin{displaymath}
  \conf \Gamma {  \nu c:(0,\cdot).({\bcastc c {v_1} \nil} | P_e |R)}  \;\;\simeq\;\; \conf \Gamma {\sigma^{\delta_{v_1}+\delta_{v_2}+2}.\bcastc d w \nil}
\end{displaymath}
provided $w = \max(v_1,v_2)$. 
\end{example}

\section{Extensional Semantics}
\label{sec:extsem}
Proving that two configurations $\confC_1$ and $\confC_2$ are barbed congruent 
can be difficult, due to the contextuality constraint imposed in Definition \ref{def:rbc}.
Therefore, we want to give a co-inductive characterisation of the
contextual equivalence $\simeq$ 
between  \MHc{configurations}, in terms
of a standard bisimulation equivalence over some extensional LTS. In this 
section
 we first present the extensional semantics, then we recall the standard
definition of (weak) bisimulation over configurations.
We show,
by means of a number of examples, the usefulness of the actions introduced 
in the extensional semantics. 

\subsection{Extensional actions}
\label{sec:ea}
The extensional semantics is designed by addressing the question: 
what actions can be detected by an external observer? 
Example~\ref{ex:time} indicates that the passage of time
is observable. The effect of inputs received from the external environment 
also has to be taken into account. In contrast, the discussion in
Example~\ref{ex:transmit} indicates that, due to the possibility of
collisions, the treatment of \MHc{transmissions}
is more subtle. It turns out that the transmission itself is not
important; instead we must take into consideration the successful
delivery of the transmitted value.

\begin{table}[!t]
\begin{align*}
&
\Txiom{Input}{\conf \Gamma W \trans{c?v}  {W'}  }
{\conf \Gamma W \exttrans{\rcva c v} \conf {\gupd{c?v}{\Gamma}}{W'}}
&&
\Txiom{Time}
{\conf \Gamma W \trans{\sigma} {W'}}
{\conf {\Gamma} {W} \exttrans{\sigma} \conf {\gupd{\sigma}{\Gamma}} {W'}
}
\\\\
&\Txiom{Shh}{ \conf{\Gamma}{W} \trans{\sndac c v} {W'} }
{\conf \Gamma W \exttrans{\tau} \conf {\gupd{c!v}{\Gamma}}{W'}}
&&
\Txiom{TauExt}{\conf \Gamma W \trans{\tau}  {W'}}
{\conf \Gamma W \exttrans{\tau} \conf {\Gamma} {W'}}
\\\\
&
\Txiom{Deliver}{\Gamma(c)=(1,v) \Q \conf \Gamma W \trans{\sigma} {W'} }
{\conf \Gamma W \exttrans{\gamma(c,v)} \conf {\gupd{\sigma}{\Gamma}}{W'}}
&&
\Txiom{Idle}{\Gamma \vdash c : \cfree}
{\conf \Gamma W \exttrans{\iota(c)} \conf \Gamma W}
% \\[25pt]
% &
% \Txiom{Barb}{\conf{\Gamma}{W} \trans{\sndac c v} \conf {\Gamma'}{W'} \Q \Gamma(c)_{\mathbf 1} = 0}
% {\conf \Gamma W \exttrans{c\downarrow} \conf {\Gamma'}{W'}}
\end{align*}
    \caption{Extensional actions}
    \label{tab:extensional}
  \end{table}

In Table~\ref{tab:extensional} we give the rules defining the extensional actions,
$\mathcal C \exttrans{\alpha} \mathcal C'$, which can take one of the forms:
\begin{itemize}
\item Input: $\mathcal C \exttrans{c?v} \mathcal C'$, 
this is inherited directly from the intensional semantics

\item Time: $\mathcal C \exttrans{\sigma} \mathcal C'$, also inherited from the 
intensional semantics

\item Internal:   $\mathcal C \exttrans{\tau} \mathcal C'$, this corresponds to the combination
of the Internal and Transmission rules from the reduction semantics, in Definition~\ref{def:step}

\item Delivery:  $\mathcal C \exttrans{\gamma(c,v)} \mathcal C'$, this corresponds to the successful 
delivery of the value $v$ which was in transmission along the channel $c$

\item Free: $\mathcal C \exttrans{\iota(c)} \mathcal C$, a predicate indicating that channel $c$ is 
not exposed, and therefore ready to start a potentially successful transmission. 
\end{itemize}

\begin{remark}
\label{rem:operational_correspondence}
The rules provided in Table \ref{tab:extensional} guarantee that 
$\tau$-extensional actions coincide with instantaneous reductions. 
In fact, whenever $\conf{\Gamma}{W} \red \conf{\Gamma'}{W'}$ then 
 either $\conf{\Gamma}{W} \trans{\tau} W'$, and hence 
$\conf{\Gamma}{W} \exttrans{\tau} \conf{\Gamma'}{W'}$ follows by 
an application of Rule \rulename{TauExt}, with  $\Gamma' = \gupd{\tau}{\Gamma}$, 
or $\conf{\Gamma}{W} \trans{c!v} W'$ and  $\conf{\Gamma}{W} \exttrans{\tau} 
\conf{\Gamma'}{W'}$ is ensured by Rule \rulename{Shh}, with 
$\Gamma' = \gupd{c!v}{\Gamma}$.  The opposite implication 
can be proved analogously. 

Similarly, it is easy to check extensional $\sigma$-actions coincide with timed reductions: 
$\conf{\Gamma}{W} \red_{\sigma} \conf{\Gamma'}{W'}$ if and only 
if $\conf{\Gamma}{W} \exttrans{\sigma} 
\conf{\Gamma'}{W'}$.
\end{remark}

%%{\MHx
%%Some more discussion? Little examples? e.g. of Deliver?%0
%%%}
\subsection{Bisimulation equivalence}
\label{sec:bisim}
\leaveout{The definition of  reduction barbed congruence is simple and 
intuitive. However, due to the universal quantification on parallel 
contexts, it may be quite difficult to prove that 
two terms are barbed congruent. Simpler proof techniques are based on 
labelled bisimilarities. In this section, we propose an 
appropriate notion of weak bisimulation between networks based on 
weak actions.} 
\MHc{The extensional actions of the previous section endows   systems in CCCP with the structure
of an LTS. 
Weak extensional actions in this LTS are defined as usual, with
 $\mathcal C    \extTrans{\alpha} \mathcal C'$  denoting  
$\mathcal C \exttransStar{\tau} \,\exttrans{\alpha} \, \exttransStar{\tau}\, \mathcal C'$.
We will use  $\mathcal C \extTrans{\ } \mathcal C'$ to denote  
$\mathcal C  \exttransStar{\tau} \mathcal C'$,  
and the formulation of bisimulations 
is facilitated by the notation  
$\mathcal C    \extTrans{\hat{\alpha}} \mathcal C'$, which is again
standard: for $\alpha = \tau$ this denotes   $\mathcal C \extTrans{\ } \mathcal C'$ while 
for $\alpha \not= \tau$ 
it is  $\mathcal C    \extTrans{\alpha} \mathcal C'$. 
}
We now have the standard definition of weak bisimulation equivalence in the resulting LTS
which for convenience we recall. 
%%
%%\begin{definition}\label{def:bisim}
%%Let $\RR$ be a relation over configurations. Then $\calB(\RR)$ is the relation over configurations
%%determined by:  $\calC_1 \RRr \calC_2$ if for every extensional label $\alpha$
%%\begin{enumerate}[(i)]
%%\item  $\calC_1 \extTrans{\alpha}  \calC_1'$ implies  $\calC_2 \extTrans{\alpha}  \calC_2'$ 
%%for some $\calC_2'$ satisfying  $\calC_1' \RRr \calC_2'$
%%%
%%\item conversely,
%%$\calC_2 \extTrans{\alpha}  \calC_2'$ implies  $\calC_1 \extTrans{\alpha}  \calC_1'$ 
%%such that   $\calC_1' \RRr \calC_2'$.
%%\end{enumerate}
%%We say $\RR$ is a (weak) bisimulation if $\RR \,\subseteq \calB(\RR)$, and we use $\approx$ to 
%%denote the largest bisimulation, which we know to exist for the standard reasons, \cite{ccs}. 
%%\end{definition}

\begin{definition}
\label{def:bisim}
Let $\RR$ be a binary relation over configurations. 
We say that $\RR$ is a bisimulation if for every extensional
action $\alpha$, whenever $\calC_1 \RRr \calC_2$
\begin{enumerate}[label=(\roman*)]
\item  $\calC_1 \exttrans{\alpha}  \calC_1'$ implies  $\calC_2 \extTrans{\hat{\alpha}}  \calC_2'$, 
for some $\calC_2'$, satisfying  $\calC_1' \RRr \calC_2'$
\item conversely,
$\calC_2 \exttrans{\alpha}  \calC_2'$ implies  $\calC_1 \extTrans{\hat{\alpha}}  \calC_1'$, for some $\calC_1'$,
such that   $\calC_1' \RRr \calC_2'$.
\end{enumerate}
We write $\calC_1 \approx \calC_2$, if there is a bisimulation $\RR$ such that 
$\calC_1 \RRr \calC_2$.
\end{definition}
Our goal is to 
\MHc{demonstrate that this form of bisimulation provides a sound and useful 
proof method for showing behavioural equivalence between wireless systems 
described in CCCP; moreover for a large class of systems it will also turn out
to be complete. 

The next two examples show that the 
introduction of the actions 
$\iota(c)$ and $\gamma(c,v)$ are necessary for soundness.}

\begin{example}[On the rule \rulename{Idle}]
\label{ex:iota.actions}
Suppose \MHc{we were to} drop the rule \rulename{Idle} in the extensional semantics;
\MHc{then} 
consider the configurations
\begin{eqnarray*}
\conf{\Gamma_1}{W_1} &=& \tau.\nil\\
\conf{\Gamma_2}{W_2} &=& \bcastzeroc{c}{v}
\end{eqnarray*}
where $\Gamma_1(c) = (1,v)$, $\Gamma_2(c) = (0,\cdot)$ and 
$\delta_v = 1$.

\MHc{
If we were to drop the actions $\iota(c)$ from the extensional semantics
then the extensional LTSs generated by these two configurations would be isomorphic; 
recall that a broadcast action in the intensional semantics always
corresponds to a $\tau$ action in its extensional counterpart. Thus 
they would be related by the amended version of bisimulation equivalence. 

However, we also have that $\conf{\Gamma_1}{W_1} \not\simeq
\conf{\Gamma_2}{W_2}$.
This can be proved by exhibiting a distinguishing context.
To this end, consider the system
$T =\matchb{\expsd{c}}{\nil}{\bcastzeroc{\mbox{eureka}}{\mbox{\arb}}}$.
Then $\conf{\Gamma_2}{W_2 | T}$ has a weak barb on the channel $\mbox{eureka}$,
which  obviously $\conf{\Gamma_2}{W_1 | T}$ can not match.
}
%this can be proved by simply exhibiting a
%context that distinguishes the two configurations above.
%To this end, consider the system
%$T =\matchb{\expsd{c}}{\nil}{\bcastzeroc{\mbox{eureka}}{\mbox{\arb}}}$.
%It is immediate to note that $\conf{\Gamma_2}{W_2 | T}$ has a weak barb
%on the channel $\mbox{eureka}$; in fact, we have the sequence of reductions 
%\begin{eqnarray*} 
%\conf{\Gamma_2}{W_2 | T} &\red_i& \conf{\Gamma_2}{W_2 | \sigma.\bcastzeroc{\mbox{eureka}}{\mbox{\arb}}}\\
%&\red_i& \conf{\Gamma_2'}{\sigma | \sigma.\bcastzeroc{\mbox{eureka}}{\mbox{\arb}}}\\
%&\red_{\sigma}& \conf{\Gamma_2''}{\nil | \bcastzeroc{\mbox{eureka}}{\mbox{\arb}}}\\
%&\red_{i}& \conf{\Gamma_2'''}{\nil | \sigma^{\delta_\arb}}
%\end{eqnarray*}
%where $\Gamma_2' = \gupd{\snda c v}{\Gamma_2}$, $\Gamma_2'' = \gupd{\sigma}{\Gamma_2'}$ 
%and $\Gamma_2''' = \gupd{\snda {\mbox{eureka}} {\arb}}{\Gamma_2''}$; 
%since $\Gamma_2''' \vdash \mbox{eureka} : \cbusy$ it follows that 
%$\conf{\Gamma_2'''}{\sigma | \sigma^{\delta_{\arb}}} \downarrow_{\mbox{eureka}}$, 
%hence $\conf{\Gamma_2}{W_2 | T} \Downarrow_{\mbox{eureka}}$. 
%If it were $\conf{\Gamma_1}{W_1} \approxeq \conf{\Gamma_2}{W_2}$, then we should 
%also have that $\conf{\Gamma_1}{W_1 | T}$ has a weak barb on $\eureka$.
\end{example}

\begin{example}[On the rule \rulename{Deliver}]
\label{ex:gamma.actions}
\MHc{
Consider the configuration $\conf{\Gamma_2}{W_2}$ from the 
previous example; consider also the configuration $\conf{\Gamma_2}{W_3}$, where $W_3 = \bcastzeroc c w$ 
for some value $w$, different from $v$, such that $\delta_w = 1$.
%\begin{equation*}
%\conf{\Gamma_2}{W_3} = \bcastzeroc{c}{w}
%\end{equation*}
%where $\delta_w = 1$ and $\Gamma_2$ is as in the previous example, 
Finally, let  
$T' = \rcvtimec c x {\matchb{x=v}{\bcastzeroc{\mbox{eureka}}{\mbox{\arb}}}{\nil}}{}$.
Then, assuming $w$ is different from $v$, 
 $\conf{\Gamma_2}{W_3 |T'}$ can not produce a barb on  $\mbox{eureka}$. 
On the other hand, $\conf{\Gamma_2}{W_2 |T'}$  can produce such a barb. 
It follows that $\conf{\Gamma_2}{W_2} \not\simeq
\conf{\Gamma_2}{W_3}$.

Note also that $\conf{\Gamma_2}{W_3} \not\approx \conf{\Gamma_2}{W_2}$, 
since the (weak) action $\conf{\Gamma_2}{W_3} \extTrans{\gamma(c,w)} \conf{\Gamma}{\nil}$ 
cannot be matched by $\conf{\Gamma_2}{W_2}$.
However, if we were to drop the rule \rulename{Deliver} in the extensional semantics, 
thereby eliminating the actions  $\gamma(c,v)$, then it would be straightforward
to exhibit a bisimulation containing the pair $(\conf{\Gamma_2}{W_3}, \conf{\Gamma_2}{W_2})$. 
Thus again the amended version of bisimulation equivalence would be unsound. 
}
\end{example}

The two examples above show that both rules \rulename{Idle} and
\rulename{Deliver} are necessary to achieve the soundness of our
bisimulation proof method for reduction barbed congruence. 

\MHc{In the remainder of this section we give a further series of examples,
showing that bisimulations in our extensional LTS offer a viable proof technique
for demonstrating behavioural equivalence for at least simple  wireless systems.}

\begin{example}[Transmission]
\label{ex:transmit.bisim} 
\MHc{
Here we revisit Example~\ref{ex:transmit}. 
Let $\Gamma$ be a stable channel environment, and 
consider the configurations 
$\confC_0 = \conf{\Gamma}{W}$,  
$\confC_1 = \conf{\Gamma}{V}$, 
where $W = {\bcastc c {v_0} P | \bcastzeroc c {v_1}}$, 
$V = {\bcastc c {v_1} P | \bcastzeroc c {v_0}}$;  
note that these two configurations are taken from
the second part of  Example \ref{ex:transmit}.  

Our aim is to show that $\confC_0 \approx \confC_1$, when $\delta_{v_0} = \delta_{v_1}$; for convenience 
let us assume that $\delta_{v_0} = \delta_{v_1} = 1$. 
The idea here is to describe the required bisimulation by matching up system terms. 
}
To this end we define 
%the channel environments $\Gamma_0 := 
%\gupd{\sndac c {v_0}}{\Gamma}, \Gamma_1 := 
%\gupd{\sndac c {v_1}}{\Gamma}, \Gamma_{\err} := 
%\Gamma[ c \mapsto (1,\err)]$ and 
the following system terms: 
\[
\begin{array}{rcl@{\hspace*{1cm}}rcl}
W_0 &=& \sigma.P | \bcastzeroc c {v_1}
 & V_1 &=& \sigma.P | \bcastzeroc c {v_0}\\[2pt]
W_1 &=& \bcastc c {v_0} P | \sigma
 & V_0 &=& \bcastc c {v_1} P | \sigma\\[2pt]
E &=& \sigma.P | \sigma 
&
{E'} &=& {P | \nil} 
%%&&\\
%%\\
%%\\ 
%%V_0^1 &=& \sigma.P | \sigma\\
%V' &=& P | \nil
\end{array}
\]
Then for any channel environment $\Delta$ 
we have 
the following transitions in the extensional semantics:
\[
\begin{array}{rcl@{\hspace*{1cm}}rcl}
\conf{\Delta}{W} &\exttrans{\tau}& \conf{\gupd{\sndac c {v_0}}{\Delta}}{W_0}&
\conf{\Delta}{V} &\exttrans{\tau}& \conf{\gupd{\sndac c {v_0}}{\Delta}}{V_0}\\
%&&\\
\conf{\Delta}{W} &\exttrans{\tau} &\conf{\gupd{\sndac c {v_1}}{\Delta}}{W_1}&
\conf{\Delta}{V} &\exttrans{\tau} &\conf{\gupd{\sndac c {v_1}}{\Delta}}{V_1}\\
%&&\\
\conf{\Delta}{W} &\exttrans{\rcva d w} &\conf{\gupd{\rcva d w}{\Delta}}{W}&
\conf{\Delta}{V} &\exttrans{\rcva d w} & \conf{\gupd{\rcva d w}{\Delta}}{V}\\
%&&\\
\conf{\Delta}{W} &\exttrans{\iota(d)} &\conf{\Delta}{W} \mbox{ if } \Delta \vdash d: \cfree&
\conf{\Delta}{V} &\exttrans{\iota(d)} &\conf{\Delta}{V} \mbox{ if } \Delta \vdash d: \cfree\\
&&\\
\conf{\Delta}{W_0} &\exttrans{\tau}& \conf{\gupd{\sndac c {v_1}}{\Delta}}{E}&
\conf{\Delta}{V_0} &\exttrans{\tau}& \conf{\gupd{\sndac c {v_1}}{\Delta}}{E}\\
%&&\\
\conf{\Delta}{W_0} &\exttrans{\rcva d w}& \conf{\gupd{\rcva d w}{\Delta}}{W_0}&
\conf{\Delta}{V_0} &\exttrans{\rcva d w}& \conf{\gupd{\rcva d w}{\Delta}}{V_0}\\
%&&\\
\conf{\Delta}{W_0} &\exttrans{\iota(d)}& \conf{\Delta}{W_0} \mbox{ if } \Delta \vdash d: \cfree&
\conf{\Delta}{V_0} &\exttrans{\iota(d)}& \conf{\Delta}{V_0} \mbox{ if } \Delta \vdash d: \cfree\\
&&\\
\end{array}
\]
\[
\begin{array}{rcl@{\hspace*{1cm}}rcl}
\conf{\Delta}{W_1} &\exttrans{\tau}& \conf{\gupd{\sndac c {v_0}}{\Delta}}{E}&
\conf{\Delta}{V_1} &\exttrans{\tau}& \conf{\gupd{\sndac c {v_0}}{\Delta}}{E}\\
%&&\\
\conf{\Delta}{W_1} &\exttrans{\rcva d w}& \conf{\gupd{\rcva d w}{\Delta}}{W_1}&
\conf{\Delta}{V_1} &\exttrans{\rcva d w}& \conf{\gupd{\rcva d w}{\Delta}}{V_1}\\
%&&\\
\conf{\Delta}{W_1} &\exttrans{\iota(d)}& \conf{\Delta}{W_1} \mbox{ if } \Delta \vdash d: \cfree&
\conf{\Delta}{V_1} &\exttrans{\iota(d)}& \conf{\Delta}{V_1} \mbox{ if } \Delta \vdash d: \cfree
\end{array}
\]
\MHc{
Here $d$ ranges over arbitrary channel names, including $c$\MHf{Andrea: is this true?}.

Then consider the following relation:
\[
\Ss = \{(\conf{\Delta}{W}, \conf{\Delta}{V}),\,\, (\conf{\Delta}{W_0},\,\,
            \conf{\Delta}{V_0}), (\conf{\Delta}{W_1}, \conf{\Delta}{V_1}) \;|\; 
\Delta \text{ is a channel environment}\} \enspace . 
\]
Using the above tabulation of actions one can now show  that $\Ss$ is a bisimulation; 
for $\confC \,\Ss\, \confC' $ 
each possible action of $\confC$  can be
matched by $\confC'$ by performing exactly the same action, and vice-versa. 

Since $(\confC_0, \confC_1) \in \Ss$,  it follows that $\confC_0 \approx \confC_1$. 
}
\end{example}

\begin{table}
\[
\begin{array}{rcl}
\conf{\Delta}{W} &\Ss&  \conf{\Delta}{V}\\ 
\conf{\Delta}{W_0} &\Ss& \conf{\Delta}{V_0}\\ 
\conf{(\Delta[c \mapsto (1, v_0)])}{W_0} &\Ss& \conf{(\Delta[c \mapsto (2, v_1)])}{V_1}\\ 
\conf{(\Delta[c \mapsto (1, \err)])}{W_0} &\Ss& \conf{(\Delta[c \mapsto (2, \err)])}{V_1}\\
\conf{\Lambda}{W_\arb} &\Ss& \conf{\Lambda}{V_{\arb}}\\ 
\conf{\Delta}{W_\err} &\Ss& \conf{\Delta}{V_\err}\\ 
\conf{\Delta}{W'} &\Ss& \conf{\Delta}{V'}
\end{array}
\]
$\Delta$ arbitrary channel environment,\\
$\Lambda$ arbitrary channel environment such that $\Lambda(c) = (k,w)$ for some $k \geq 2$
\caption{A relation $\Ss$ for comparing the configurations 
$\confC_0, \confC_1$ of Example \ref{ex:equators.bisim}}
\label{fig:equators.bisimulation}
\end{table}

\begin{example}[Equators]
\label{ex:equators.bisim}
Let us consider  the configurations $\confC_0, \confC_1$ 
of Example \ref{ex:equators}. Recall that 
$\confC_0 = \conf{\Gamma}{W}$, where $W = \bcastzeroc c {v_0} | \sigma^{h}.
\bcastzeroc c {\arb}$ and $\confC_1 = \conf{\Gamma}{V}$, where 
$V = \bcastzeroc c {v_1} 
| \sigma^{h}.\bcastzeroc c {\arb}$; further, recall that $\Gamma$ is a stable 
channel environment and $h, \arb$ are a positive integer and a 
value, respectively, such that $h < \min{(\delta_{v_0}, \delta_{v_1})}$, 
$\delta_{\arb} \geq \max{(\delta_{v_0}, \delta_{v_1})} - h$. 
Without loss of generality, for this example we assume 
$\delta_{v_0} = 1, \delta_{v_1} = 2$, 
$h = 0$ and $\delta_{\arb} = 2$.

For the sake of \MHc{convenience}  we define the following 
%channel 
%environments: $\Gamma_0 = \gupd{\sndac c {v_0}}{\Gamma}$, 
%$\Gamma_1 = \gupd{\sndac c {v_1}}{\Gamma}$, 
%$\Gamma_{\arb} = \gupd{\sndac c {\arb}}{\Gamma}$, 
%$\Gamma_{\err} = \Gamma[c \mapsto (2, \err)]$,
%$\Gamma'_{\err} = \Gamma[c \mapsto (1,\err)]$.
%We also assume
%the following 
system terms:
\[
\begin{array}{rcl@{\hspace*{1cm}}rcl}
W_0 &=& \sigma | \bcastzeroc {c} {\arb}&
V_1 &=& \sigma^2 | \bcastzeroc {c} {\arb}\\
W_{\arb} &=& \bcastzeroc c {v_0} | \sigma^2&
V_{\arb} &=& \bcastzeroc c {v_1} | \sigma^2\\
W_\err &=& \sigma | \sigma^2&
V_{\err} &=& \sigma^2 | \sigma^2\\
W' &=& \nil | \sigma&
V' &=& \sigma | \sigma\\
E &=& \nil | \nil
\end{array}
\]

Let us consider the relation $\Ss$ depicted in 
\MHc{Table~\ref{fig:equators.bisimulation}}; note that 
$(\confC_0, \confC_1) \in \Ss$, so that in order to 
prove that $\confC_0 \approx \confC_1$ it is sufficient to 
show that \MHc{$\Ss$} is  a bisimulation. 
Note that in the relation $\Ss$ the system terms $W_\arb, V_\arb$ are always associated with a channel environment 
in which the channel $c$ is exposed. In fact, if $\Lambda$ were a channel environment such that 
$\Lambda \vdash c: \cfree$, it would not be difficult to prove that $\conf{\Lambda}{W_\err} \not\approx 
\conf{\Lambda}{V_\err}$; this is because the values broadcast by these two configurations are different. 

\MHc{
Let us  list the main the extensional actions from configurations using these system terms:
}
\[
\begin{array}{lcr} 
\conf{\Delta}{W} &\exttrans{\tau}& \conf{(\Delta[c \mapsto (1, v_0)])}{W_0} \mbox{ if } \Delta \vdash c: \cfree\\
\conf{\Delta}{V} & \exttrans{\tau}& \conf{(\Delta[c \mapsto (2, v_1)])}{V_1}\mbox{ if } \Delta \vdash c: \cfree\\
\conf{\Delta}{W} &\exttrans{\tau}& \conf{(\Delta[c \mapsto (2, \arb)])}{W_\arb}\\
\conf{\Delta}{V} & \exttrans{\tau}& \conf{(\Delta[c \mapsto (2, \arb)])}{V_\arb}\\
\conf{\Delta}{W} &\exttrans{\rcva d w }& \conf{(\gupd{\rcva d w}{\Delta})}{W}\\
\conf{\Delta}{V} &\exttrans{\rcva d w }& \conf{(\gupd{\rcva d w}{\Delta})}{V}\\
\\
\conf{(\Delta[c \mapsto (1,v_0)])}{W_0} &\exttrans{\tau}& \conf{(\Delta[c \mapsto (2,\err)])}{W_\err}\\
\conf{(\Delta[c \mapsto (2,v_1)])}{V_1} &\exttrans{\tau}& \conf{(\Delta[c \mapsto (2,\err)])}{W_\err}\\
\conf{\Delta}{W_0} & \exttrans{\rcva c w}& \conf{(\Delta[c \mapsto (1, \err)])}{W_0} \mbox{ if } \Delta \vdash c: \cbusy, \delta_w = 1\\
\conf{\Delta}{V_1} & \exttrans{\rcva c w}& \conf{(\Delta[c \mapsto (2, \err)])}{V_1} \mbox{ if } \Delta \vdash c: \cbusy, \delta_w = 1\\
\conf{\Delta}{W_0} & \exttrans{\rcva c w}& \conf{(\Delta[c \mapsto (\delta_w, \err)])}{W_0} \mbox{ if } \Delta \vdash c: \cbusy, \delta_w > 1\\
\conf{\Delta}{V_1} & \exttrans{\rcva c w}& \conf{(\Delta[c \mapsto (\delta_w, \err)])}{V_1} \mbox{ if } \Delta \vdash c: \cbusy, \delta_w > 1\\
\\
\conf{\Lambda}{W_{\arb}} &\exttrans{\tau}& \conf{(\gupd{\sndac c {v_0}}{\Lambda})}{W_\err}\\
\conf{\Lambda}{V_\arb} &\exttrans{\tau}& \conf{(\gupd{\sndac c {v_1}}{\Lambda})}{V_{\err}}\\
\\
\conf{\Delta}{W_\err} &\exttrans{\sigma}& \conf{(\gupd{\sigma}{\Delta})}{W'}\\
\conf{\Delta}{V_\err} &\exttrans{\sigma}& \conf{(\gupd{\sigma}{\Delta})}{V'}\\
\\
\conf{\Delta}{W'} &\exttrans{\sigma}& \conf{(\gupd{\sigma}{\Delta})}{E}\\
\conf{\Delta}{V'} & \exttrans{\sigma}& \conf{(\gupd{\sigma}{\Delta})}{E}
\end{array}
\]
Here $\Delta, \Lambda$ are two arbitrary channel environments, \MHc{but $\Lambda$ 
is subject to the constraint that} $\Lambda(c) = (k, w)$ for 
some value $w$ and integer $k \geq 2$. 
This last requirement ensures that $(\gupd{\snda c v_0}{\Lambda}) = (\gupd{\snda c v_1}{\Lambda})$.
\MHc{
With the aid of this tabulation one can now show that $\Ss$ is indeed a bisimulation and therefore that
 $\confC_0 \approx \confC_1$.}  
\end{example}

\begin{table}
\[
\begin{array}{rcl}
\conf{\Delta}{W} &\Ss& \conf{\Delta}{V}\\
\conf{\Delta[c \mapsto (1,w)]}{W_0} &\Ss& \conf{\Delta[c \mapsto (2,w)]}{V_1}\\ 
\conf{\Delta[c \mapsto (k+2,w)}{W_0}&\Ss&\conf{\Delta[c \mapsto (k+2,w)]}{V_1}\\
\conf{\Delta[c \mapsto (k+3,w)]{W_\arb}}&\Ss&\conf{\Delta[c \mapsto (k+3,w)]}{V_\arb}\\
%\conf{\Delta_3}{W_\arb} &\Ss& \conf{\Delta_3}{V_{\arb}}\\
\conf{\Delta[c\mapsto(k+3,\err)}{W_{\err}} &\Ss& \conf{\Delta[c \mapsto(k+3,\err)]}{V_{\err}}\\
\conf{\Delta[c \mapsto(k+2,\err)}{W'} &\Ss& \conf{\Delta[c \mapsto(k+2,\err)}{V'}\\
\conf{\Delta[c \mapsto(k+2, \err)]}{W_1}&\Ss& \conf{\Delta[c \mapsto(k+2,\err)]}{V'}\\ 
\conf{\Delta[c \mapsto(k+1, \err)]}{E'}&\Ss& \conf{\Delta[c \mapsto(k+1, \err)]}{V''}
\end{array}
\]
$\Delta$ arbitrary channel environment, $w$ arbitrary value 
(possibly $\err$) and $k \geq 0$.
\caption{A relation $\Ss$ for comparing the configurations 
$\confC_0, \confC_1$ of Example \ref{ex:merging.bisimilar}}
\label{fig:merging.bisimulation}
\end{table}

\begin{example}[Merging]
\label{ex:merging.bisimilar}
The last example we provide considers the merging of 
two transmissions in a single \MHc{transmission}
as suggested in the Example~\ref{ex:merge}.  
Let $\Gamma$ be a stable channel environment and $v_0, v_1$ 
be two values such that $\delta_{v_0}=1, \delta_{v_1}=2$. 
\MHc{Also let}  $\arb$ be a value such that $\delta_{\arb} = 3$. 
Consider the configurations 
\begin{align*}
  &\confC_0 = \conf{\Gamma}{W}
  &&
  &\confC_1 = \conf{\Gamma}{V}
\end{align*}
where $W = {\bcastc c {v_0} 
{\bcastzeroc c {v_1}} | \bcastzeroc c \arb}$ 
and $V = {\bcastc c {v_1} {\bcastzeroc c {v_0}} | 
\bcastzeroc c \arb}$.

Then  $\confC_0 \approx \confC_1$. 
As in previous examples, this statement can be proved formally 
by exhibiting a bisimulation that contains the \MHc{pair} $(\confC_0, \confC_1)$; 
to this end, define the following system terms:
\[
\begin{array}{rcl@{\hspace*{1cm}}rcl}
W_0 &=& \sigma.\bcastzeroc {c}{v_1} | \bcastzeroc c {\arb}&
V_1 &=& \sigma^2.\bcastzeroc{c}{v_0} | \bcastzeroc c {\arb}\\
W_{\arb} &=& \bcastc c {v_0} {\bcastzeroc c {v_1}} | \sigma^3&
V_{\arb} &=& \bcastc c {v_1}{\bcastzeroc c {v_0}} | \sigma^3\\
W_{\err} &=& \sigma.\bcastzeroc c {v_1} | \sigma^3&
W_{\err} &=& \sigma^2.\bcastzeroc c {v_0} | \sigma^3\\
W' &=& \bcastzeroc c {v_1} | \sigma^2\\
W_1 &=&\sigma^2 | \sigma^2&
V' &=& \sigma.\bcastzeroc c {v_0} | \sigma^2\\
E' &=& \sigma | \sigma&
V'' &=& \bcastzeroc c {v_0} | \sigma\\
E &=& \nil | \nil
\end{array}
\]
\MHc{
Consider now the relation $\Ss$ depicted in Table~\ref{fig:merging.bisimulation}; 
note that $\confC_0 \,\,\Ss\,\, \confC_1$. 
%We leave the reader to check that 
%$\Ss$ is also a weak bisimulation, from which  $\confC_0 \approx \confC_1$
%follows. 
Also, $\Ss$ is a weak bisimulation. In order to show this, we list the non-trivial transitions 
for both configurations $\confC_0, \confC_1$ and their 
derivatives, which are needed to perform the proof.
}
\[
\begin{array}{lcr} 
\conf{\Delta[(c \mapsto (0,\cdot)]}{W} &\exttrans{\tau}& \conf{\Delta[c\mapsto(1,v_0)]}{W_0}\\
\conf{\Delta[(c \mapsto(0,\cdot)]}{V} &\exttrans{\tau}& \conf{\Delta[c\mapsto(2,v_1)]}{V_1}\\
&&\\
\conf{\Delta[c \mapsto(0,\cdot)]}{W} &\exttrans{\tau}& \conf{\Delta[c \mapsto(3,\arb)]}{W_\arb}\\
\conf{\Delta[c \mapsto(0,\cdot)]}{V} &\exttrans{\tau}&\conf{\Delta[c \mapsto(3,\arb])]}{V_{\arb}}\\
&&\\
\conf{\Delta[c \mapsto(k, \cdot)]}{W} & \exttrans{\tau} &\conf{\Delta[c \mapsto (k, \err)]}{W_0} \mbox{ if } k > 0\\
\conf{\Delta[c \mapsto(k, \cdot)]}{V} & \exttrans{\tau} &\conf{\Delta[c \mapsto(2, \err)]}{V_1} \mbox{ if } 0 < k \leq 2\\
\conf{\Delta[c \mapsto(k,\cdot)]}{V} & \exttrans{\tau} &\conf{\Delta[c \mapsto(k, \err)]}{V_1} \mbox{ if } k > 2\\
&&\\
\conf{\Delta[ c \mapsto(k,\cdot)]}{W} & \exttrans{\tau}& \conf{\Delta[c \mapsto(3,\err)]}{W_\arb} \mbox{ if } 0 < k \leq 3\\
\conf{\Delta[ c \mapsto(k,\cdot)]}{W} & \exttrans{\tau}& \conf{\Delta[c \mapsto(k,\err)]}{W_\arb} \mbox{ if } k > 3\\
\conf{\Delta[ c \mapsto(k, \cdot)]}{V} & \exttrans{\tau}& \conf{\Delta[c \mapsto(3,\err)]}{V_\arb} \mbox{ if } 0 < k \leq 3\\
\conf{\Delta[ c \mapsto(k, \cdot)]}{V} & \exttrans{\tau}& \conf{\Delta[c \mapsto(k,\err)]}{V_\arb} \mbox{ if } k > 3\\

&&\\
\conf{\Delta}{W} &\exttrans{d?v}& \conf{\gupd{d?v}{\Delta}}{W}\\
\conf{\Delta}{V} &\exttrans{d?v}& \conf{\gupd{d?v}{\Delta}}{V}\\

\end{array}
\]
\[
\begin{array}{lcr}
\conf{\Delta[(c \mapsto(1,v_0)]}{W_0} &\exttrans{\tau}& \conf{\Delta[(c \mapsto (3, \err)]}{W_\err}\\
\conf{\Delta[c \mapsto(2,v_1)]}{V_1} &\exttrans{\tau}& \conf{\Delta[(c \mapsto (3, \err)]}{V_{\err}}\\
&&\\
\conf{\Delta[c \mapsto(k,\cdot)]}{W_0}&\exttrans{\tau}& \conf{\Delta[c \mapsto (3, \err)]}{W_\err} \mbox{ if } 0 < 3 \leq k\\
\conf{\Delta[c \mapsto(k,\cdot)]}{V_1} & \exttrans{\tau}& \conf{\Delta[c \mapsto (3, \err)]}{V_\err} \mbox{ if } 0 < 3 \leq k\\
&&\\
\conf{\Delta[c \mapsto(k,\cdot)]}{W_0} &\exttrans{\tau}& \conf{\Delta[c \mapsto(k, \err)]}{W_{\err}} \mbox{ if } k > 3\\
\conf{\Delta[c \mapsto(k, \cdot)]}{V_1} &\exttrans{\tau}& \conf{\Delta[c \mapsto(k, \err)]}{V_{\err}} \mbox{ if } k > 3\\
&&\\
\conf{\Delta}{W_0} &\exttrans{d?w}& \conf{\gupd{d?w}{\Delta}}{W_0}\\
\conf{\Delta}{V_1} &\exttrans{d?v}& \conf{\gupd{d?w}{\Delta}}{V_1}
\end{array}
\]
\[
\begin{array}{lcr}
\conf{\Delta[c \mapsto (k,\cdot)]}{W_\arb} &\exttrans{\tau}& \conf{\Delta[c \mapsto k,\cdot]}{W_\err} \mbox{ if } k > 3\\
\conf{\Delta[c \mapsto (k,\cdot)]}{V_\arb} &\exttrans{\tau}& \conf{\Delta[c \mapsto k,\cdot]}{V_\err} \mbox{ if } k > 3\\
&&\\
\conf{\Delta}{W_\arb} &\exttrans{d?w}& \conf{\gupd{d?w}{\Delta}}{W_\arb}\\
\conf{\Delta}{V_\arb} &\exttrans{d?w}&\conf{\gupd{d?w}{\Delta}}{V_\arb}
\end{array}
\]
\[
\begin{array}{lcr}

\end{array}
\]
\end{example}

\section{Full abstraction}
\label{sec:fullabstraction}
\MHc{
In this section, we show that the co-inductive proof method based on the
bisimulation of the previous section is sound with 
respect to the contextual equivalence of Section~\ref{sec:cxt}; 
this is the
subject of Section~\ref{sec:soundness}.
Moreover it is complete for a large class of systems. 
This class is isolated in Section \ref{sec:wellformed}, and the completeness
result is then given in Section~\ref{sec:completeness.proof}. 
}

\subsection{Soundness}
\label{sec:soundness}

\MHc{
In this section we prove that 
(weak) bisimulation equivalence is 
contained in reduction barbed congruence. The main difficulty is in proving the contextuality of 
the bisimulation equivalence. But first some  auxiliary 
results. 
}

%\MHc{
%We end this section with a small technical result, which will be extremely useful
%in the development of our behavioural theory. Informally it says that instantaneous reductions
%do not decrease the remaining delivery time of values. }
%\begin{lemma}
%\label{lem:instantaneous.exposure}
%Whenever $\conf{\Gamma}{W}\red_i^{\ast} \conf{\Gamma'}{W'}$ it holds that 
%$\Gamma \leq \Gamma'$. 
%\end{lemma}
%\begin{proof}
%Follows from the definition of $\gupd{\snda c v}{\Gamma}$ and 
%$\gupd{\tau}{\Gamma}$
%\end{proof}
\begin{lemma}[Update of Channel Environments]
\label{lem:taus-Gamma}
If $\conf{\Gamma}{W} \extTrans{} 
\conf{\Gamma'}{W'}$ then $\Gamma \leq \Gamma'$.
\end{lemma}

\begin{proof}
See the Appendix, Page \pageref{proof:taus-gamma}.

\end{proof}

%
%Corollary \ref{cor:iota.actions} is very 
%useful when proving that the exposure state of 
%channels is preserved by bisimilar configurations.

Below we report a result on channel exposure for bisimilarity;
\MHc{ a similar result for  reduction barbed congruence will
also be proved,}
in Proposition~\ref{prop:exposure}. 
\begin{lemma}[Channel exposure w.r.t. $\approx$]
\label{lem:channel-exposure}
Whenever $\conf {\Gamma_1} {W_1} \approx \conf {\Gamma_2} {W_2}$
then  $\Gamma_1 \vdash c: \cfree$ if and only if $\Gamma_2 \vdash c: \cfree$. 
\label{lem:exp-equivalent}
\end{lemma}
\begin{proof}
See the Appendix, Page \pageref{proof:channel-exposure}
\end{proof}

%Next we need to prove a side results concerning extensional input 
%actions. This informally says that once an input action has been performed, 
%then the system term of  configuration is no longer affected by 
%(strong) inputs along the same channel. 
%
%\begin{lemma}
%\label{prop:input-enabled}
%For any configuration $\conf{\Gamma}{W}$, channel 
%$c$ and value $v$, there exist $\Gamma'$ and $W'$ 
%such that $\conf{\Gamma}{W} \extTrans{\rcva c v} 
%\conf{\Gamma'}{W'}$.
%\end{lemma}
%
%The second result we need states that subsequent inputs along the same 
%channel, within the same time slot, do not alter the state of a system term 
%(though they may alter the state of the channel environment).

%\begin{lemma}
%\label{lem:subsequent-inputs}
%Whenever $\conf{\Gamma}{W} \exttrans{\rcva c v} \conf{\Gamma_1}{W_1}\exttrans{\rcva c v} \conf{\Gamma_2}{W_2}$, then $W_2 = W_1$.
%\end{lemma}
%
%\begin{proof}
%See the Appendix for details.
%\end{proof}

%\begin{proof}
%Proposition \ref{prop:input} ensures that there exists a sequence of system 
%terms $W_0, \cdots, W_k$ for some $k\geq 0$ such that 
%$W = W_0$, $\conf{\Gamma}{W_i} \trans{\tau} W_{i+1}$ for any 
%$i < k$ and $\conf{\Gamma}{W_k} \trans{c?v} W'$ for some $W'$. 
%In the extensional semantics this translates to the sequence of transitions 
%$\conf{\Gamma}{W_0} \exttrans{\tau} \cdots \exttrans{\tau} \conf{\Gamma}{W_{k}} 
%\exttrans{c?v} \conf{\gupd{\rcva c v}{\Gamma}}{W'}$, that is 
%$\conf{\Gamma}{W} \extTrans{\rcva c v} \conf{\gupd{\rcva c v}{\Gamma}}{W'}$.
%\end{proof}

In order to prove that weak bisimulation is 
sound with respect to reduction barbed congruence we need 
to show that $\approx$ is 
preserved by parallel composition.

\begin{theorem}[$\approx$ is contextual]
\label{thm:congruence}  
\MHc{Suppose
$\conf {\Gamma_1}{W_1} \approx \conf {\Gamma_2}{W_2}$. 
Then for any system term $W$, $\conf {\Gamma_1} 
{(W_1 | W)} \approx\conf {\Gamma_2} {(W_2 | W)}$.}
\end{theorem}
\begin{proof}
Let the relation $\Ss$ over configurations be defined as follows:
\[
\{ \big( \conf {\Gamma_1} {W_1|W} \, , \, \conf {\Gamma_2} {W_2|W}) : \q
\conf {\Gamma_1} {W_1} \approx \conf {\Gamma_2}{W_2}\, \}
\]
It is sufficient to show that $\Ss$ is a bisimulation in the extensional semantics. 
To do so, by symmetry,  we need to show that an arbitrary extensional action
\begin{align}\label{eq:anaction}
  \conf {\Gamma_1} {W_1 | W} \exttrans{\alpha}
\conf {\widehat{\Gamma_1}} {\widehat{W_1}}  
\end{align}
can be matched by  $\conf {\Gamma_2} {W_2 | W}$ via a corresponding weak extensional action.

\MHc{
The action (\ref{eq:anaction}) can be inferred by any of the  six rules in Table~\ref{tab:extensional}.
We consider only one case, the most interesting one \rulename{Shh}. 
So here $\alpha$ is $\tau$ and 
 $\conf {\Gamma_1} {W_1 | W} \trans{c!v} {\widehat{W_1}} $, for
some $c$ and $v$, and $\widehat{\Gamma_1} = \gupd{c!v}{\Gamma}$. This transition in turn can always be inferred by an application of the rule
\rulename{Sync}, or its symmetric counterpart, from Table~\ref{tab:proc}. Here we only consider the former case; 
the proof for the second case is slightly different, though it uses the same proof strategies illustrated below. 
For the case we are considering, we have that 
}
\begin{itemize}
\item $\conf {\Gamma_1} {W_1} \trans{c!v}W'_1$
\item  $\conf {\Gamma_1} {W} \trans{c?v}W'$
\item ${\widehat{W_1}} = W'_1 | W'$
\end{itemize}

By an application of rule \rulename{Shh} it follows that $\conf {\Gamma_1}
{W_1} \exttrans{\tau} \conf {\widehat{\Gamma_1}}{W'_1}$. 
Since $\conf {\Gamma_1}{W_1} \approx \conf{\Gamma_2}{W_2}$, there is $\conf {\Gamma'_2}{W'_2}$ 
such that $\conf {\Gamma_2}{W_2} \extTrans{\ }\conf {\Gamma'_2}{W'_2}$ and 
$\conf {\Gamma'_1}{W'_1} \approx \conf{\Gamma'_2}{W'_2}$. Note that Lemma \ref{lem:exp-equivalent} 
ensures that whenever $\Gamma_1 \vdash d: \cbusy$ then also $\Gamma_2 \vdash d: \cbusy$, for any channel 
$d$ . Similarly, if $\widehat{\Gamma_1} \vdash d: \cbusy$ then $\widehat{\Gamma_2} \vdash d: \cbusy$. 
That is, $\Gamma_1$ agrees with $\Gamma_2$ on the exposure state of each channel; the same applies 
to $\widehat{\Gamma_1}$ and $\widehat{\Gamma_2}$.

Further, recall that $\widehat{\Gamma_1} = \gupd{c!v}{\Gamma_1}$. Therefore we have that, for any channel $d \neq c$, 
$\Gamma_1 \vdash d: \cbusy$ iff $\widehat{\Gamma_1} \vdash d: \cbusy$; for channel $c$, 
we have that $\widehat{\Gamma_1} \vdash c: \cbusy$. 
That is, the exposure states of $\Gamma_1$ and $\widehat{\Gamma_1}$ differ only in the entry at channel $c$, 
and only if such a channel was idle in $\Gamma_1$. 

Since $\Gamma_1$ and $\widehat{\Gamma_1}$ agree 
with $\Gamma_2, \widehat{\Gamma_2}$, respectively, on the exposure state of each channel, 
it has also to be that the exposure states of $\Gamma_2$ and $\widehat{\Gamma_2}$ differ only at the entry 
at channel $c$, and only when the latter is idle in $\Gamma_2$; formally $\Gamma_2 \vdash d: \cbusy$ iff 
$\widehat{\Gamma_2} \vdash d: \cbusy$ when $d \neq c$, 
and  $\widehat{\Gamma_2} \vdash c: \cbusy$.

Next we show that the action $\conf{\Gamma_1}{W_1 | W} \exttrans{\tau}$ $\conf{\widehat{\Gamma_1}}{W_1' | W'}$ 
can be matched by a weak action $\conf{\Gamma_2}{W_2 | W} \extTrans{\,} \conf{\widehat{\Gamma_2}}{W_2' | W'}$. 
Since $\conf{\widehat{\Gamma_1}}{W_1'} \approx \conf{\widehat{\Gamma_2}}{W_2'}$, the above statement would imply that 
$(\conf{\widehat{\Gamma_1}}{W_1' | W}) \, \Ss \, (\conf{\widehat{\Gamma_2}}{W_2' | W})$, which is exactly what we want to prove.
There are two possible cases, according to whether $\conf{\Gamma_1}{W}$ is able to detect a value 
broadcast along channel $c$:
%\MHc{depending on whether or not $c$ is exposed in $\Gamma_1$.} 
\begin{enumerate}
\item 
\label{firstcase1.1}
%Let $\Gamma_1 \vdash c: \cbusy$. 
$\neg\isrcv{\conf{\Gamma_1}{W}, c}$. 
By Lemma~\ref{lem:rcv-enabling}(\ref{rcv-enabling1}), in the 
transition $\conf {\Gamma_1}{W} \trans{c?v} W'$ it must be \MHc{that} $W'=W$. 
We have to show that the transition $\conf{\Gamma_2}{W_2} \extTrans{\,} \conf{\widehat{\Gamma_2}}{W_2'}$ 
implies that $\conf{\Gamma_2}{W_2 | W} \extTrans{\,} \conf{\widehat{\Gamma_2}}{W_2' | W}$. 
To this end, we prove a stronger statement: whenever we have a sequence of transitions 
\[
\conf{\Gamma^0}{V^0} \exttrans{\tau} \conf{\Gamma^1}{V^1} 
\exttrans{\tau} \cdots \exttrans{\tau} \conf{\Gamma^n}{V^n}
\]
of arbitrary length $n \geq 0$, and such that for any $d \neq c$,  
$\Gamma^0 \vdash d: \cbusy$ if and only if $\Gamma^n \vdash d: \cbusy$, and 
$\neg\isrcv{\conf{\Gamma_0}{W}}$. Then  
\[
\conf{\Gamma^0}{V^0 | W} \exttrans{\tau} \conf{\Gamma^1}{V^1 | W} 
\exttrans{\tau} \cdots \exttrans{\tau} \conf{\Gamma^n}{V^n | W}
\]
Further, $\neg\isrcv{\conf{\Gamma^n}{W}, c}$.
By choosing $\conf{\Gamma^0}{V^0} = \conf{\Gamma_2}{W_2}$ and $\conf{\Gamma^n}{V^n} 
= \conf{\widehat{\Gamma_2}}{W'_2}$ we obtain that $\conf{\Gamma_2}{W_2 | W} \extTrans{\,} \conf{\widehat{\Gamma_2}}{W_2' | W}$.

The proof of the aforementioned statement is by induction on $n$.
\begin{enumerate}
\item If $n = 0$ then there is nothing to prove. 
\item Let $n > 0$. By inductive hypothesis we
assume that the statement is true for $n-1$. 
By Lemma \ref{lem:taus-Gamma} we know that 
$\Gamma^0 \leq \Gamma^{n-1} \leq \Gamma^n$. 
Let $d \neq c$;
if $\Gamma^0 \vdash d: \cbusy$, then 
$\Gamma^{n-1} \vdash d: \cbusy$ since $\Gamma^0 \leq \Gamma^{n-1}$. 
Conversely, if $\Gamma^{n-1} \vdash d: \cbusy$ then $\Gamma^{n-1} \leq 
\Gamma^n$ implies that $\Gamma^n \vdash d: \cbusy$, and by hypothesis 
we get that $\Gamma^{0} \vdash d: \cbusy$. 

Therefore we can apply the inductive 
hypothesis to obtain the sequence of transitions 
\[
\conf{\Gamma^0}{V^0 | W} \exttrans{\tau} \conf{\Gamma^1}{V^1 | W} 
\exttrans{\tau} \cdots \exttrans{\tau} \conf{\Gamma^{n-1}}{V^{n-1} | W}
\]
and infer that $\neg\isrcv{\conf{\Gamma^{n-1}}{W}, c}$.
Consider now the transition $\conf{\Gamma^{n-1}}{V^{n-1}} \exttrans{\tau} \conf{\Gamma^n}{V^n}$. 
There are different ways in which this extensional transition could have been inferred: 
\begin{itemize}
\item if this transition has been obtained by an application of Rule \rulename{TauExt} of Table \ref{tab:extensional}, 
then we have that $\conf{\Gamma^{n-1}}{V^{n-1}} \trans{\tau} {V^n}$, and $\Gamma^n = \gupd{\tau}{\Gamma^{n-1}}$. By Rule 
\rulename{TauPar} we also have that $\conf{\Gamma^{n-1}}{V^{n-1} | W} \trans{\tau} V^n | W$, 
which can now be translated in an extensional $\tau$-action 
$\conf{\Gamma^{n-1}}{V^{n-1} | W} \exttrans{\tau} \conf{\Gamma^n}{V^n | W}$ 
via an application of Rule \rulename{TauExt}.
\item if the transition has been obtained by an application of Rule \rulename{Shh} of Table \ref{tab:extensional}, 
then $\conf{\Gamma^{n-1}}{V^{n-1}} \trans{d!w} V^n$, and $\Gamma^n = \gupd{d!w}{\Gamma^{n-1}}$. 
Let us perform a case analysis on the channel $d$:

\begin{itemize}
\item
If $d = c$, then since $\neg\isrcv{\conf{\Gamma^{n-1}}{W}, c}$, Lemma \ref{lem:rcv-enabling}(1) ensures that we 
have the transition $\conf{\Gamma^{n-1}}{W} \trans{c?w} W$. Note also that 
now $\Gamma^n \vdash c: \cbusy$, so that it follows $\neg\isrcv{\conf{\Gamma^n}{W}, c}$. Now by applying 
Rule \rulename{Sync} to the transitions $\conf{\Gamma^{n-1}}{V^{n-1}} \trans{c!w} {V^n}$ 
and $\conf{\Gamma^{n-1}}{W} \trans{c?w} {W}$ we obtain $\conf{\Gamma^{n-1}}{V^{n-1} | W} 
\trans{d!v} V^n | W$. The latter can be converted into an extensional $\tau$-transition 
$\conf{\Gamma^{n-1}}{V^{n-1} | W} \exttrans{\tau} \conf{\Gamma^n}{V^n | W}$
using Rule \rulename{Shh} and the fact that $\Gamma^{n} = \gupd{c!w}{\Gamma^{n-1}}$.

\item It remains to check the case $d \neq c$. First note that if we have $\Gamma^{n-1} \vdash c: \cbusy$ then 
also $\Gamma^{n} \vdash c: \cbusy$ (since $\Gamma^{n-1} \leq \Gamma^n)$, 
so that $\neg\isrcv{\conf{\Gamma^{n}}{W}, c}$. On the other hand, if $\Gamma^{n-1} \vdash c: \cfree$, we 
can still prove that $\neg\isrcv{\conf{\Gamma^{n}}{W}, c}$ via an induction on the structure of $W$.\footnote{Intuitively, 
we just need to check that there are no unguarded receivers along channel $c$ appearing in $W$.}

Finally, note that since $\Gamma^n = \gupd{d!w}{\Gamma^{n-1}}$ implies 
that $\Gamma^n \vdash d : \cbusy$. By hypothesis we get that $\Gamma^0 \vdash d: \cbusy$, which 
leads to $\Gamma^{n-1} \vdash d: \cbusy$ (recalling that $\Gamma^0 \leq \Gamma^{n-1}$). Therefore 
we have that $\neg\isrcv{\conf{\Gamma^{n-1}}{W}, d}$, and by Lemma \ref{lem:rcv-enabling}(1) we 
obtain that $\conf{\Gamma^{n-1}}{W} \trans{d?v} W$. Now we can proceed as in the 
case $d = c$ to infer the extensional transition $\conf{\Gamma^{n-1}}{V^{n-1} | W} \exttrans{\tau} 
\conf{\Gamma^n}{V^n | W}$.

\end{itemize}
\end{itemize}

\end{enumerate}

\item Suppose now that $\isrcv{\conf{\Gamma_1}{W},c}$.  
By Lemma~\ref{lem:rcv-enabling}(\ref{rcv-enabling2})
the transition $\conf {\Gamma_1}{W} \trans{c?v} W'$ leads to $W' \neq
W$. Also, in this case we have that $\Gamma_1 \vdash c: \cfree$, which 
also gives $\Gamma_2 \vdash c: \cfree$ by Lemma \ref{lem:channel-exposure}. Since we have $\widehat{\Gamma_2} \vdash 
c: \cbusy$, it has to be the case that we can unfold the weak transition $\conf{\Gamma_2}{W_2} 
\extTrans{\,} \conf{\widehat{\Gamma_2}}{W'_2}$ as 
\[
\conf{\Gamma_2}{W_2} \extTrans{\,} \conf{\Gamma_2^{\mbox{pre}}}{W_2^{\mbox{pre}}} 
\exttrans{\tau} \conf{\Gamma_2^{\mbox{post}}}{W_2^{\mbox{post}}} \extTrans{\,} \conf{\widehat{\Gamma_2}}{W_2'}
\]
where $\Gamma_2^{\mbox{pre}} \vdash c: \cfree$ and $\Gamma_2^{\mbox{post}} \vdash c: \cbusy$. 
Note also that Lemma \ref{lem:taus-Gamma} ensures that, for any channel $d \neq c$,  
$\Gamma_2 \vdash d: \cbusy$ implies $\Gamma_2^{\mbox{pre}} \vdash d: \cbusy$, and 
$\Gamma_2^{\mbox{pre}} \vdash d: \cbusy$ implies $\widehat{\Gamma_2} \vdash d: \cbusy$, which 
by hypothesis leads to $\Gamma_2 \vdash d: \cbusy$. Similarly we can show that $\Gamma_2^{\mbox{post}} 
\vdash d: \cbusy$ if and only if $\widehat{\Gamma_2} \vdash d: \cbusy$. 
That is, $\Gamma_2, \Gamma_2'$ agree with $\Gamma_2^{\mbox{pre}}, 
\Gamma_2^{\mbox{post}}$ on the exposure state of each channel, respectively. 
Now, in a way similar to the first case, we can prove that we have the following transitions:
\begin{itemize}
\item $\conf{\Gamma_2}{W_2 | W} \extTrans{\,} \conf{\Gamma_2^{\mbox{pre}}}{W_2^{\mbox{pre}} | W}$,
\item $\conf{\Gamma_2^{\mbox{post}}}{W_2^{\mbox{post}} | W'}\extTrans{\,} \conf{\widehat{\Gamma_2}}{W'_2 | W'}$.
\end{itemize}
so that it remains to show that $\conf{\Gamma_2^{\mbox{pre}}}{W_2^{\mbox{pre}} | W} \exttrans{\tau} 
\conf{\Gamma_2^{\mbox{post}}}{W_2^{\mbox{post}} | W'}$. Note that, since 
$\Gamma_2^{\mbox{pre}} \vdash c: \cfree$ and $\Gamma_2^{\mbox{post}} \vdash c: \cbusy$, 
it has to be the case that the transition $\conf{\Gamma_2^{\mbox{pre}}}{W_2^{\mbox{pre}}} \exttrans{\tau} 
\conf{\Gamma_2^{\mbox{post}}}{W_2^{\mbox{post}}}$ has been induced by the intensional one 
$\conf{\Gamma_2^{\mbox{pre}}}{W_2^{\mbox{pre}}} \trans{c!w} W_2^{\mbox{post}}$, 
and $\Gamma_2^{\mbox{post}} = \gupd{c!w}{\Gamma_2^{\mbox{pre}}}$. 

Now note that, since $\conf{\Gamma_1}{W}{\trans{c?v}}W'$ we 
also have that $\conf{\Gamma_1}{W}{\trans{c!w}} W'$ by Lemma \ref{lem:rcv-enabling}(2). 
Finally, note that for any channel $c$, $\Gamma_1 \vdash d: \cbusy$ iff 
$\Gamma_2 \vdash d: \cbusy$ (as $\conf{\Gamma_1}{W_1} \approx \conf{\Gamma_2}{W_2}$) iff 
$\Gamma_2^{\mbox{pre}} \vdash d: \cbusy$. By Proposition \ref{prop:exposure-consistency} 
it follows that $\conf{\Gamma_1}{W}{\trans{c!w}} W'$ implies $\conf{\Gamma_2^{\mbox{pre}}}{W} 
\trans{c?w} W'$. We can now apply Rule \rulename{Sync} to such a transition, and the transition 
$\conf{\Gamma_2^{\mbox{pre}}}{W_2^{\mbox{pre}}} \trans{c!w} W_2^{\mbox{post}}$, to infer  
$\conf{\Gamma_2^{\mbox{pre}}}{W_2^{\mbox{pre}} | W} \trans{c!w} W_2^{\mbox{post}} | W'$. 
The last transitions induces the extensional action $\conf{\Gamma_2^{\mbox{post}}}{W_2^{\mbox{post}} | W'}
\exttrans{\tau} \conf{\Gamma_2^{\mbox{post}}}{W_2^{\mbox{post}} | W'}$, as we wanted to prove. 

We have built the sequence of transitions 
\[
\conf{\Gamma_2}{W_2 | W} \extTrans{\,} \conf{\Gamma_2^{\mbox{pre}}}{W_2^{\mbox{pre}} | W} 
\exttrans{\tau} \conf{\Gamma_2^{\mbox{post}}}{W_2^{\mbox{post}} | W'} \extTrans{\,} \conf{\widehat{\Gamma_2}}{W_2' | W'}
\]
which can be synthesised as $\conf{\Gamma_2}{W_2 | W} \extTrans{\,} \conf{\widehat{\Gamma_2}}{W_2' | W'}$, 
which is exactly the transition that we wanted to derive.

\end{enumerate}
\end{proof}

\begin{theorem}\MHc{[Soundness] $\confC_1 \approx \confC_2$ implies  $\confC_1 \simeq \confC_2$.} 
\end{theorem}
\begin{proof}
It suffices to prove that bisimilarity is reduction-closed, barb preserving
and contextual. 
\begin{description}
\item[Reduction Closure] Note that if $\confC_1 \red \confC_1'$, then 
we have two possible cases; either $\confC_1 \red_i \confC_1'$ or $\confC_1 \red_{\sigma} 
\confC_1'$. If $\confC_1 \red_{i} \confC_1'$  
then it is not difficult to see that $\confC_1 \exttrans{\tau} \confC_1'$
(see Remark~\ref{rem:operational_correspondence}). 
Similarly, if 
$\confC_1 \red_{\sigma} \confC_1'$  then $\confC_1 \exttrans{\sigma} \confC_1'$. 
Since $\confC_1 \approx \confC_2$, it follows that there exists $\confC_2'$ such that 
$\confC_2 \extTrans{} \confC_2'$ (respectively, $\confC_2 \extTrans{\sigma} \confC_2')$ with 
$\confC_1' \approx \confC_2'$. By Remark~\ref{rem:operational_correspondence} the last transition can be rewritten as 
a sequence of reductions $\confC_2 \red_i^{\ast} \confC_2'$ (respectively, $\confC_2 \red_i^{\ast}\red_{\sigma}\red_i^{\ast}
\confC_2'$), from which it follows $\confC_2 \red^{\ast} \confC_2'$,
\item[Barb Preservation] Let $\confC_1 = \conf{\Gamma_1}{W_1}$ and 
$\confC_2 = \conf{\Gamma_2}{W_2}$. 
Suppose that $\confC_1 \downarrow c$ for some channel $c$; by definition 
we have that  $\Gamma-1 \vdash c: \cbusy$. By Lemma 
\ref{lem:exp-equivalent} we also have that  $\Gamma_2 \vdash c: \cbusy$. 
This ensures that $\confC_2 \downarrow_c$, and more generally $\confC_2 \Downarrow_c$. 
\item[Contextuality] contextuality has already been proved as Theorem~\ref{thm:congruence}. 
\end{description}
\end{proof}

\subsection{Completeness}
\label{sec:completeness}
Having proved soundness, it remains to check whether our bisimulation proof 
technique is also complete with respect to reduction barbed congruence; that is, 
whenever we have $\conf{\Gamma_1}{W_1} \simeq \conf{\Gamma_2}{W_2}$, 
then there exists a bisimulation that contains the pair $(\conf{\Gamma_1}{W_1}, \conf{\Gamma_2}{W_2})$.
Unfortunately, this is not true for arbitrary configurations, 
as shown by the following Example:

\begin{example}
\label{ex:illformed}
 Let $\Gamma_1 \vdash c : \cbusy$,
  $\Gamma_2 \vdash c: \cfree$ and 
  consider the two configurations $\confC_1 = \conf{\Gamma_1}{ \crest{d}(0,\cdot).(\arcv d x {\nil}})$ 
   and $\confC_2 = \conf{\Gamma_2}{\bcastzeroc c v | \crest{d}(0,\cdot).(\arcv d x {\nil})}$. 
    Note that both configurations include an active receiver 
    placed along an  idle, restricted channel. The presence of such an active receiver is somewhat 
   problematic, as it does not allow the passage of time in both configurations, 
    according to our definition of timed reductions. Indeed, the reader can check that, 
    in the intensional semantics, no transition $\trans{\sigma}$ is defined for a configuration 
    of the form $\conf{\Gamma[d \mapsto (0,\cdot)]}{\arcv d x {P}}$; 
    consequently, $\sigma$-transitions are not allowed for the configuration 
    $\conf{\Gamma}{\crest{d}(0,v).(\arcv d x {P})}$ either. Similarly, weak $\sigma$-transitions 
    are note enabled in $\confC_2$.
    
    Now note that, since any occurrence of channel $d$ is restricted in both $\confC_1, \confC_2$, 
    we cannot enable the passage of time for them via the composition with a system term 
    $T$. That is, for any system term $T$, and configuration $\widehat{\confC_1}, \widehat{\confC_2}$, such that 
    $\confC_1 | T \red_i^{\ast} \widehat{\confC_1}, \;\confC_2 | T \red_{i}^{\ast} \widehat{\confC_2}$, we have that 
    $\widehat{\confC_1} \not\red_{\sigma}$ and $\widehat{\confC_2} \not\red{\sigma}$. 
    
    Now it is not difficult to show that $\confC_1 \simeq \confC_2$. At least informally, the only difference 
    between these two configurations lies in the exposure state of  channel $c$, and in the fact 
    that $\confC_2$ can broadcast along channel $c$. Such a broadcast ensures that the 
    strong barb at channel $c$, enabled in $\confC_1$, can be matched by a weak barb enabled 
    at $\confC_2$. On the other hand, the difference in the exposure state of 
    channel $c$ in $\confC_1, \confC_2$ could  
    be detected via a test $T$ which contains an exposure check $\expsd{c}$; however, this construct 
    requires the passage of time in order   to determine that channel $c$ is free (exposed) in 
    $\confC_1 | T$ (respectively, $\confC_2 | T$). But, as we have already noticed, time is not allowed to pass 
    in such configurations. Formally, to prove $\confC_1 \simeq \confC_2$ 
    it suffices to show that the relation 
    \begin{eqnarray*}\{&(\conf{\Delta}{\crest{d}(0,\cdot).(\arcv d x {\nil})}, 
    \conf{\Delta'}{\crest{d}(0,\cdot).( \bcastzeroc c v | \arcv d x {P}}) &|\\
    |& \Delta \vdash c: \cbusy, \Delta \vdash d: \cbusy \mbox{ iff } \Delta' \vdash d: \cbusy \mbox{ for } 
    d \neq c & \}
    \end{eqnarray*}
    is barb-preserving, reduction closed and contextual.
    
    Therefore we have shown that $\confC_1 \simeq \confC_2$; however, $\Gamma_1 \vdash c: \cfree$, 
    while $\Gamma_2 \vdash c: \cbusy$. Therefore, by Lemma \ref{lem:channel-exposure} it 
    also has to be $\confC_1 \not\approx \confC_2$.
\end{example}

\subsubsection{Well-formed systems}
\label{sec:wellformed}

The counterexample to completeness illustrated in Example \ref{ex:illformed} 
relies on the existence of configurations which do not let time pass. 
These can be built by placing an active receiver along an idle, 
restricted channel. 
However, such configurations are not interesting per se, as it is  
counter-intuitive to allow wireless stations to receive a value along a channel, 
when there is no value being transmitted. 

It is interesting, in fact, to ask ourselves if our proof 
methodology based on bisimulations is complete, if we were to restrict our focus to a setting where active receivers 
along idle channels were explicitly forbidden. These take the name of 
\emph{well-formed} configurations, and can be defined as below:

\begin{definition}[Well-formedness]
\label{def:wellformed}
The set of \MHc{well-formed configurations} $\text{WNets}$ is the least set such that 
\vspace{-1em}
\begin{eqnarray*}
\conf{\Gamma}{P} \in \text{WNets}&&\MHc{\mbox{for all processes}\, P}\\
\Gamma \vdash c: \cbusy &\text{implies}& \conf{\Gamma}{\arcv c x P} \in \text{WNets}\\
\conf{\Gamma}{W_1}, \conf{\Gamma}{W_2} \in \text{WNets} &\text{implies}& \conf{\Gamma}{W_1 | W_2} \in \text{WNets}\\
\conf{\Gamma[c \mapsto (n,v)]}{W} \in \text{WNets} &\text{implies}& \conf{\Gamma}{\nu c:(n,v).W} \in \text{WNets}
\end{eqnarray*}
\end{definition}

A configuration $\conf{\Gamma}{W}$ is \MHc{well-formed} if it 
does not contain any receiving station along an idle channel. 
Note that the configurations from Example \ref{ex:illformed} are not 
well-formed. 
Clearly, well-formed configurations are preserved at runtime.
%Note that the configuration used in Example~\ref{ex:nonWF} is not well-formed. 
\begin{lemma}
\label{lem:wf.preserved}
  Suppose $\confC$ is well-formed and $\confC \red \confC'$. Then
$\confC'$ is also well-formed. 
\end{lemma}
\begin{proof}
%It suffices to show that whenever $\confC \trans{\lambda} W'$, then 
%$\confC' = \conf{\gupd{\lambda}{\Gamma}}{W'}$ is well-formed. 
%Then the result follows from the definition of 
%$\confC \red \confC'$.
%To show that $\confC'$ is well-formed, it suffices to perform a Rule induction on 
%the derivation $\confC \trans{\lambda} W'$. 
See the Appendix, Page \pageref{proof:wf.preserved}.
%Suppose $\conf{\Gamma}{W} \trans{\lambda} W'$ and 
% $\conf{\Gamma}{W} $. Then by rule induction on 
% $\conf{\Gamma}{W} \trans{\lambda} W'$
%one can show that
% $\conf{\gupd{\lambda}{\Gamma}}{W'}$ is also well-formed.
%
%The result now follows by consideration of the three possible cases for 
%deriving $\confC \red \confC'$ in Definition~\ref{def:step}. 
%\qed
\end{proof}

The main property of well-formed systems is that they allow the passage of time, 
so long as all internal activity has ceased:
\begin{proposition}[Patience]\label{prop:noti2sigma}
Let $\confC$ be a well-formed configuration for which 
there is no $\confC'$ such that 
$\confC  \red_i\confC'$; then $\confC \red_{\sigma}\confC''$, for some configuration $\confC''$.
\end{proposition}
\MHf{Andrea: Remember we found a bug in your proof of this result}
\begin{proof}
%By structural induction on the definition of the set $\text{WNets}$ of 
%well-formed networks; this result relies on the fact 
%that we only allow guarded recursion in the language. 
Details 
for the most important cases are given in the Appendix; see Page \pageref{proof:noti2sigma}.
\end{proof}

However, Patience alone does not preclude the possibility of exhibiting a configuration 
in which time never passes. In fact, it only ensures the passage of time when 
 instantaneous reduction are not possible anymore. However, it could be the 
 case that a configuration $\confC$ enables an infinite sequence of 
 instantaneous reductions, and by maximal progress (Proposition \ref{prop:maximal-progress}) 
 the passage of time would be forbidden. 
 As we will prove presently, this phenomenon does not arise for 
 CCCP configurations; we recall in fact that, in recursive processes of the form $\fix X P$, 
 we require all free occurrences of the process variable $X$ in $P$ to be guarded by a 
 time-consuming construct. This limitation is sufficient to prevent the existence 
 of configurations which do not allow time to pass; further, it is also necessary, 
 as shown by the following example.
 
%
%As we will prove presently, infinite sequences 
%of instantaneous reductions cannot occur for configurations of CCCP; this 
%follows at once by allowing guarded recursion in our calculus. 
%The following example shows that it is relatively easy to exhibit an infinite sequence 
%of instantaneous reductions, if generally recursive processes were admitted. 
%the possibility that a configuration, and its derivatives, can always  is possible In fact, it would seem that restricting attention to well-formed configurations does not preclude the phenomenon 
%exhibited in Example~\ref{ex:nonWF} from occurring. Since our language for station code includes recursion
%the reader could argue that it is possible to write systems which can perform an infinite sequence of instantaneous 
%reductions; we first identify this systems formally, then we show that these cannot be obtained in our calculus. 

\begin{example}\label{ex:nonWT}\MHf{Andrea: please check}
  Suppose we  remove the constraint in the syntax  that process variables have to be guarded 
  by time-consuming constructs in fixed point processes.
  Let $W$ denote the code $\fix X {(\tau.X)}$. 
Then we have an infinite sequence of internal actions
\begin{align*}
  &\conf{\Gamma}{W} \red_i \confC_1 \red_i \ldots \confC_k \red_i
\end{align*}
Indeed one can show that if $\conf{\Gamma}{W} \red^\ast \confC'$ then 
$\confC' \red_i$. Maximal progress then ensures that $\confC' \notred_{\sigma}$.

%Taking up the discussion again from Example~\ref{ex:nonWF} this means that the testing
%environment $T_{c,v}$ will also not have the desired behaviour for this system. Once more
% $\conf{\Gamma}{(W | T_{c,v} )}$ can never produce a barb on
%  $\eureka$.
\end{example}

\begin{example}\label{ex:nonWT.noncomplete}
Again, suppose we remove the constraint on guarded 
recursion in the syntax of  CCCP. Then our bisimulation proof principle would not 
be complete; to see this, it is sufficient to consider the two 
configurations $\conf{\Gamma}{\fix X {(\tau.X)}}$ and 
$\conf{\Gamma'}{\fix X {(\tau.X)} | \bcastzeroc c v}$, where 
$\Gamma \vdash c: \cbusy$ and $\Gamma' \vdash c: \cfree$. By Lemma \ref{lem:channel-exposure} 
these two configurations are not bisimilar, as they differ in the exposure state of channel $c$. 
On the other hand, none of these two configurations allow the passage of time. 
As we have already argued in 
Example \ref{ex:illformed}, when the passage of time is not allowed in a configuration, it is not possible 
to provide a context that determines the exposure state of a channel. Then it is not difficult 
to show that $\conf{\Gamma}{\fix X {(\tau.X)}} \simeq \conf{\Gamma'}{\fix X {(\tau.X)} | \bcastzeroc c v}$. 
This can be done by simply showing that the relation 
\begin{eqnarray*}
{\mathcal S} &=& \{(\conf{\Delta}{\fix X{\tau.X}}, \conf{\Delta'}{\fix X {\tau.X} | \bcastzeroc c v}),
 (\conf{\Delta_c}{\fix X{\tau.X}}, \conf{\Delta_c'}{\fix X {\tau.X} | \sigma^{\delta_v}}) |\\
&& \Delta \vdash d: \cbusy \text{ if and only if } \Delta' \vdash d: \cbusy, d \neq c,\\
&&\Delta_c \vdash d: \cbusy \text{ if and only if } \Delta'_c \vdash d: \cbusy, \text{with } d \text{ arbitrary }\}
\end{eqnarray*}
is a bisimulation.
\end{example}

Let us state precisely what we mean when we say that 
infinite sequences of instantaneous reductions are not allowed in 
our calculus. In practice, we give a slightly stronger definition, requiring 
that the amount of instantaneous reductions that can be performed in 
sequence  by a configuration $\confC$ is bounded.

\begin{definition}[Well-timed configurations]
A configuration $\confC$  is \emph{well-timed}, \cite{MerMac12}, if 
there exists an upper bound $k \in \mathbb{N}$ 
such that whenever  $\confC \; (\red_i)^{h} \; \confC'$ for some 
$h \geq 0$, then $h \leq k$.
\end{definition}
Contrarily to well-formedness, which is
a simple syntactic constraint, \emph{well-timedness} means that the
designer of the network has to ensure that the code placed at the
station nodes can never lead to divergent behaviour. 
As we already argued, however, the constraint we have placed on 
the syntax of system terms that each recursive definition is weakly guarded in $P$, 
is sufficient to ensure well-timedness.
%However it is
%only with these configurations can the extensional actions be
%accurately implemented using system contexts. In particular a
%variation on the testing context $T_{c,v}$ defined above and used in
%Examples~\ref{ex:nonWF} and \ref{ex:nonWT} will form the basis in
%Proposition~\ref{prop:input.preservation} below, for detecting the input of values. 
One simple method for ensuring this is to only use recursive definitions $\fix X P$ 
where $X$ is weakly guarded in $P$; that is, every occurrence of $X$ is within an 
input, output or time delay prefix, or it is included within a branch of a matching construct.
These are exactly the conditions that we placed for recursion variables when defining our calculus. 
Thus, we would expect every configuration in our calculus to be well-timed.
% In order to 
%prove formally this statement, we need the following technical result:

\begin{proposition}
\label{cor:wf2wt}
Any configuration $\conf{\Gamma}{W}$
%, where $W$ is closed, 
is well-timed.
%In particular, any well-formed configuration is also well-timed.
\end{proposition}

\begin{proof}
See the Appendix, Page \pageref{proof:wf2wt}.
\end{proof}

%\begin{corollary}
%\label{cor:timepasses}
%Let $\conf{\Gamma}{W}$ be well timed. Then 
%whenever $\conf{\Gamma}{W} \red^{\ast} \conf{\Gamma'}{W'}$ it follows that 
%$\conf{\Gamma'}{W'} \red^{\ast} \red_{\sigma} \red^{\ast} \conf{\Gamma''}{W''}$ 
%for some $\Gamma'', W''$.
%\end{corollary}
%
%\begin{proof}
%It suffices to show that if $\conf{\Gamma}{W}$ is well-formed then 
%$\conf{\Gamma}{W} \red^{\ast} \red_{\sigma} \red^{\ast} \conf{\Gamma'}{W'}$
%for some $\Gamma', W'$. Then the result follows since well-formed 
%configurations are closed under reductions, Lemma \ref{lem:wf.preserved}. 
%
%Suppose then that $\conf{\Gamma}{W}$ is well-formed. By Corollary \ref{cor:wf2wt} 
%it is also well-timed. By definition, there exists an upper bound $k$ such that 
%whenever $\conf{\Gamma}{W} \red^{h}_{i} \conf{\Gamma''}{W''}$ then 
%$h \leq k$. In particular, there exists a finite index $h$ such that 
%$\conf{\Gamma}{W} \red^{h}_{i} \conf{\Gamma''}{W''}$ and 
%$\conf{\Gamma''}{W''} \not\red_i$. Since $\conf{\Gamma}{W}$ is 
%well-formed, then also $\conf{\Gamma''}{W''}$ is well-formed. 
%By Patience, Proposition \ref{prop:noti2sigma}, it follows that 
%$\conf{\Gamma''}{W''} \red_{\sigma} \conf{\Gamma'}{W'}$ for some 
%$\Gamma', W'$.
%\end{proof}

Next we prove a very useful result for well-defined
configurations; the proof emphasises the roles of well-formedness and
well-timedness in the configurations being tested.
\begin{proposition}\label{prop:exposure}
  Suppose $\conf {\Gamma_1} W_1 \simeq \conf{\Gamma_2} W_2$, where
  both are well-formed.  Then $\Gamma_1 \vdash c:\cfree$ implies
  $\Gamma_2 \vdash c:\cfree$.
\end{proposition}
\begin{proof}
Let $\conf{\Gamma_1}{W_1} \simeq \conf{\Gamma_2}{W_2}$ 
and suppose $\Gamma_1 \vdash c : \cfree$ for some channel $c$. 
Consider the testing \MHc{code:}
$$T =  \matchb{\expsd{c}}{\nil} 
                                { \bcastzeroc  {\eureka} \arb } 
                                 \enspace  
$$
\MHc{From the  definition of $\simeq$ we know   that} 
$\conf{\Gamma_1}{W_1 | T} \simeq \conf{\Gamma_2}{W_2 |T}$. 
Since $\conf{\Gamma_1}{W_1}$ is well-timed, by definition 
there is a configuration $\confC$ such that 
 $\conf{\Gamma_1}{W_1} \,\red_i^\ast \confC $ and 
$\confC \not{\red}_i$. Because  $\conf{\Gamma_1}{W_1}$ is well-formed
so is $\confC$. By Proposition~\ref{prop:noti2sigma} 
there is a configuration $\confC'$ such that $\confC \red_{\sigma} \confC'$.
Let $\confC' = \conf{\Gamma'}{W'}$, for some $\Gamma'$ and $W'$. Now, if we define $\confC'' = \conf{\gupd{\eureka!\arb}{\Gamma'}}{W'}$ 
and $T' = \sigma.\bcastzeroc \eureka \arb$, it is easy to see that there exists
a sequence 
of reductions of the following shape: 
%
%It follows, by the existence of this $\confC'$, that 
%$\conf{\Gamma_1}{W_1 | T} \Downarrow_{\eureka}$; 
%this is because, if we let $\confC' = \conf{\Gamma'}{W'}$ 
%and $T' = \sigma.\bcastzeroc {\eureka}{\arb}$, 
%we have the sequence of reductions 
\[
\conf{\Gamma_1}{W_1 | T} \red_i \conf{\Gamma_1}{W_1 | T'} 
\red^{\ast}_i \confC | T' \red_{\sigma} \confC' | \bcastzeroc {\eureka}{\arb} 
\red_i \confC'' | \sigma^{\delta_{\arb}} 
\]
where $\confC'' | \sigma^{\delta_{\arb}}\downarrow_{\eureka}$. By definition this implies that $\conf{\Gamma_1}{W_1 | T} \Downarrow_{\eureka}$.

Note that the existence of the sequence of reductions above relies on the fact that 
$\conf{\Gamma_1}{W_1}$ is well-timed.The timed transition $\confC | T' \red_{\sigma} \confC' | \bcastzeroc \eureka \arb$ 
in such a sequence
is derived from the timed transitions performed by their components; if $\confC$ were not able to perform a $\sigma$-transition, 
in fact, we would have not been able to derive the timed reduction for the overall configuration $\confC | T'$.

Since $\conf{\Gamma_1}{W_1 | T} \simeq \conf{\Gamma_2}{W_2 | T}$ we also 
have that $\conf{\Gamma_2}{W_2 | T} \Downarrow_{\eureka}$. 
%By maximal progress, Proposition \ref{prop:maximal-progress}, 
This is only possible if 
$$
\conf{\Gamma_2}{W_2 | T} \,\,\red_i^{\ast}\,\, \conf{\Gamma_2'}{W_2' | T' } 
\,\,\red_i^{\ast} \red_{\sigma} \red_i^{\ast}\,\, \conf{\Gamma_2''}{W_2'' | \sigma^{\delta_\arb}} 
$$
where $\Gamma_2'$ is a channel environment such that $\Gamma_2' \vdash c: \cfree$. 
\MHc{From Lemma \ref{lem:taus-Gamma} (recall that $\tau$-extensional actions coincide 
with instantaneous reductions) we get the required  $\Gamma_2 \vdash c: \cfree$.}
\end{proof}

We remark once again that restricting our attention to well-formed configurations is crucial in order to ensure 
the validity of Proposition \ref{prop:exposure}. In fact, in Example \ref{ex:illformed} 
we have already provided an example of two (ill-formed) configurations which are reduction 
barbed congruent, but which differ in  the exposure state of a channel.
%provides a counterexample Proposition \ref{prop:exposure} does not hold for ill-formed
% configurations, as we have already shown in Example \ref{ex:illformed}.
%\begin{example}
%\label{ex:illformed}
% Let $\Gamma_1 \vdash c : \cbusy$,
%  $\Gamma_1 \vdash d: \cfree$ and $\Gamma_2 \vdash c,d: \cfree$ and
%  consider the two configurations $\confC_1 = \conf{\Gamma_1}{\nil |
%    \arcv d x {P}}$ and $\confC_2 = \conf{\Gamma_2}{\bcastzeroc c v |
%    \arcv d x {P}}$, neither of which are well-formed; 
%   nor do they let time pass,  $\confC_i
%  \not\red_\sigma $. As a consequence  $\confC_1 \simeq \confC_2$.
%\end{example}

Another important property that we will need from well-formed configurations 
concerns the definition of reduction barbed congruence itself; 
the reduction closure property which we used to define $\simeq$ can be 
strengthened by requiring instantaneous reductions to be matched by sequences 
of instantaneous reductions, and timed reductions to be matched by timed reductions, 
possibly preceded and followed by sequences of instantaneous ones. 
To prove this property we will need the following technical result, 
which will also be used later: 

\begin{lemma}\label{lem:fresh}
  Suppose  
$\conf {\Gamma_1} {W_1 | T } 
\simeq \conf {\Gamma_2} {W_2 | T } $ where each channel occurring free in 
$T$ does not occur free in $W_1$, nor in $W_2$ and is idle in both $\Gamma_1$ and 
$\Gamma_2$; then  
 $\conf {\Gamma_1} {W_1} \simeq \conf {\Gamma_2} {W_2} $.
\end{lemma}
\begin{proof}
See the Appendix, Page \pageref{proof:fresh}, for an outline.
\end{proof}

%
%Proposition \ref{prop:exposure} can be exploited to 
%give a simpler characterisation of reduction barbed congruence, at least 
%when dealing with well-formed configurations. We have already observed, 
%in Example \ref{ex:time}, that the passage of time can be observed. In other 
%words, we can provide a context which detects whether a (weak) timed reduction 
%has been performed. We expect that the ability of contexts to distinguish timed 
%reductions should be reflected in the behaviour of contextually equivalent configurations.

\begin{proposition}
\label{prop:reduction.isolation}
Let $\conf{\Gamma_1}{W_1}, \conf{\Gamma_2}{W_2}$ be two well-formed configurations 
such that $\conf{\Gamma_1}{W_1} \simeq \conf{\Gamma_2}{W_2}$. 
Then 
\begin{enumerate}[label=(\roman*)]
\item whenever $\conf{\Gamma_1}{W_1} \red_i \conf{\Gamma_1'}{W_1'}$ there exists 
a configuration $\conf{\Gamma_2'}{W_2'}$ such that $\conf{\Gamma_2}{W_2} \red_i^{\ast} 
\conf{\Gamma_2'}{W_2'}$, and $\conf{\Gamma_1'}{W_1'} \simeq \conf{\Gamma_2'}{W_2'}$,
\item whenever $\conf{\Gamma_1}{W_1} \red_\sigma \conf{\Gamma_1'}{W_1'}$ there 
exists a configuration $\conf{\Gamma_2'}{W_2'}$ such that $\conf{\Gamma_2}{W_2} 
\red_i^{\ast}\red_{\sigma}\red_{i}^{\ast} \conf{\Gamma_2'}{W_2'}$, and 
$\conf{\Gamma_1'}{W_1'} \simeq \conf{\Gamma_2'}{W_2'}$.
\end{enumerate}

\end{proposition}
\begin{proof}
See the Appendix, Page \pageref{proof:reduction.isolation}.
\end{proof}

\subsubsection{Proving Completeness}
\label{sec:completeness.proof}

We are now in the position to prove that, for well-formed configurations, 
our proof methodology is also complete. Given two well-formed configurations $\confC_1 \simeq \confC_2$, 
there exists a bisimulation $\Ss$ such that $\confC_1 \;\Ss\; \confC_2$. 

To prove completeness, we show that reduction barbed congruence is a bisimulation. 
That is, we need to show that for any extensional action $\alpha$, if $\confC_1 \simeq \confC_2$ and 
$\confC_1 \exttrans{\alpha} \confC_1'$, then there exists $\confC_2'$ such that 
$\confC_2 \extTrans{\hat{\alpha}} \confC_2'$ and $\confC_1' \simeq \confC_2'$. 
The special cases $\alpha = \tau$ and $\alpha = \sigma$ follow as a direct 
consequence of Proposition \ref{prop:reduction.isolation}. However, 
we state the results for the sake of consistency.
%
%To prove this result, we provide for each 
%\MHc{extensional} action $\alpha$, a distinguishing context $T_\alpha$ which 
%is able to test whether a configuration can perform the aforementioned (weak) extensional action. 
%\MHc{For some particular $\alpha$ the distinguishing contexts will only work for well-formed configurations.} 
%First we show the case regarding extensional $\tau$-actions. 
%
%Preservation of $\tau$ and $\sigma$-actions follows immediately from Proposition \ref{prop:reduction.isolation}. 
%However, we include their statement here for the sake of completeness.
%%
%The following results says that reduction barbed congruence
%is preserved by extensional $\tau$-actions. 
\begin{proposition}[Preserving extensional $\tau$s]\label{prop:tau}
  Suppose  $\conf {\Gamma_1} {W_1} \simeq \conf {\Gamma_2} {W_2} $ and 
      $\conf {\Gamma_1} {W_1} \exttrans{\tau} \conf {\Gamma'_1} {W'_1}$.
Then $\conf {\Gamma_2} {W_2} \extTrans{} \conf {\Gamma'_2} {W'_2}$ such that
 $\conf {\Gamma'_1} {W'_1} \simeq \conf {\Gamma'_2} {W'_2} $. \hfill\qed
\end{proposition}
\begin{proposition}[Preserving extensional $\sigma$s]\label{prop:delay.preservation}
  Suppose $\conf {\Gamma_1} {W_1} \simeq \conf {\Gamma_2} {W_2} $. Then 
$\conf {\Gamma_1} {W_1} \exttrans{\sigma} \conf {\Gamma'_1} {W'_1} $ implies 
$\conf {\Gamma_2} {W_2}  \extTrans{\sigma}  \conf {\Gamma'_2} {W'_2} $ such that 
$\conf {{\Gamma'_1}} {W'_1} \simeq \conf {\Gamma'_2} {W'_2} $. \hfill\qed
\end{proposition}

Let us turn our attention to the remaining cases $\alpha \in \{c?v, \iota(c), \gamma(c,v)\}$. 
For each of them we define a distinguishing context $T_{\alpha}$; these are defined so 
that, given a well-formed configuration $\confC$, $\confC \extTrans{\hat{\alpha}} \confC'$ 
if and only if $\confC | T_{\alpha} \red^{\ast} \confC' | T^{\checkmark}_{\alpha}$, where $T^{\checkmark}_{\alpha}$ 
is uniquely determined by the action $\alpha$. Intuitively, the latter corresponds to the first state 
reached by the testing component when it has detected that the configuration 
$\confC$ has performed a weak $\alpha$-action; the system $T^{\checkmark}_{\alpha}$ is 
called the successful state for the action $\alpha$.

The tests $T_{\alpha}$ are defined below; here we assume that $\eureka, \fail$ are fresh channels, 
while $\delta_\arb = \delta_\no = 1$.

\begin{eqnarray*}
T_{\gamma(c,v)} &\deff& \nu d{:}(0, \cdot).(\arcv c x
{(\matchb{x{=}v}{\bcastzeroc{d}{\arb}}{\nil}) + \bcastzeroc \fail \no} |
\sigma^2.\matchb{\expsd{d}}{\bcastzeroc {\eureka}{\arb}}{\nil})\\
T_{c?v} &\deff &  (\bcastc c v {\bcastzeroc {\eureka}{\arb}} + \bcastzeroc \fail \no)\\
T_{\iota(c)} &\deff & (\matchb{\expsd{c}}{\nil}{\bcastzeroc {\eureka}{\arb}}) + \bcastzeroc {\fail}{\no}.
\end{eqnarray*}

We also list their respective successful states $T^{\checkmark}_{\alpha}$:
\begin{eqnarray*}
T^{\checkmark}_{\gamma(c,v)} &\deff & \nu d{:}(0, \cdot).(\sigma.\bcastzeroc{d}{\arb}{\nil} |
\sigma.\matchb{\expsd{d}}{\bcastzeroc {\eureka}{\arb}}{\nil})\\
T^{\checkmark}_{c?v} &\deff &  (\sigma^{\delta_v}.\bcastzeroc {\eureka}{\arb})\\
T^{\checkmark}_{\iota(c)} &\deff & \sigma.{\bcastzeroc {\eureka}{\arb}}
\end{eqnarray*}

\MHc{As an example we consider in detail the behaviour of the testing context $T_{\gamma(c,v)}$.}
%\[
% \nu d{:}(0, \cdot).(\arcv c x
%{(\matchb{x{=}v}{\bcastzeroc{d}{\arb}}{\nil}) + \bcastzeroc \fail \no} |
%\sigma^2.\matchb{\expsd{d}}{\bcastzeroc {\eureka}{\arb}}{\nil} | \sigma.\bcastzeroc {\alt}{\arb})
%\]
%where $\eureka, \alt$, $\fail$ are fresh channels and $\arb$ 
%is a message such that $\delta_{\arb}=1$.
This is designed to detect whether a configuration $\conf{\Gamma}{W}$ has performed 
a weak $\gamma(c,v)$-action. 
Let us discuss informally how the testing context 
$T_{\gamma(c,v)}$ operates. 
The fresh channels $\eureka$, $\fail$ 
play a different role: $\fail$ ensures that the 
reception along channel $c$ has finished, while $\eureka$ guarantees
that the received values is actually $v$. 

We provide a possible evolution of the testing contexts $T_\gamma(c,v)$ 
when running in a channel environment $\Gamma$ such that $\Gamma(c) = (1,v)$, 
and then we discuss how it works.
\[
\begin{array}{clcl}
&\multicolumn{3}{l}{\conf {\Gamma} {T_{\gamma(c,v)}}}\\
\red_{\sigma}& \conf {\Gamma_1}{T_1} &=& \conf{\Gamma_1'}{\nu d{:}(0, \cdot).(({\matchb{v{=}v}{\bcastzeroc{d}{\arb}}{\nil}}) + \bcastzeroc \fail \no  |\\
&&&| \sigma.\matchb{\expsd{d}}{\bcastzeroc {\eureka}{\arb}}{\nil})}\\
\red_{i} & \conf{\Gamma^{\checkmark}}{T^{\checkmark}} &= & \conf{\Gamma_2}{\nu d{:}(0, \cdot).({\sigma.{\bcastzeroc{d}{\arb}}} |
\sigma.\matchb{\expsd{d}}{\bcastzeroc {\eureka}{\arb}}{\nil})}\\
%T_3 &=& \nu d{:}(0, -).({\sigma.{\bcastzeroc{d}{\arb}}} |
%\sigma.\matchb{\expsd{d}}{\bcastzeroc {\eureka}{\arb}}{\nil} %)\\ 
%T_3' &=& \nu d{:}(0, -).({\matchb{v=v}{\bcastzeroc{d}{\arb}}{\nil}} |
%\sigma.\matchb{\expsd{d}}{\bcastzeroc {\eureka}{\arb}}{\nil} | \bcastzeroc {\alt}{\arb} )
%\red_{i}&\conf{\Gamma_2}{T_2} &=& \conf{\Gamma_2}{\nu d{:}(0, \cdot).({\sigma.{\bcastzeroc{d}{\arb}}} |
%\sigma.\matchb{\expsd{d}}{\bcastzeroc {\eureka}{\arb}}{\nil})}\\
\red_{\sigma}&\conf{\Gamma_3}{T_3} &=& \conf{\Gamma_3}{\nu d{:}(0, \cdot).({{\bcastzeroc{d}{\arb}}} |
\matchb{\expsd{d}}{\bcastzeroc {\eureka}{\arb}}{\nil})}\\
\red_{i}&\conf{\Gamma_4}{T_4} &=& \conf{\Gamma_4}{\nu d{:}(1, \arb).({\sigma} |
\matchb{\expsd{d}}{\bcastzeroc {\eureka}{\arb}}{\nil})}\\
\red_{i}&\conf{\Gamma_5}{T_5} &=& \conf{\Gamma_5}{\nu d{:}(1, \arb).({\sigma} |
\sigma.\bcastzeroc {\eureka}{\arb})}\\
\red_{\sigma}& \conf{\Gamma_6}{T_6} &=& \conf{\Gamma_6}{\nu d{:}(0, \cdot).(\nil |
\bcastzeroc {\eureka}{\arb})}
\end{array}
\]

Initially a configuration of the form $\conf{\Gamma}{W | T_{\gamma(c,v)}}$ 
has a weak barb at channel $\fail$. Further, 
the testing component has an active receiver over channel $c$; note that 
the configuration $\conf{\Gamma}{W | T_{\gamma(c,v)}}$ 
is well-formed only if $\Gamma \vdash c: \cbusy$. 
If $\conf{\Gamma}{W | T_{\gamma(c,v)}} \extTrans{\gamma(c,v)} \conf{\Gamma_1}{W'}$, that is if 
$\Gamma(c) = (1,v)$, then after time passes the reception along channel $c$ in the 
testing component $T_{\gamma(c,v)}$ terminates. Formally, we have the sequence of 
reductions $\conf{\Gamma}{W | T_{\gamma(c,v)}} \red^{\ast}_i \red_{\sigma} \red_i^{\ast} 
\conf{\Gamma_1}{W' | T_1}$. Note that the component $T_1$ compares the received value along channel 
$c$ with $v$; this test can only succeed, and as a consequence we obtain a further instantaneous reduction 
$\conf{\Gamma_1}{W' | T_1} \red_i \conf{\Gamma^{\checkmark}}{W' |
  T^{\checkmark}}$; In  practice here 
we have $\Gamma^{\checkmark} = \Gamma_1$).  At this point 
we have detected that the configuration $\conf{\Gamma_1}{W_1}$ has performed the weak $\gamma(c,v)$-action, 
ending in $\conf{\Gamma_1}{W'}$. The rest of the computation is already determined, at least for the part 
concerning the testing component $T_1$, and leads $\conf{\Gamma^{\checkmark}}{W' | T^{\checkmark}}$ to 
output a barb on $\eureka$; further, in this configuration it is not possible to output a barb on $\fail$ anymore.
%When channel $c$ is released, i.e. when the tested configuration $\conf{\Gamma}{W}$ terminates 
%the transmission along channel $c$, the testing component checks if the delivered value along channel $c$ 
%is $v$; note that this corresponds to check whether the configuration $\conf{\Gamma}{W}$ performed 
%a weak action $\conf{\Gamma}{W} \extTrans{\gamma(c,v)} \conf{\Gamma'}{W'}$. 
%If this is the case, the (weak) barb along channel $\fail$ disappears; 
%note that the state of the tested component now is exactly $T^{\checkmark}_{\gamma(c,v)}$, meaning that the test detected 
%that $\conf{Gamma}{W}$ has performed a weak $\gamma(c,v)$ action. Further, the 
%configuration $\conf{\Gamma'}{W'}{T^{\checkmark}_{\gamma(c,v)}}$ is now equipped with a weak 
%barb on the idle channel $\eureka$. 

To see why this is true, note that in $\conf{\Gamma^{\checkmark}}{W' | T^{\checkmark}}$ the testing 
component $T^{\checkmark}$ is waiting for time to pass, before broadcasting value 
$\arb$ along a restricted channel $d$. Formally, we have the sequence of reductions 
$\conf{\Gamma^{\checkmark}}{W' | T^{\checkmark}} \red_i^{\ast} \red_{\sigma} 
\conf{\Gamma_3}{W_3 | T_3} \red_{i} \conf{\Gamma_4}{W_4 | T_4}$, where $\conf{\Gamma^{\checkmark}}{W'} 
\red_i^{\ast} \conf{\Gamma_2}{W_2}$ and $W_3 = W_4$ (note that each instantaneous reduction performed by 
the tested component does not affect the test at this point). 

Finally, in $\conf{\Gamma_4}{W_4 | T_4}$ the test checks whether the restricted channel $d$ is exposed. 
As this channel is effectively restricted in $T_4$, the test can only succeed, leading to 
$\conf{\Gamma_4}{W_4 | T_4} \red_i \conf{\Gamma_5}{W_5 | T_5}$, where $\Gamma_5 = \Gamma_4$ and 
$W_5 = W_4$. At this point we can let time pass, via a sequence of reductions of the form 
$\conf{\Gamma_5}{W_5 | T_5} \red_i^{\ast}\red_{\sigma}\red_i^{\ast} \conf{\Gamma_6}{W_6 | T_6}$. 
Now it is trivial to see that this configuration has a barb on $\eureka$.

Note that  in the computation of $\conf{\Gamma}{W | T_{\gamma(c,v)}}$ discussed 
above, there are two crucial checks that lead to enabling a barb over channel $\eureka$: 
\begin{itemize}
\item The received value is exactly $v$, 
\item The check that a broadcast along the restricted channel $d$ is performed after two time instants. 
Since the broadcast along channel $d$ is performed only one time instant after value $v$ has been 
delivered, this check ensures that  such a value has been actually delivered after one time instant.
\end{itemize}

%
%a broadcast along a restricted channel $d$ is fired.
%after channel $c$ has been released, and only in the case that 
%the received value along channel $c$ matches value $v$.
%At this point, a second parallel component checks whether 
% channel $d$ is exposed; in this case it broadcasts over channel 
%$\eureka$. Note that the exposure check is passed 
%only if a broadcast along channel $d$ has been fired in the third 
%time instant of a computation of $\conf{\Gamma}{W | T_{\gamma(c,v)}}$; 
%this is possible only if channel $c$ has been 
%released in the second instant of time of such a computation, and 
%the received value is $v$; \MHc{that is, $\Gamma(c) = (1,v)$}.

\begin{proposition}[Detecting Inputs]
\label{prop:input.detection}
For any well-formed configuration $\conf{\Gamma}{W}$ we have 
that $\conf{\Gamma}{W} \extTrans{c?v} \conf{\Gamma'}{W'}$ if and only if 
$\conf{\Gamma}{W | T_{c?v} } \red_{i}^{\ast} \conf{\Gamma'}{W' | T^{\checkmark}_{c?v}}$.
\end{proposition}

\begin{proof}
See the Appendix, Page \pageref{proof:input.detection}.
\end{proof}

\begin{proposition}[Detecting Exposure Checks]
\label{prop:exp.detection}
For any well-formed configuration $\conf{\Gamma}{W}$ we have 
that $\conf{\Gamma}{W} \extTrans{\iota(c)} \conf{\Gamma'}{W'}$ if and only if 
$\conf{\Gamma}{W | T_{\iota(c)} } \red_{i}^{\ast} \conf{\Gamma'}{W' | T^{\checkmark}_{\iota(c)}}$.
\end{proposition}

\begin{proof}
See the Appendix, Page \pageref{proof:exp.detection}.
\end{proof}

\begin{proposition}[Detecting Delivery of Values]
\label{prop:delivery.detection}
For any well-formed configuration $\conf{\Gamma}{W}$ we have 
that $\conf{\Gamma}{W} \extTrans{\gamma(c,v)} \conf{\Gamma'}{W'}$ if and only if 
$\conf{\Gamma}{W | T_{\gamma(c,v)} } \red_{i}^{\ast}\red{\sigma}\red_i^{\ast} \conf{\Gamma'}{W' | T^{\checkmark}_{\gamma(c,v)}}$.
\end{proposition}

\begin{proof}
See the Appendix, Page \pageref{proof:delivery.detection}.
\end{proof}

Note that  in Propositions \ref{prop:input.detection}, \ref{prop:exp.detection} and \ref{prop:delivery.detection}, 
we emphasized whether the reductions needed to reach the successful configuration $\conf{\Gamma}{W' | T^{\checkmark}_{\alpha}}$ 
from $\conf{\Gamma}{W | T_{\alpha}}$ are instantaneous or timed.

We have stated all the results needed to prove completeness.
\begin{theorem}[Completeness]
On \MHc{well-formed} configurations, reduction barbed congruence 
implies bisimilarity. 
\end{theorem}
\begin{proof}
It is sufficient to show that the relation 
\[
\Ssr \deff \{ \big( \conf {\Gamma_1}{W_1} \, , \, \conf {\Gamma_2}{W_2}\big): \q
\conf {\Gamma_1}{W_1} \simeq \conf {\Gamma_2}{W_2}\}
\]
is a  bisimulation. 
To do so, suppose that $\conf{\Gamma_1}{W_1} \red \conf{\Gamma_2}{W_2}$, 
and that $\conf{\Gamma_1}{W_1} \simeq \conf{\Gamma_2}{W_2}$. 
If $\alpha = \tau$ or $\alpha = \sigma$, the result follows directly from propositions  
\ref{prop:tau} and \ref{prop:delay.preservation}, respectively. 

Now suppose that $\alpha = \gamma(c,v)$ for some channel $c$ and value $v$. 
%The details for the cases $\alpha = c?v$ and $\alpha = \iota(c)$ are similar. 
Let  $\conf{\Gamma_1}{W_1} \exttrans{\gamma(c,v)} \conf{\Gamma_1'}{W_1'}$; 
by Proposition \ref{prop:delivery.detection} it follows 
that $\conf{\Gamma_1}{W_1 | T_{\gamma(c,v)}} \red_{i}^{\ast}\red_{\sigma}\red_i^{\ast} 
\conf{\Gamma_1'}{W_1' | T^{\checkmark}_{\gamma(c,v)}}$.
By the contextuality of reduction barbed congruence, and by Proposition \ref{prop:reduction.isolation}, 
it follows that $\conf{\Gamma_1}{W_2 | T_{\gamma(c,v)}} \red_{i}^{\ast} \red_{\sigma} \red_{i}^{\ast} 
\confC_2$ for some $\confC_2$ such that $\conf{\Gamma_1'}{W_1' | T^{\checkmark}_{\gamma(c,v)}} 
\simeq \confC_2$. Let $\confC_2 = \conf{\Gamma_2'}{\widehat{W_2}}$; note that $\Gamma_1' \vdash \eureka : \cfree$ (recall that we assumed that 
$\eureka$ is a fresh channel), so that by Proposition \ref{prop:exposure} it follows 
that $\Gamma_2' \vdash \eureka : \cfree$. Further, $\conf{\Gamma_1'}{W_1' | T^{\checkmark}_{\gamma(c,v)}} \Downarrow_{\eureka}$ 
and $\conf{\Gamma_1'}{W_1'}{T^{\checkmark}_{\gamma(c,v)}} \not\Downarrow_{\fail}$; therefore, 
we also have that $\conf{\Gamma'_2}{\hat{W_2}} \Downarrow_{\eureka}$ and 
$\conf{\Gamma'_2}{\hat{W_2}} \not\Downarrow_{\fail}$. Now, by inspecting all the possible evolutions 
of the configuration $\conf{\Gamma_2}{W_2 | T_{\gamma(c,v)}}$ it follows that the sequence of reductions 
$\conf{\Gamma_1}{W_2 | T_{\gamma(c,v)}} \red_{i}^{\ast} \red_{\sigma} \red_{i}^{\ast} \conf{\Gamma'_2}{\hat{W_2}}$, 
where $\Gamma_2' \vdash \eureka : \cfree$, $\conf{\Gamma'_2}{\hat{W_2}} \Downarrow_{\eureka}$ and 
$\conf{\Gamma'_2}{\hat{W_2}} \not\Downarrow_{\fail}$, is possible only if 
$\hat{W_2} = W_2' | T^{\checkmark}_{\gamma(c,v)}$. Consequently, 
Proposition \ref{prop:delivery.detection} ensures that $\conf{\Gamma_2}{W_2} 
\extTrans{\gamma(c,v)} \conf{\Gamma_2'}{W_2'}$.

We also need to show that $\conf{\Gamma_1'}{W_1'} \simeq \conf{\Gamma_2'}{W_2'}$; 
but this follows immediately from Lemma \ref{lem:fresh} and the fact that 
$\conf{\Gamma_1'}{W_1' | T^{\checkmark}_{\gamma(c,v)}} \simeq \conf{\Gamma_2'}{W_2' | T^{\checkmark}_{\gamma(c,v)}}$.

It remains to check the cases $\alpha = c?v$ and $\alpha = \iota(c)$; these can be 
proved analogously to the previous case, using proposition \ref{prop:input.detection} and 
\ref{prop:exp.detection}, respectively, in lieu of Proposition \ref{prop:delivery.detection}.
\end{proof}

\section{Applications }
\label{sec:apps}

In this section, we show how our calculus CCCP can be 
used to model different interesting behaviours 
which arise at the MAC sub-layer \cite{macsurvey} 
of wireless networks. 
\MHc{
Then, 
we exploit our bisimulation proof technique 
to provide examples of behaviourally equivalent networks.
In particular we give some examples comparing the behaviour of
routing protocols and \emph{Time Division Multiplexing}.

We start with some simple examples. The first show that 
stations which do not transmit on unrestricted channels
can not be detected. To this end we use $\fsn{W}$ to denote the
set of unrestricted channel names in the code $W$ which have 
transmission occurrences. Formally $\fsn{W}$ is defined inductively on 
(a possibly open system term) $W$ 
as the least set such that 
\begin{eqnarray*}
\fsn{\nil} &=& \fsn{X} = \emptyset\\
\fsn{\bcast c v P} &=& \{c\} \cup \fsn{P}\\
\fsn{\tau.P} &=& \fsn{\sigma.P} = \fsn{\arcv c x P} = \fsn{\fix X P} = \fsn{P}\\
\fsn{P+Q} &=& \fsn{\matchb b P Q} = \fsn{\rcvtimec c x P Q} = \fsn{P} \cup \fsn{Q}\\
&&\\
\fsn{W_1 | W_2} &=& \fsn{W_1} \cup \fsn{W_2}\\
\fsn{\nu c:(n,v).W} &=& \fsn{W_1} \setminus \{c\}
\end{eqnarray*}

\begin{example}[Unobservable systems]
Consider a wireless system in which no station  can  broadcast on any free channel. 
Intuitively none of its behaviour should be observable. In CCCP this means
that the system should be behaviourally equivalent to the \emph{empty}
system $\nil$.

Formally consider the configuration $\conf{\Gamma}{\nil}$ where
$\Gamma$ is an arbitrary channel environment. This configuration has
non-trivial extensional behaviour. For example it is input enabled, and
so can perform all extensional actions of the form $c?v$. It can also 
perform $\sigma$ actions, indicating the passage of time. 

Now let $W$ be arbitrary station code such that $\fsn{W} = \emptyset$, that is
it can not broadcast on any free channel. The configuration $\conf{\Gamma}{W}$
has similar behaviour. Indeed let $\Ss$ be the relation 
\begin{align*}
  &\{ (\conf{\Gamma}{W},\, \conf{\Gamma}{\nil})\,\,|\,\,  \fsn{W} = \emptyset\}
\end{align*}
Then it is straightforward to show that $\Ss$ is a bisimulation in the
extensional LTS. Our soundness result therefore ensures that 
\[ \conf{\Gamma}{W} \; \simeq \; \conf{\Gamma}{\nil}  \]
whenever $\fsn{W} = 0$. 
\end{example}
}

Next we consider what happens when  a channel
\MHc{becomes permanently exposed}. This situation can be modelled by 
using two stations $s_0, s_1$ which repeatedly  
send a value along channel $c$; 
each broadcast performed by $s_1$ takes place 
before the transmission of $s_0$ ends, and 
vice versa. In this case we say that the 
channel $c$ is \emph{corrupted\/}. 
Clearly, if a \MHc{system} transmits only on corrupted 
\MHc{channels}; then it cannot be detected. 
\MHc{Let us see how this scenario} is reflected in our behavioural
theory.

\begin{example}[Noise obfuscates transmissions]
Let $v$ be a value such that 
$\deltav = 2$ and let 
$\text{Snd}(c)$ denote the code   $\fix X \bcastc c v X$, 
which continually broadcasts an arbitrary value $v$ along $c$. To model the two stations
$s_0$ and $s_1$ discussed informally above we use the code
$\text{Noise}(c) = \text{Snd}(c) | \sigma.\text{Snd}(c)$.

Then, consider a configuration\MHf{does this need to be well-formed?}
$\conf{\Gamma}{W}$ such that $\fsn{W} \subseteq \{c\}$; that is 
does not transmit on free channels different from $c$. 
\MHc{Then\[
\conf{\Gamma}{W | \text{Noise(c)}} 
\; \simeq \; \conf{\Gamma}{\text{Noise(c)}}  \]
To prove this, it is sufficient to exhibit  bisimulation
containing the pair of configurations 
$(\conf{\Gamma}{W | \text{Noise(c)}},\,\,
\conf{\Gamma}{\text{Noise(c)}})$. 

We use the following abbreviations:
\begin{align*}
\text{Noise}'(c) &= \sigma^2.\text{Snd}(c) | \sigma.\text{Snd}(c)\\
\text{Noise}''(c) &= \sigma.\text{Snd}(c) | \text{Snd}(c)\\
\text{Noise}'''(c) &= \sigma.\text{Snd}(c) | \sigma^2.\text{Snd}(c)
\end{align*}
Then let $\Ss$ denote the following set of pairs of configurations:
\[
\begin{array}{lcr}
\{(\conf{\Delta}{W | \text{Noise(c)}}&,& \conf{\Delta'}{\text{Noise(c)}}),\\ 
(\conf{\Delta}{W | \text{Noise}'(c)}&,& \conf{\Delta'}{\text{Noise}'(c)}), \\
(\conf{\Delta}{W | \text{Noise''}(c)}&,& \conf{\Delta'}{\text{Noise}''(c)}),\\ 
(\conf{\Delta}{W | \text{Noise'''}(c)}&,&\conf{\Delta'}{\text{Noise}'''(c)})\,\,\,| \\ 
\multicolumn{3}{c}{\Delta, \Delta' \vdash c: \cbusy, \fsn{W} \subseteq \{c\}\,\,\, \}}
\end{array}
\]
Then it is possible to check that $\Ss$ is a weak bisimulation in the 
extensional LTS. \MHf{Andrea: is this actually a strong bisimulation?}
}
At least intuitively, this is because 
in the extensional LTS all outputs fired along the obfuscated 
channel $c$ corresponds to internal actions; further, in the configurations 
included in $\Ss$,  channel $c$ is never released, so that neither 
$\iota(c)$-actions nor $\gamma(c,v)$-actions can be performed by any 
configuration included in $\Ss$.
% Obviously, the above result could be generalised by restricting over 
% all channels different from $c$  appearing in sending position 
% inside $W$. 
\end{example}

The \emph{Carrier Sense Multiple Access} (CSMA)
scheme~\cite{standard802.11} is a widely used MAC-layer  protocol
in which a device senses the channel  (\emph{physical carrier sense})
before transmitting.
More precisely, if the channel is sensed free the sender starts transmitting
immediately, that is in the next instant of time\footnote{\MHc{Recall} that
in wireless systems channels are half-duplex.};
if the channel is busy, that is some other station is transmitting,
the device keeps listening  \MHc{to} the channel until it becomes
idle and then starts transmitting immediately. This
strategy is called \emph{1-persistent\/} CSMA 
% More generally, in
% a \emph{$p$-persistent\/} CSMA strategy (where $p$ is a probability)
% the sender transmits with probability $p$, and waits for the
% next available time slot,
%  with probability $1-p$.
%The (1-persistent) CSMA broadcast 
\MHc{and} can be easily expressed in our calculus 
in terms of the following process:
\[
\bcastCSMA c v P = \fix X {\matchb {\expsd c} X {\bcastc c v P}}
\]
%%The CSMA scheme introduces a well-known inconvenient, called \emph{exposed
%%terminal problem}, where an exposed station delays its transmission 
%%until the channel is felt busy, despite  
%%the intended destination is actually ready to receive. 
So, by definition CSMA \MHc{transmissions} are delayed whenever the channel 
is busy. 

\MHc{In the next example we prove a natural property of CSMA transmissions}.
\begin{example}[Delay in CSMA broadcast]
\MHc{Suppose  $\Gamma \vdash_{\mathrm{t}} c: n$ for some $n>0$. }
Then, for any $k \leq n+1$ 
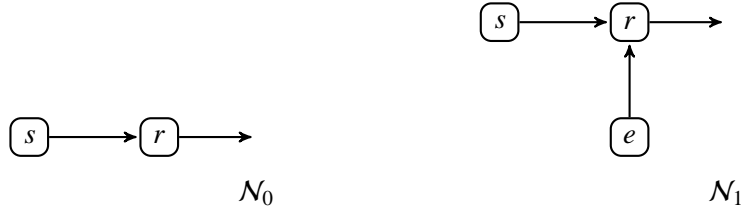
\begin{align}\label{eq:csma}
\conf{\Gamma}{\bcastCSMA c v P} \; \simeq \; 
\conf{\Gamma}{\sigma^k.\bcastCSMA c v P }
\end{align}
\MHc{Intuitively,  since $\Gamma \vdash_{\mathrm{t}} = n$,}
the transmission of value $v$ in $\conf{\Gamma}{\bcastCSMA c v P}$  
can take place only after at least $n$ instants of time. 
The same happens in $\conf{\Gamma}{\sigma^k.\bcastCSMA c v P}$.

\MHc{
Formally, to prove (\ref{eq:csma})  we need to exhibit a
bisimulation $\Ss$ 
which contains all pairs of the form $( \conf{\Gamma}{\bcastCSMA c v P},\, \sigma^k.\bcastCSMA c v P )$,
where $\Gamma$ is such that $\Gamma \vdash_{\mathrm{t}} :n > 0$ for some $n$ satisfying  $k \leq (n+1)$. 
One possible $\Ss$ takes the form $\RR \cup \ID$ where $\ID$ is the identity relation over
configurations and $\RR$ is given by:\MHf{Andrea: What is this bisimulation?}
\[
  \RR = \{ (\conf{\Delta_n}{\bcastCSMA c v P}, \conf{\Delta_n}{\sigma^h.\bcastCSMA c v P}) \;|\; 
  \Delta_n \vdash_{\mathrm{t}} c: n, h \leq n\}
\]
}
\end{example}

\MHc{In} our calculus the 
network topology is \emph{assumed to be} flat. However, we can 
exploit the presence of multiple channels 
to model networks with a more complicated 
topological structure. 
The idea \MHc{is to associate a particular 
channel with a collection of  stations} which 
are in the same neighbourhood.
\begin{table}
\input{net.topology}
\caption{A simple topology for a network}.
\label{fig:net.topology}
\end{table}

\begin{example}[Network Topology]
\label{ex:net.topology} 
Suppose that we want to model a network 
with two stations $s$, $r$ with the 
following features:
\begin{itemize}
\item the range of transmission of $s$ is 
too short to reach external agents, 
\item the station $r$ is in the range of 
transmission of $s$, 
\item the range of transmission of $r$ 
is long enough to \MHc{also} reach external agents. 
\end{itemize}
A graphical representation of the network  
we want to model is given \MHc{as ${\mathcal N}_0$  of Table~\ref{fig:net.topology}. }
We can model this network topology
\MHc{by using a specific restricted channel, say $d$, for the 
local communication between stations $s$ and $r$. 
In CCCP a wireless system running on ${\mathcal N}_0$ would therefore take the
form 
\begin{align*}
  \confC_0 = 
\conf{\Gamma}{\nu d:(0,\cdot).(S | R)}
\end{align*}
where 
\begin{itemize}
\item $S$ represents the code running at station $s$; it can therefore 
only broadcast and receive along the restricted channel $d$ (recall that we do not want 
station $s$ to be able to communicate directly with the external environment)

\item $R$ represents the code running at station $r$; it can only receive 
values along the restricted channel $d$ (since in ${\mathcal N}_0$ station $r$ can receive messages 
broadcast by station $r$, but not by the external environment), while it is free to broadcast on other channels (since 
station $r$ is able to broadcast messages to the external environment) 
\end{itemize}
As a specific example we could let $S$ denote the single broadcast
$\bcastzeroc d v$, and $R = \fix X {\rcvtimec d x {\bcastzeroc c x} {X}}$.
Then in the configuration $\confC_0$ the station $s$ broadcasts 
as a value and station $r$ acts as a forwarder; this behaviour is reminiscent of 
range repeaters in wireless terminology.  

Suppose now that we want to add a second station $e$ 
to the above network topology, so that 
\begin{itemize}
\item  broadcasts  from $e$ can be detected by $r$; this can be accomplished by 
allowing the process used to model station $e$ to broadcasts along a restricted 
channel $d$.

\item  broadcasts from $e$ can not reach  $s$, nor the external 
environment. For this to be true, it is sufficient to require that 
the process which models the behaviour of station $e$ can broadcast 
values only along the restricted channel $d$; further, in order for 
ensuring that the station $e$ cannot detect values broadcast by 
$s$, we require that the process used to represent station $e$ 
does not use receivers along channel $d$.

\MHf{Andrea: But $s$ can use channel $d$ and so I don't understand
why you say that broadcasts from $e$ can not reach $s$} 
\end{itemize}
 The network topology 
we wish to model is depicted as ${\mathcal N}_1$ in 
Table~\ref{fig:net.topology} and so a wireless system running 
on this network takes the form 
\begin{align*}
  &\confC_1 = \crest{d}{(0,\cdot)}.(S | R | E)
\end{align*}
where $E$ is the code running at station $e$. 
As an example we could take $E$ to be the faulty code
$\bcastzeroc d v + \tau.\nil$. 

Then in $\confC_1$ station $r$ still acts as a forwarder for station 
$s$; however station $e$ can non-\newline deterministically 
decide whether to corrupt the transmission from node 
$s$ to $r$, causing a collision. 

\medskip

Let  us assume that the transmission time of the value used in these networks, $v$,
satisfies   $\deltav = \delta_{\err}$.  
Then we can show
\begin{align*}
\confC_0 &\simeq  \conf{\Gamma}{\sigma^{\delta_v}.\bcastzeroc c v}\\
\confC_1 & \simeq  \conf{\Gamma}{\tau.\sigma^{\delta_v}.\bcastzeroc c v + 
\tau.\sigma^{\delta_v}.\bcastzeroc c  {\err}} 
\end{align*}
Intuitively the reasons for these equivalences are obvious.}
The transmission along channel $d$ is restricted in $\confC_0$, so  
it cannot be observed by the external environment. The only activity which 
can be observed is the broadcast of value $v$ along channel $c$, which takes 
place after $\delta_v$ instants of time. 
For $\confC_1$, a collision can happen along channel $d$, which is again 
restricted; the only activity that can be detected by the external environment 
is a transmission which takes place after $\delta_v$ instants of time. Such a 
transmission will contain either the value $v$ or an error message of length 
$\delta_v$.

\MHc{
The formal proof of these identities involves exhibiting two bisimulations, 
containing the relevant pairs of configurations. Here we exhibit a bisimulation 
%for showing that $\confC_1 \simeq \conf{\Gamma}{\tau.\sigma^{\delta_v}.\bcastzeroc c v}$}. 
for showing that $\confC_1 \simeq \conf{\Gamma}{\tau.\sigma^{\delta_v}.\bcastzeroc c v + 
\tau.\sigma^{\delta_v}.\bcastzeroc c  {\err}}$.}
For the sake of simplicity, let $\delta_\err = \deltav = 1$ and define the system terms 
\[
\begin{array}{lcl@{\hspace*{20mm}}lcl}
W &=& \nu d:(0,\cdot).(S | E | R) & W_s &=& \nu d:(1,v).(\sigma | E | \arcv c x {\bcastzeroc c x})\\ 
W_e &=& \nu d:(1,\err).(S | \sigma | \arcv c x {\bcastzeroc c x})& 
W' &=& \nu d:(0,\cdot).(S | \nil | R)\\
W'' &=& \nu d:(1,\err).(\sigma | \sigma | \arcv c x {\bcastzeroc c x})&
W_{\text{ok}} &=& \nu d:(0,\cdot).(\nil | \nil | \bcastzeroc c v)\\
W_{\err} &=& \nu d:(0,\cdot).(\nil | \nil | \bcastzeroc c {\err})&
W_{c} &=& \nu d:(0,\cdot).(\nil | \nil | \sigma)
\end{array}
\]

Then it is easy to show that the relation
\[
\begin{array}{lcllcrr}
\Ss &=& \{& (\conf{\Delta}{W} &,& \conf{\Delta}{\tau.\sigma.\bcastzeroc c v + \tau.\sigma.\bcastzeroc c \err})&,\\
&&&(\conf{\Delta}{W_s}&,& \conf{\Delta}{\tau.\sigma.\bcastzeroc c v + \tau.\sigma.\bcastzeroc c \err})&,\\
&&&(\conf{\Delta}{W_e}&,& \conf{\Delta}{\sigma.\bcastzeroc c \err})&,\\
&&&(\conf{\Delta}{W'}&,&\conf{\Delta}{\sigma.\bcastzeroc c v})&,\\
&&&(\conf{\Delta}{W''}&,&\conf{\Delta}{\sigma.\bcastzeroc c {\err}})&,\\
&&&(\conf{\Delta}{W_{\text{ok}}}&,& \conf{\Delta}{\bcastzeroc c v})&,\\
&&&(\conf{\Delta}{W_{\err}}&,& \conf{\Delta}{\bcastzeroc c {\err}})&,\\
&&&(\conf{\Delta}{W_c}&,&\conf{\Delta}{\sigma})&\\
&&|& \Delta && \text{arbitrary channel environment}&\}
\end{array}
\]
is a weak bisimulation.
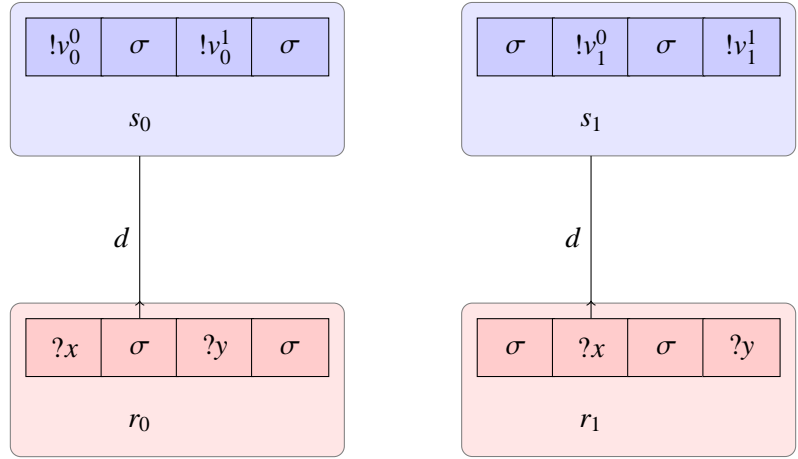
\MHf{Andrea: One of these bisimulations needs to be described.}
\end{example}

\bigskip
The next example shows how the 
TDMA modulation technique \cite{tanenbaum} \MHc{can be described in CCCP}.  \emph{Time
  Division Multiple Access} (TDMA) is a type of Time Division
Multiplexing, where instead of having one transmitter connected to one
receiver, there are multiple transmitters. TDMA is used in the digital
2G cellular systems such as \emph{Global System for Mobile
  Communications} (GSM). TDMA allows several users to share the same
frequency channel by dividing the signal into different time
slots. The users transmit in rapid succession, one after the other,
each using his own time slot. This allows multiple stations to share
the same transmission medium (e.g. radio frequency channel) while
using only a part of its channel capacity.

As a simple example let us describe how two messages $v_0$ and $v_1$
can be delivered in TDMA style; for simplicity, we assume  $\delta_{v_0} = \delta_{v_1} = 2$.   
The main idea here is 
to \MHc{split each of these values into two packets of length 
one, transmit the packets individually, which will then} 
be concatenated together before being forwarded to the 
external environment. So let us assume values 
$v_0^0, v_0^1, v_1^0, v_1^1$, 
each of which requires one time instant to be transmitted, 
and  a binary operator $\circ$ for composing messages such that 
\begin{eqnarray*}
v_0^0 \circ v_0^1 &=& v_0\\
v_1^0 \circ v_1^1 &=& v_1\\
v \circ \err &=& \err \circ v = \err
\end{eqnarray*}
where $v$ is an arbitrary value; in this case we assume that $\delta_{\err} = 2$.

More specifically, for this example we assume four different 
stations, $s_0, s_1, r_0, r_1$, running the code $\hat{S}_0,\, \hat{S}_1,\,
\hat{R}_0,\,\hat{R}_1$ respectively. 
\begin{table}
\begin{center}
\input{tdma}
\caption{Two transmitting stations using different 
time slots to broadcast values\label{fig:tdma}}
\end{center}

\end{table}
\MHf{MM: In Figure~\ref{fig:tdma} the second $s_0$ must be $s_1$ \MHc{MH: I don't 
understand this comment by Massimo. }}
The network we consider for modelling the TDMA transmission 
is \MHc{then} given by
\[
\confC_0 = \conf{\Gamma}{\crest{d}{(0,\cdot)} \big( \hat{S}_0 | \hat{S}_1 | 
\hat{R}_0 | \hat{R}_1 \big)}
\]
where 
\begin{eqnarray*}
\hat{S}_0 &=& \bcastc d {v_0^0} {\sigma.\bcastzeroc d {v_0^1}}\\
\hat{S}_1 &=& \sigma.\bcastc d {v_1^0} {\sigma.\bcastzeroc d {v_1^1}}\\
\hat{R}_0 &=& \rcvtimec d x {\sigma.\rcvtimec d y {\sigma.\bcastzeroc c {x \circ y}}{}}{}\\
\hat{R}_1 &=& \sigma.\rcvtimec d x {\sigma.\rcvtimec d y {\sigma^2.\bcastzeroc c {x \circ y}}{}}{}
\end{eqnarray*}
The intuitive behaviour of this network is 
depicted in Table~\ref{fig:tdma}.
Station $s_0$ wishes to broadcast value $v_0$, while 
$s_1$ wishes to broadcast value $v_1$. They both use 
the same (restricted) channel $d$ to broadcast their 
respective values; however, both stations split 
the value to be broadcast in two packets. Value 
$v_0$ is split in $v_0^0$ and $v_0^1$, while 
$v_1$ is split in $v_1^0$ and $v_1^1$. 

The two stations run a TDMA protocol with a time 
frame of length two. Station $s_0$ takes control 
of the first time frame, hence transmitting its 
two packets $v_0^0$ and $v_0^1$  in the first and the 
third time slot, respectively. 
Station $s_1$  takes control of 
the second time frame; hence the two packets $v_1^0$ 
and $v_1^1$ are broadcast in the 
second and fourth time slot, respectively. 

Stations $r_0$ and $ r_1$ wait to collect the values broadcast 
along channel $d$. However, the former is interested 
only in packets sent in the first time frame, 
while the latter detects only values 
sent in the second time frame. 
At the end of their associated time frame the 
 stations $r_0$ and $r_1$  have  received 
two packets which are concatenated together and then
 broadcast to the external environment along channel $c$. Note that 
station $r_1$ is a little slower than $r_0$, for we have added 
a delay of two time units before broadcasting the concatenated values.

\begin{table}
\begin{center}
\input{net.topology.2}
\caption{Forwarding two messages to the external environment\label{tab:top2}}
\end{center}
\end{table}

\MHc{As an alternative to TDMA, the two}
values $v_0$, $v_1$ 
can be also be delivered to the external environment by means of 
a simple routing, along the lines suggested in Example~\ref{ex:net.topology}. 
\MHc{Here we} consider the configuration 
\[
\MHc{   \confC_1 = \conf{\Gamma}{\crest{d}{(0,\cdot)}.(S_0 | S_1 | R)}   }
\]
where 
\begin{eqnarray*}
S_0 &=& \sigma^4.\bcastzeroc c {v_0}\\
S_1 &=& \sigma^4.\bcastzeroc d {v_1}\\
R &=& \rcvc d x {\bcastzeroc c x} 
\end{eqnarray*} 

Intuitively, 
the configuration $\confC_1$ models three 
wireless stations $s_0,s_1,r$, running
 the code $S_0,S_1$,\\$R$, respectively, and  connected as in 
Table~\ref{tab:top2}. 
Station $s_0$ waits four instants of time, 
then it broadcasts value $v_0$ directly to 
the external environment via the free channel $c$. Similarly, 
after four instants of time the 
station $s_1$ broadcasts value $v_1$ 
to station $r$ via the restricted channel $d$. Finally, 
$r$ forwards the message to the external 
environment via the free channel $c$.  

From the point of view of the external environment 
the configuration $\confC_1$ performs the following 
activities:
\begin{itemize}
\item it remains idle for the first four instants of time
\item it transmits value $v_0$ in the fifth and sixth time instants
\item it transmits value $v_1$ in the seventh and eighth time instants. 
\end{itemize}
In this manner, \MHc{at least informally the observable behaviour of $\confC_1$, which uses direct routing, is the same as 
that of $\confC_0$, which uses TDMA. }

Formally, we can prove 
\begin{align}\label{eq:zero.one}
   & \confC_0 \simeq \confC_1 
\end{align}
\MHc{However, instead of proving this by giving a bisimulation containing this pair of configurations,
we prove them individually bisimilar to a simpler specification. Let ${\mathcal S}_1$ denote the 
configuration $\conf{\Gamma}{S_1}$ where $S_1$ is the code\MHf{Andrea: Can you describe these bisimulations? Or at least one
of them?} 
\begin{align*}
  &\sigma^4. \bcastc c {v_0} \bcastc c {v_1}
\end{align*}
Then we can show that $\confC_0 \approx {\mathcal S}_1$ and  $\confC_1 \approx {\mathcal S}_1$, from
which (\ref{eq:zero.one}) follows by soundness.
}
Let us show that $\confC_0 \simeq {\mathcal S}_1$; for the sake of simplicity, 
it will be convenient to define the following system terms: 
\[
\begin{array}{lcr@{\hspace*{20mm}}lcl}
\hat{S}_0^n&=& \sigma^n.\bcastzeroc d {v_0^1}&
\hat{S}_1' &=& \bcastc d {v_1^0} {\sigma.\bcastzeroc d {v_1^1}}\\
\hat{S}_1^n &=& \sigma^n.\bcastzeroc d {v_1^1} &
(\hat{R}_0)^\text{act} &=& \arcv d x {\sigma.\rcvtimec d y {\sigma.\bcastzeroc c {x \circ y}}{}}\\
\hat{R}_0' &=& \rcvtimec d y {\sigma.\bcastzeroc c {v_0^0 \circ y}}{}&
(\hat{R}_0')^{\text{act}} &=& \arcv d y {\sigma.\bcastzeroc c {v_0^0 \circ y}}\\
\hat{R}_0^f &=& \bcastzeroc c {v_0^0 \circ v_0^1}&
\hat{R}_1' &=& \rcvtimec d x {\sigma.\rcvtimec d y {\sigma^2.\bcastzeroc c {x \circ y}}{}}{}\\
(\hat{R}_1')^{\text{act}} &=& \arcv d x {\sigma.\rcvtimec d y {\sigma^2.\bcastzeroc c {x \circ y}}{}} & 
\hat{R}_1'' &=& \rcvtimec d y {\sigma^2.\bcastzeroc c {v_1^0 \circ y}}{}\\
(\hat{R}_1'')^{\text{act}} &=& \arcv d y {sigma^2.\bcastzeroc c {v_1^0 \circ y}}& 
\hat{R}_1^f &=& \bcastzeroc c {v_1^0 \circ v_1^1}\\
W_n &=& \sigma^n.\bcastc c {v_0} {\bcastzeroc c {v_1}}
\end{array}
\]
Then the relation 
\[
\begin{array}{lcllcrr}
\RR &=& \{& (\conf{\Delta}{\nu d:(0,\cdot).(\hat{S}_0 | \hat{S}_1 | \hat{R}_0 | \hat{R}_1)} &,& \conf{\Delta}{W_4})&,\\
&&&(\conf{\Delta}{\nu d:(1,v_0^0).(\hat{S}_0^2 | \hat{S}_1 | \hat{R}_0^{\text{act}} | \hat{R}_1)}&,&\conf{\Delta}{W_4})&,\\
&&&(\conf{\Delta}{\nu d:(0,\cdot).(\hat{S}_0^1 | \hat{S}_1' | \sigma.\hat{R}_0' | \hat{R}_1')}&,&\conf{\Delta}{W_3})&,\\
&&&(\conf{\Delta}{\nu d:(1,v_1^0).(\hat{S}_0^1 | \hat{S}_1^2 | \sigma.\hat{R}_0' | (\hat{R}_1')^\text{act})}&,&\conf{\Delta}{W_3})&,\\
&&&(\conf{\Delta}{\nu d:(0,\cdot).(\bcastzeroc d {v_0^1} | \hat{S}_1^1 | \hat{R}_0' | \sigma.\hat{R}_1'')}&,&\conf{\Delta}{W_2})&,\\
&&&(\conf{\Delta}{\nu d:(1,v_0^1).(\sigma | \hat{S}_1^1 | (\hat{R}_0')^{\text{act}} | \sigma.\hat{R}_1'')}&,&\conf{\Delta}{W_2})&,\\
&&&(\conf{\Delta}{\nu d:(0,\cdot).(\nil | \bcastzeroc c {v_1^1} | \sigma.\hat{R}_0^f | \hat{R}_1'')}&,&\conf{\Delta}{W_1})&,\\
&&&(\conf{\Delta}{\nu d:(1,v_1^1).(\nil | \bcastzeroc c {v_1^1} | \sigma.\hat{R}_0^f | (\hat{R}_1'')^{\text{act}})}&,&\conf{\Delta}{W_1})&,\\
&&&(\conf{\Delta}{\nu d:(0,\cdot).(\nil | \nil | \bcastzeroc c {v_0} | \sigma^2.\bcastzeroc c {v_1})}&,&\conf{\Delta}{\bcastc c {v_0} {\bcastzeroc c {v_1}}})&,\\
&&&(\conf{\Delta}{\nu d:(0,\cdot).(\nil | \nil | \sigma^2 | \sigma^2.\bcastzeroc c {v_1})}&,&\conf{\Delta}{\sigma^2. {\bcastzeroc c {v_1}}})&,\\
&&&(\conf{\Delta}{\nu d:(0,\cdot).(\nil | \nil | \sigma | \sigma.\bcastzeroc c {v_1})}&,&\conf{\Delta}{\sigma. {\bcastzeroc c {v_1}}})&,\\
&&&(\conf{\Delta}{\nu d:(0,\cdot).(\nil | \nil | \nil | \bcastzeroc c {v_1})}&,&\conf{\Delta}{{\bcastzeroc c {v_1}}})&,\\
&&&(\conf{\Delta}{\nu d:(0,\cdot).(\nil | \nil | \nil | \sigma^2)}&,&\conf{\Delta}{{\sigma^2}})&,\\
&&&(\conf{\Delta}{\nu d:(0,\cdot).(\nil | \nil | \nil | \sigma)}&,&\conf{\Delta}{{\sigma}})&,\\
&&&(\conf{\Delta}{\nu d:(0,\cdot).(\nil | \nil | \nil | \nil)}&,&\conf{\Delta}{{\nil}})&,\\
&&|& \Delta \text{ arbitrary channel environment}&&&\}
\end{array}
\]

is a bisimulation. 
Below we also
show  that $\confC_1 \simeq {\mathcal S}_1$; for the sake of 
simplicity, define the following terms:
\[
\begin{array}{lcr@{\hspace*{40mm}}lcl}
S_0^n &=& \sigma^n.\bcastzeroc c {v_0}& 
S_1^n &=& \sigma^n.\bcastzeroc d {v_1}\\
R' &=& \arcv d x {\bcastzeroc c x}&
W_n &=& \sigma^n.\bcastc c {v_0} {\bcastzeroc c {v_1}}
\end{array}
\]
for any $n \in \mathbb{N}$. 
Then the relation 
\[
\begin{array}{lcllcrr}
\RR' &=& \{& (\conf{\Delta}{\nu d:(0,\cdot).(S_0^n | S_1^n | R)} &,& \conf{\Delta}{W_n})&,\\
&&&(\conf{\Delta}{\nu d:(0,\cdot).(\sigma^2 | \bcastzeroc d {v_1} | R)}&,&\conf{\Delta}{\sigma^2.\bcastzeroc c {v_1}})&,\\
&&&(\conf{\Delta}{\nu d:(2,v_1).(\bcastzeroc c {v_0} | \sigma^2 | R')}&,&\conf{\Delta}{\bcastc c {v_0} {\bcastzeroc c {v_1}}})&,\\
&&&(\conf{\Delta}{\nu d:(2,v_1).(\sigma^2 | \sigma^2 | R')}&,&\conf{\Delta}{\sigma^2.\bcastzeroc c {v_1}})&,\\
&&&(\conf{\Delta}{\nu d:(1,v_1).(\sigma | \sigma | R')}&,& \conf{\Delta}{\sigma.\bcastzeroc c {v_1}})&,\\
&&&(\conf{\Delta}{\nu d:(0,\cdot).(\nil | \nil | \bcastzeroc c {v_1})}&,& \conf{\Delta}{\bcastzeroc c {v_1}})&|\\
&&|& \Delta \text{ arbitrary channel environment}&&&\}
\end{array}
\]
is a relation which contains the most relevant couples needed for showing that $\confC_1 \approx {\mathcal S}_1$.
\medskip

\begin{example}
\MHc{As a final example we can modify the} 
behaviour of the two configurations $\confC_0$ and 
$\confC_1$ seen above \MHc{by adding the}
possibility of getting a \emph{collision} when delivering values $v_0,v_1$ to the external environment. 
In the routing case, this is accomplished by requiring 
that both stations $s_0, s_1$ can either broadcast 
their value directly to the external 
environment or to the forwarder node $r$, 
while in the TDMA case it is sufficient to allow 
both the stations $s_0$, $s_1$ to non-deterministically 
choose the time slot to be used to broadcast packets. 

To this end, let 
\begin{eqnarray*}
S_0^{\mathsf{c}} &=& \tau.\sigma^4.\bcastzeroc c {v_0} + 
\tau.\sigma^4.\bcastzeroc d {v_0}\\
S_1^{\mathsf{c}} &=& \tau.\sigma^4.\bcastzeroc c {v_1} + 
\tau.\sigma^4.\bcastzeroc d {v_1}\\[5pt]
\hat{S}_0^{\mathsf{c}} &=& \bcastc d {v_0^0} {\sigma.\bcastzeroc d {v_0^1}}+ 
\tau.\sigma.\bcastc d {v_0^0} {\sigma.\bcastzeroc d {v_0^1}}\\
\hat{S}_1^{\mathsf{c}} &=& \bcastc d {v_1^0} {\sigma.\bcastzeroc d {v_1^1}}
+\tau.\sigma.\bcastc d {v_1^0} {\sigma.\bcastzeroc d {v_1^1}}\\
\end{eqnarray*} 
and consider the configurations 
\begin{eqnarray*}
\MHc{\confC_1^{\mathsf{c}}} &=& \conf{\Gamma}{\crest{d}{(0,\cdot)}.(S_0^{\mathsf{c}} | S^{\mathsf{c}}_1 | R)}\\
\MHc{\confC_0^{\mathsf{c}}} &=& \conf{\Gamma}{\crest{d}{(0,\cdot)}.(\hat{S}^{\mathsf{c}}_0 | \hat{S}^{\mathsf{c}}_1 | \hat{R}_0 | \hat{R}_1)}
\end{eqnarray*}

It is not difficult to \MHc{see informally that} the observable behaviour of these 
two configurations is the same.  \MHc{Specifically}
\begin{itemize}
\item either value $v_0$ is broadcast in the fifth and sixth 
time slots and $v_1$ is broadcast in the  seventh and eighth instants of time slots, or
\item value $v_1$ is broadcast in the fifth and sixth time slots, 
while value $v_0$ is broadcast in the seventh and eighth time slots, or 
\item a collision occur in the fifth and sixth time slots, or
\item a collision occur in the seventh and eighth time slots.
\end{itemize}
\MHc{
This informal behaviour can be described by the term\MHf{Andrea: can you actually
describe $S_2$ and then describe one of the bisimulations?} 
\[
\begin{array}{lcll}
  S_2 &=& \tau.\sigma^4.\bcastc c {v_0} {\bcastzeroc c {v_1}}&+\\
  && \tau.\sigma^4.\bcastc c {v_1} {\bcastzeroc c {v_0}}&+\\
  && \tau.\sigma^4.\bcastzeroc c {\err}&+\\
  && \tau.\sigma^6.\bcastzeroc c {\err}&
\end{array}
\]
and once more we can exhibit  bisimulations to establish 
$\conf{\Gamma}{S_2} \approx \confC_0^{\mathsf{c}}$ and 
$\conf{\Gamma}{S_2} \approx \confC_1^{\mathsf{c}}$. 
Then soundness again ensures that }
\[\confC_0^{\mathsf{c}} 
\; \simeq \; 
\confC_1^{\mathsf{c}}
\]
\end{example}

\section{Conclusions and related work}
\label{sec:conclusions}
\leaveout{
We end this paper by performing a comparison with the most relevant related work.

We start with the literature on process calculi for wireless systems. 
Nanz and Hankin~\cite{nanz} have introduced  a calculus for Mobile 
Wireless Networks (CBS$^{\sharp}$), relying on graph representation of node localities.
The main goal of the paper is to  present a framework for specification and security analysis of communication protocols for mobile wireless networks. 
Merro~\cite{Merro07}  has proposed an untimed process calculus for 
Mobile Ad Hoc Networks with 
 a labelled characterisation of reduction barbed congruence; some of 
our algebraic properties already appear in \cite{Merro07}. 
Godskesen~\cite{Godskesen07}  has proposed  a calculus for mobile 
ad hoc networks (CMAN). 
The paper proves 
a characterisation of reduction barbed congruence in terms of a
 contextual bisimulation. It also contains a formalisation of an attack on
 the cryptographic routing protocol ARAN.  
Singh, Ramakrishnan and Smolka \cite{SRS06} 
have proposed the $\omega$-calculus, a conservative extension of the 
$\pi$-calculus. A key feature of the $\omega$-calculus is 
the separation of a node's communication and computational behaviour 
from the description of its physical transmission range. 
The authors provide a labelled transition semantics and  
a bisimulation in ``open'' style. The
$\omega$-calculus is then used for modelling the AODV 
routing protocol.  
Ghassemi et al.\ \cite{GhasWanFok08} have proposed
a process algebra for mobile ad hoc networks (RBPT) where,  
topology changes are implicitly modelled in the 
(operational) semantics rather than 
in the syntax. The authors propose a notion of
bisimulation for networks parametrised on a set of 
topology invariants that must be respected by equivalent 
networks. This work in then refined in \cite{GhasWanFok09}
where the authors propose an equational theory for an extension 
of RBPT. Godskesen and Nanz~\cite{GodNanz09} 
have proposed a 
simple timed calculus for wireless systems to express a
wide range of mobility models. 
Kouzapas and Philippou~\cite{KouzapasP11} have developed a theory of confluence for
 a calculus of dynamic networks and use their machinery to verify a leader-election algorithm for mobile ad hoc networks.
Borgstr\"om et al. \cite{Borgstrom_etal11} have proposed an extension 
of the $\pi$-calculus to model ad-hoc routing protocols, called LUNAR, 
verifying a basic reachability property.  
All the previous calculi abstract from the presence of interferences. Lanese and Sangiorgi~\cite{LaneseS10} have instead proposed the CWS calculus, 
a lower level untimed calculus
 to describe interferences in wireless systems. In their operational
semantics there is a 
separation between transmission beginning and transmission ending. 
In \cite{Wang13}, Wang et al. introduce a generalisation of CWS with timed communication and node mobility.
Another timed generalisation of CWS has been proposed in \cite{MerroBS11} to 
express MAC-layer protocols such as CSMA/CA; the authors propose
a bisimilarity which is proved to be sound but not complete with respect
to a notion of reduction barbed congruence. The current work wants
to be a simplification and a generalisation of \cite{MerroBS11}.

%%Song and Godskesen~\cite{Song_Godskesen10}, have proposed
%%a probabilistic calculus for wireless systems modelling unreliable connection. 
%%They have characterised a notion of weak bisimilarity by a variant of PCTL. 

None of the calculi mentioned above, except for \cite{GodNanz09,MerroBS11}, deals with time, although there is an extensive literature on timed process algebra. From a purely syntactic point of 
view, the earliest proposals are extensions of the three main process algebras, 
ACP, CSP and CCS. For example, \cite{BB89} presents a real-time extension of 
ACP, \cite{Re88} contains a denotational model for a timed extension of CSP, 
while CCS is the starting point for \cite{MT90a}. 
In \cite{BB89} and \cite{Re88} time is real-valued, and at least semantically, 
associated directly with actions. The other major approach to representing 
time is to introduce special actions to model the passage of time, which 
the current paper shares with \cite{Gr89,BaBe96,MT90a,Sifakis94} and \cite{Yi90, Yi91}, 
although the basis for all those proposals may be found in \cite{BC84}.
The current paper shares many of the assumptions of the languages 
presented in these papers. For example, all the papers above 
assume that actions are instantaneous and only the extension of ACP 
presented in \cite{Gr89} does not incorporate time determinism; however 
maximal progress is less popular and patience is even rarer. 

More recent works on timed process algebra include the following papers. 
Aceto and Hennessy~\cite{AceHen93} have presented a simple process 
algebra where time emerges in the definition of a \emph{timed observational equivalence},
assuming that beginning and termination of actions are distinct events 
which can be observed.
Hennessy and Regan~\cite{HennessyR95}  have proposed a timed version of CCS
enjoying time determinism, maximal progress and patience. 
Our action $\sigma$ takes inspiration from theirs. 
The authors have developed a semantic theory based on testing and characterised in terms of a particular kind of ready traces.
Prasad~\cite{Pra96} has proposed a timed variant of his 
CBS~\cite{CBS}, called TCBS.
In TCBS a time out can force a process wishing to speak to remain idle for a specific interval of time; this corresponds to have a priority. 
TCBS also assumes time determinism and maximal progress. 
Corradini et al.~\cite{CFP01}  deal with \emph{durational actions}
 proposing a framework relying on the notions of reduction and observability 
to naturally  incorporate timing information in terms of process interaction.
Our definition of timed reduction barbed congruence takes inspiration from 
theirs.
Corradini and Pistore~\cite{CorPist01} have studied durational actions to describe and reason about the performance of systems. Actions have lower and upper time bounds, specifying their possible different durations. Their 
\textit{time equivalence} refines the untimed one.
Baeten and Middelburg~\cite{BaMi02} have proposed 
several timed process algebras treated 
in a common framework, and related by embeddings and conservative extensions
relations. These process algebras,  $\mathrm{ACP^{sat}}$, $\mathrm{ACP^{srt}}$, 
$\mathrm{ACP^{dat}}$ and  $\mathrm{ACP^{drt}}$, allow the execution of two 
or more actions consecutively at the same point in time, separate the execution
of actions from the passage of time, and consider actions to have no duration.
The process algebra $\mathrm{ACP^{sat}}$ is a real-time process algebra with 
absolute time, $\mathrm{ACP^{srt}}$ is a real-time process algebra with
relative time. Similarly, $\mathrm{ACP^{dat}}$ and  $\mathrm{ACP^{drt}}$ are 
discrete-time process algebras with absolute time and relative time, respectively. In these process algebra the focus is on unsuccessful termination or 
deadlock. 
In \cite{BaRe04} Baeten and Reniers extend the framework of \cite{BaMi02} to model successful
termination for the relative-time case. 
Laneve and Zavattaro~\cite{LaneveZ05} have proposed a timed extension of 
$\pi$-calculus 
where time proceeds asynchronously at the network level, while it is constrained by the local urgency at the process level. They propose a timed bisimilarity
whose discriminating is weaker when local urgency
is dropped. 
}

In this paper we have given a behavioural theory of wireless systems
at the MAC level.  In our framework individual wireless stations
broadcast information to their neighbours along virtual channels.
These broadcasts take a certain amount of time to complete, and are
subject to collisions.  If a broadcast is successful a recipient may
choose to ignore the information it contains, or may act on it, in
turn generating further broadcasts. We believe that our reduction
semantics, given in Section~\ref{sec:calculus}, captures much of the
subtlety of intensional MAC-level behaviour of wireless systems.

Then based on this reduction semantics we defined a natural contextual
equivalence between wireless systems which captures the intuitive idea
that one system can be replaced by another in a larger network without
affecting the observable behaviour of the original network. In the main result of
the paper, we then gave a sound and complete characterisation of this
behavioural equivalence in terms of \emph{extensional actions}. This
characterisation is important for two reasons. Firstly it gives an
understanding of which aspects of the intensional behaviour is important
from the point of view of external users of these wireless systems. Secondly
it gives a powerful sound and complete co-inductive proof method for demonstrating
that two systems are behaviourally equivalent. We have also demonstrated the
viability of this proof methodology by a series of examples.

Let us now examine some relevant related work.  We start with the
literature on process calculi for wireless systems.  Nanz and
Hankin~\cite{NaHa06} have introduced the first (untimed) calculus for Mobile
Wireless Networks (CBS$^{\sharp}$), relying on a graph representation
of node localities.  The main goal of that paper is to present a
framework for specification and security analysis of communication
protocols for mobile wireless networks.  Merro~\cite{Merro07} has
proposed an untimed process calculus for mobile ad-hoc networks with a
labelled characterisation of reduction barbed congruence, while
\cite{Godskesen07} contains a calculus called CMAN, also with mobile
ad-hoc networks in mind. This latter paper also gives a
characterisation of reduction barbed congruence, this time in terms of
a contextual bisimulation. It also contains a formalisation of an
attack on the cryptographic routing protocol ARAN. Kouzapas and
Philippou~\cite{KouzapasP11} have developed a theory of confluence for
a calculus of dynamic networks and they use their machinery to verify
a leader-election algorithm for mobile ad hoc networks.

Singh, Ramakrishnan and Smolka \cite{SRS06} have
proposed the $\omega$-calculus, a conservative extension of the
$\pi$-calculus. A key feature of the $\omega$-calculus is the
separation of a node's communication and computational behaviour from
the description of its physical transmission range.  The authors
provide a labelled transition semantics and a bisimulation in \emph{open}
style. The $\omega$-calculus is then used for modelling the AODV ad-hoc 
routing protocol. Another extension of the $\pi$-calculus for modelling
mobile wireless systems may be found in \cite{Borgstrom_etal11}; the calculus
is used to verify  reachability properties of the ad-hoc routing protocol
LUNAR.  
Fehnker et al.~\cite{Fehnker_etal2012}
have proposed a process algebra for wireless mesh networks that combines novel treatments of local broadcast, conditional unicast and data structures. In this framework, they also model the AODV routing protocol and (dis)prove crucial properties such as loop freedom and packet delivery. 
Vigo et al.~\cite{Vigo_etal2013}
have proposed a calculus of broadcasting processes that enables to reason about unsolicited messages and lacking of expected communication. Moreover, standard cryptographic mechanisms can be implemented in the calculus via term rewriting. 
The modelling framework is complemented by an executable specification of the semantics of the calculus in Maude. 

All the calculi, mentioned up to now, except for~\cite{NaHa06}, 
represent 
topological changes of mobile networks in the syntax. In contrast
Ghassemi et al.\ \cite{GhasWanFok08} have proposed a process algebra
called RBPT where topological changes to the connectivity graph are
implicitly modelled in the operational semantics rather than in the
syntax. They propose a notion of bisimulation for networks
parametrised on a set of topological invariants that must be
respected by equivalent networks. This work in then refined in
\cite{GhasWanFok09} where the authors propose an equational theory for
an extension of RBPT. 
Godskesen and Nanz~\cite{GodNanz09} have
proposed a simple timed calculus for wireless systems to express a
wide range of mobility models. 

A simple notion of time is also adopted in the calculus for wireless 
systems 
by Macedonio and Merro~\cite{Macedonio_Merro2012} to verify key management 
protocols 
for wireless sensor networks by applying semantics-based techniques. 
In~\cite{Lanotte_Merro11}
this notion of time is extended with probabilities. In this paper a probabilistic simulation theory is proposed 
 to evaluate the performances gossip protocols in the context of wireless sensor networks. 
Paper~\cite{Song_Godskesen2010} also presents a probabilistic broadcast calculus for wireless networks where, unlike \cite{Lanotte_Merro11}, nodes are 
mobile; due to mobility the connection probabilities may change. The authors examine the relation between a notion of weak bisimulation and a minor variant of PCTL*. Paper~\cite{Cerone_Hennessy2013} investigate in detail 
the probabilistic behaviour of wireless networks. The paper presents a compositional theory based on a probabilistic generalisation of the well known may-testing and must-testing pre-orders. Also, it provides an extensional semantics to define both simulation and deadlock simulation preorders for wireless networks. 
Gallina et al.~\cite{BGMR12} propose a process algebraic model targeted at the analysis of both connectivity and communication interference in ad hoc networks. 
The framework includes a probabilistic process calculus and a suite of analytical techniques based on a probabilistic observational congruence and an interference-sensitive preorder. 
In particular, the preorder makes it possible to evaluate the interference level of different, behaviourally equivalent, networks. 
They use their framework to analyse the Alternating Bit Protocol. 
Song and Godskesen~\cite{Song_Godskesen2012} introduce a continuous time stochastic broadcast calculus for mobile and wireless networks. The mobility between nodes in a network is modelled by a stochastic mobility function which allows to change part of a network topology depending on an exponentially distributed delay and a network topology constraint. They define a weak bisimulation congruence and apply their theory on a leader election protocol.

All the calculi mentioned up to now abstract away from the 
possibility of interference between broadcasts.  Lanese and
Sangiorgi~\cite{LaneseS10} have instead proposed the CWS calculus, a
lower level untimed calculus to describe interferences in wireless
systems. In their operational semantics there is a separation between
the beginning and ending of a broadcast, so there is some implicit
representation of the passage of time.  A more explicit timed
generalisation of CWS is given \cite{MerroBS11} to express MAC-layer
protocols such as CSMA/CA, where the authors propose a bisimilarity
which is proved to be sound but not complete with respect to a notion
of reduction barbed congruence. We view the current paper as a simplification
and generalisation of  \cite{MerroBS11}.
 
The research we have mentioned so far has been focused on formalising various
aspects of ad-hoc networks. However other than 
\cite{GodNanz09,MerroBS11}, these various calculi abstract away from 
time. Nevertheless there is an extensive literature on  timed process algebras,
which we now briefly review. From a purely syntactic
point of view, the earliest proposals are extensions of the three main
process algebras, ACP, CSP and CCS. For example, \cite{BB89} presents
a real-time extension of ACP, \cite{Re88} contains a denotational
model for a timed extension of CSP, while CCS is the starting point
for \cite{MT90a}.  In \cite{BB89} and \cite{Re88} time is real-valued,
and at least semantically, associated directly with actions. The other
major approach to representing time is to introduce a special action to
model the passage of time, and to assume that all other actions are
instantaneous. This approach is advocated in     
\cite{Gr89,BaBe96,MT90a,Sifakis94} and \cite{Yi90, Yi91}, although the
basis for this approach  may be found in \cite{BC84}.  The
current paper shares many of the assumptions of the languages
presented in these papers; in particular we have been influenced
by ~\cite{HennessyR95} which contains a timed version of CCS enjoying
time determinism, maximal progress and patience.
All the just mentioned papers  assume
that actions are instantaneous and only the extension of ACP presented
in \cite{Gr89} does not incorporate time determinism; however maximal
progress is less popular and patience is even rarer.

From this early work on timed process calculi a flourishing literature has
emerged. Here we briefly mention some highlights  of this research. 
\leaveout{
More recent works on timed process algebra include the following
papers.  Aceto and Hennessy~\cite{AceHen93} have presented a simple
process algebra where time emerges in the definition of a \emph{timed
  observational equivalence}, assuming that beginning and termination
of actions are distinct events which can be observed.  Hennessy and
Regan~\cite{HennessyR95} have proposed a timed version of CCS enjoying
time determinism, maximal progress and patience.  Our action $\sigma$
takes inspiration from theirs.  The authors have developed a semantic
theory based on testing and characterised in terms of a particular
kind of ready traces.}
Prasad~\cite{Pra96} has proposed a timed
variant of his CBS~\cite{CBS}, called TCBS.  In TCBS a timeout can
force a process wishing to speak to remain idle for a specific
interval of time; this corresponds to have a priority.  TCBS also
assumes time determinism and maximal progress.  Corradini et
al.~\cite{CFP01} deal with \emph{durational actions} proposing a
framework relying on the notions of reduction and observability to
naturally incorporate timing information in terms of process
interaction.  Our definition of timed reduction barbed congruence
takes inspiration from theirs.  Corradini and Pistore~\cite{CorPist01}
have studied durational actions to describe and reason about the
performance of systems. Actions have lower and upper time bounds,
specifying their possible different durations. Their \textit{time
  equivalence} refines the untimed one.  Baeten and
Middelburg~\cite{BaMi02} consider a range  timed process algebras
within a  common framework,  related by embeddings and
conservative extensions relations. These process algebras,
$\mathrm{ACP^{sat}}$, $\mathrm{ACP^{srt}}$, $\mathrm{ACP^{dat}}$ and
$\mathrm{ACP^{drt}}$, allow the execution of two or more actions
consecutively at the same point in time, separate the execution of
actions from the passage of time, and consider actions to have no
duration.  The process algebra $\mathrm{ACP^{sat}}$ is a real-time
process algebra with absolute time, $\mathrm{ACP^{srt}}$ is a
real-time process algebra with relative time. Similarly,
$\mathrm{ACP^{dat}}$ and $\mathrm{ACP^{drt}}$ are discrete-time
process algebras with absolute time and relative time,
respectively. In these process algebra the focus is on unsuccessful
termination or deadlock.  In \cite{BaRe04} Baeten and Reniers extend
the framework of \cite{BaMi02} to model successful termination for the
relative-time case.  Laneve and Zavattaro~\cite{LaneveZ05} have
proposed a timed extension of $\pi$-calculus where time proceeds
asynchronously at the network level, while it is constrained by the
local urgency at the process level. They propose a timed bisimilarity
whose discriminating is weaker when local urgency is dropped.

%%\section{Some definitions}
%%$\isrcv{W,c}$ is the least predicate on station code which satisfies the following rules:
%%\begin{itemize}
%%\item $\isrcv{W,c}$  if $W \equiv  \timeout{\rcvc c x P} Q + R$%
%%
%%\item $\isrcv{W_1,c}$ implies $\isrcv{W_1 | W_2,c}$ and  $\isrcv{W_2 | W_1,c}$
%%
%%\item $\isrcv{W,c}$ implies $\isrcv{\nu d. W,c}$ provided $c \not = d$
%%\end{itemize}

%\bibliography{main,mh}
\bibliography{mh}
\bibliographystyle{plain}   

\appendix
\section{Technical Definitions and Proofs of some Lemmas and Propositions}

\begin{definition}[Process Environments]
A process environment, is a mapping 
from process variables to system terms. In the following 
we use $\rho$ to range over process environments. 
Given an open system term $W$ and a process environment 
$\rho$, the (possibly open) system term $W\rho$ correspond to 
the system term obtained from $W$ by replacing 
each free occurrence of any process variable $X$ with 
$\rho(X)$.
\end{definition}

\begin{lemma}
\label{lem:isrcv.procenv}
Let $\Gamma$ be a channel environment, and $W$ be an  
(open) system term whose free occurrences of process variables are time guarded. 
Then, given a channel $c$ and two process environments $\rho, \rho'$ such that 
both $(W\rho)$ and $(W\rho')$ are closed, $\isrcv{\conf{\Gamma}{W\rho, c}} = 
\isrcv{\conf{\Gamma}{W\rho', c}}$.
\end{lemma}

\proof
Note that if $\Gamma \vdash c: \cbusy$ then, for any channel environment $\rho$ such that 
$W\rho$ is closed, we have that $\isrcv{\conf{\Gamma}{(W\rho)}, c} = \ffalse$, and there is nothing else left 
to prove. 
%The statement that $\isrcv{\conf{\Gamma}{(W\rho)}, c} = \ffalse$ can be proved by a simple induction 
%on the structure of $W$. 

Suppose then that $\Gamma \vdash c: \cfree$, and let $\rho, \rho'$ be two process environments 
such that $W\rho$ and $W\rho'$ are closed. We proceed by induction on the structure of $W$. 
\begin{itemize}
\item $W = \rcvtimec c x P Q$. In this case we have $\isrcv{\conf{\Gamma}{(\rcvtimec c x P Q)\rho}, c} = 
\isrcv{\conf{\Gamma}{(\rcvtimec c x P Q)\rho'}, c} = \ttrue$,

\item $W = X$. This case is vacuous, as it contains an unguarded free occurrence of a process variable.

\item $W  = \bcastc c e P$. In this case $\isrcv{\conf{\Gamma}{(\bcastc c e P)\rho}, c} = 
\isrcv{\conf{\Gamma}{\bcastc c e {(P\rho)}}, c} = \ffalse$, and $\isrcv{\conf{\Gamma}{(\bcastc c e P)\rho'}, c} = 
\isrcv{\conf{\Gamma}{\bcastc c e {(P\rho')}}, c} = \ffalse$,

\item $W = \tau.P$, $W = \sigma.P$, $W = \matchb b P Q$, $W = \nil$ or $W = \arcv d x P$ where $d$ 
is an arbitrary (possibly equal to $c$) channel; this case is analogous to the previous one,

\item $W = P + Q$.Then we have that
\begin{eqnarray*}
\isrcv{\conf{\Gamma}{(P+Q)\rho}, c} &=& \isrcv{\conf{\Gamma}{(P\rho)}, c} \vee 
\isrcv{\conf{\Gamma}{(Q\rho)}, c}\\
&=& \isrcv{\conf{\Gamma}{(P\rho')}, c} \vee
\isrcv{\conf{\Gamma}{(Q\rho')}, c}\\
&=& \isrcv{\conf{\Gamma}{(P+Q)\rho'}, c}
\end{eqnarray*}
where the equalities $\isrcv{\conf{\Gamma}{(P\rho)}, c} = \isrcv{\conf{\Gamma}{(P\rho')}, c}$ 
and $\isrcv{\conf{\Gamma}{(Q\rho)}, c} = \isrcv{\conf{\Gamma}{(Q\rho')}, c}$ follow by 
induction.

\item $W = \fix X P$. Then we have that 
\begin{eqnarray*}
\isrcv{\conf{\Gamma}{(\fix X P)\rho}, c} &=& \isrcv{\conf{\Gamma}{(P \rho)}, c}\\
&=& \isrcv{\conf{\Gamma}{(P \rho'), c}}\\
&=& \isrcv{\conf{\Gamma}{(\fix X P) \rho', c}}
\end{eqnarray*}
Again, the equality $\isrcv{\conf{\Gamma}{(P \rho)}, c} = \isrcv{\conf{\Gamma}{(P \rho'), c}}$ 
follows by induction.

\item $W = W_1 | W_2$. This case is analogous to the case $W = P + Q$, 

\item $W = \nu c:(t,\cdot).W'$. In this case $\isrcv{\conf{\Gamma}{(\nu c:(t,\cdot).W')\rho}, c} = 
\isrcv{\conf{\Gamma}{(\nu c:(t,\cdot).W')\rho'}, c} = \ffalse$, 

\item $W = \nu d:(t,\cdot).W'$, where $d \neq c$. Then we have 
\begin{eqnarray*}
\isrcv{\conf{\Gamma}{(\nu d:(t,\cdot).W')\rho}, c} &=& 
\isrcv{\conf{\Gamma}{(W'\rho)}, c}\\
&=& \isrcv{\conf{\Gamma}{(W'\rho')}, c}\\
&=& \isrcv{\conf{\Gamma}{(\nu d:(t,\cdot).W')\rho'}, c}\rlap{\hbox
  to88 pt{\hfill\qEd}}
\end{eqnarray*}
\end{itemize}

\begin{lemma}
\label{lem:rcv-enabling.procenv}
Let $\Gamma$ be a channel environment and $W$ be an open  
system term where every free occurrence of process variables is guarded. 
Let also $c$ be a channel and $v$ be a value.
There exists an open system term $W'$ such that, for any process environment 
$\rho$ for which $(W\rho)$ is closed, then $\W'\rho$ is also closed, and 
$\conf{\Gamma}{W\rho} \trans{c?v} {W'\rho}$.
\end{lemma}
\begin{proof}

Note that if $\isrcv{\conf{\Gamma}{(W\rho), c} = \ffalse}$ for some environment 
$\rho$, it suffices to choose $W' = W$. In fact, 
by Lemma \ref{lem:isrcv.procenv} we have that 
$\isrcv{\conf{\Gamma}{W\rho'}, c} =\ffalse$ for any environment $\rho'$ 
such that $W\rho'$ is closed. By applying Rule \rulename{RcvIgn} we 
obtain the transition $\conf{\Gamma}{(W\rho')} \trans{c?v} (W\rho')$.

Therefore, suppose that $W$ is such that $\isrcv{\conf{\Gamma}{(W\rho)}, c} = \ttrue$ 
for some process environment $\rho$ (and, as a consequence of Lemma \ref{lem:isrcv.procenv}, 
$\isrcv{\conf{\Gamma}{(W\rho')}, c} = \ttrue$ for any other process environment 
$\rho'$). Note that in this case we have that $\Gamma \vdash c: \cfree$, and $W$ cannot take 
the form $\bcast c e P$, $\tau.P$, $\sigma.P$, $\matchb b P Q$, $\nil$ or $\arcv d x P$. 
We check the remaining cases, by performing an induction on $W$. In the following 
$\rho$ is an arbitrary process environment.

\begin{itemize}
\item Suppose that $W = \rcvtimec c x P Q$ for some processes $P,Q$. In this case we let 
$W' = \arcv c x P$. 
By definition $(\rcvtimec c x P Q)\rho = \rcvtimec c x {(P \rho')} {(Q \rho)}$, where $\rho' = \rho[x \mapsto x]$;
by applying Rule \rulename{Rcv} we obtain that $\conf{\Gamma}{(\rcvtimec c x {(P \rho')} {(Q \rho)}} \trans{c?v} \arcv c x (P \rho')$. 
note that the latter system term can be rewritten as $(\arcv c x P)\rho$; note in fact that the process 
environments $\rho$ and $\rho'$ differ only at the entry for variable $x$, which is bound in $\arcv c x P$. Therefore 
we have the transition $\conf{\Gamma}{(\rcvtimec c x P Q)\rho} \trans{c?v} (\arcv c x P)\rho$. 

\item Suppose that $W = P + Q$. Note that, in order to ensure that $\isrcv{\conf{\Gamma}{(P+Q)\rho}, c} = \ttrue$, 
it must be either $\isrcv{\conf{\Gamma}{(P\rho)}, c} = \ttrue$ or $\isrcv{\conf{\Gamma}{(Q\rho)},c} = \ttrue$. 
We consider only the first case, as the second one can be handled similarly. 
If $\isrcv{\conf{\Gamma}{(P\rho)}, c} = \ttrue$ then by inductive hypothesis we have that 
there exists a system term $W'$ such that $\conf{\Gamma}{(P\rho)} \trans{c?v} (W'\rho)$. 
By Rule \rulename{SumRcv}, we can derive the transition $\conf{\Gamma}{(P\rho)+(Q\rho)} \trans{c?v} 
W\rho$, which can be rewritten as $\conf{\Gamma}{(P+Q)\rho} \trans{c?v} W'\rho$.
Note also that if $\isrcv{\conf{\Gamma}{(P\rho)}, c} = \ttrue$, then 
$\isrcv{\conf{\Gamma}{(P\rho')},c} = \ttrue$ for any other process environment $\rho'$, 
as a consequence of Lemma \ref{lem:isrcv.procenv}, so that the choice of $W'$ is independent 
from the process environment.

\item Suppose that $W = \fix X P$. By inductive hypothesis, there exists a process 
$W''$ such that, for any process environment $\rho'$, $\conf{\Gamma}{P\rho'} \trans{c?v} W''\rho'$. 
In particular, let $\rho' = \rho[X \mapsto (\fix X P)\rho]$, where $\rho$ is an arbitrary process 
environment. We obtain that $\conf{\Gamma}{P\rho[X \mapsto (\fix X P)\rho]} \trans{c?v} W''\rho[X \mapsto (\fix X P)\rho]$. 
$\conf{\Gamma}{P\rho[X \mapsto (\fix X P)\rho]}  = (\{\fix X P/X\})P\rho$, and 
$W''\rho[X \mapsto (\fix X P)\rho] = (\{\fix X P/X\} W'') \rho$. Let then 
$W' = \{\fix X P/X\} W'$. It suffices to apply Rule \rulename{Rec} to obtain the transition 
$\conf{\Gamma}{(\fix X P)\rho} \trans{c?v} W'\rho$.

\item Suppose that $W = W_1 | W_2$. By inductive hypothesis there exist 
$W_1', W_2'$ such that $\conf{\Gamma}{(W_1\rho)} \trans{c?v} W_1'\rho$, and 
$\conf{\Gamma}{(W_2\rho)} \trans{c?v} W_2'\rho$. In this case we let 
$W' = W_1' | W_2'$. In fact, by Rule \rulename{RcvPar} 
it follows that $\conf{\Gamma}{(W_1\rho)} | (W_2\rho) \trans{c?v} (W_1'\rho) | (W_2'\rho)$, 
or equivalently $\conf{\Gamma}{(W_1 | W_2)\rho} \trans{c?v} (W_1' | W_2')\rho$.

\item Finally, suppose $W = \nu d:(n,v).W_1$, where $d \neq c$. By inductive hypothesis 
we have that $\conf{\Gamma[d \mapsto (n,v)]}{(W_1\rho)} \trans{c?v} W'\rho$ for some 
$W'$. Now it suffices to apply Rule \rulename{ResI} to obtain $\conf{\Gamma}{(W\rho)} \trans{c?v} (W'\rho)$.\qedhere
\end{itemize}
\end{proof}

\paragraph{\textbf{Proof of Lemma \ref{lem:rcv-enabling}.}}
\label{proof:rcv-enabling}
Let $\conf{\Gamma}{W}$ be a configuration.
First note that $W$ is a closed system term, hence $W\rho = W$ for any 
process environment $\rho$. Given an arbitrary channel $c$ and 
an arbitrary value $v$, Lemma \ref{lem:rcv-enabling.procenv} ensures that 
there exists a system term $W'$ such that
$\conf{\Gamma}{W} \trans{c?v} W'$. 

It remains to show that whenever 
$\conf{\Gamma}{W} \trans{c?v} W'$ for some $W'$, if $\isrcv{\conf{\Gamma}{W}, c} = \ttrue$ 
then $W' \neq W$; conversely, 
if $\isrcv{\conf{\Gamma}{W}, c} = \ffalse$ then $W' = W$. 
This last statement can be proved by performing an induction on the proof of the derivation 
$\conf{\Gamma}{W} \trans{c?v} W'$; the proof is relatively simple, and the details are 
left to the reader. 

The case where $\isrcv{\conf{\Gamma}{W}, c} = \ttrue$ is slightly more complicated. 
In practice, we define a function $\rcvno{\cdot, c}$ which maps any system term into 
its number of active receivers along channel $c$ and we show that, whenever 
$\conf{\Gamma}{W} \trans{c?v} W'$, then $\rcvno{W'} > \rcvno{W}$. As an immediate 
consequence, $W' \neq W$. 
Formally, the function $\rcvno{\cdot, c}$ is defined inductively on the structure of system terms, by letting 
for any process $P$ and system terms $W_1, W_2$, 
\begin{enumerate}[label=\({\alph*}]
\item $\rcvno{P, c} = 0$, 
\item $\rcvno{\arcv d x P,c} = 1$ if $d = c$, $0$ otherwise,  
\item $\rcvno{\nu d.(W_1), c} = \rcvno{W_1, c}$, when $d \neq c$, 
\item $\rcvno{(W_1 | W_2), c} = \rcvno{W_1, c} + \rcvno{W_2, c}$.
\end{enumerate}
We proceed by induction on the proof of the derivation $\conf{\Gamma}{W} \trans{c?v} W'$.

\begin{itemize}
\item The last rule applied in the proof of $\conf{\Gamma}{W} \trans{c?v} W'$ is Rule 
\rulename{Rcv}. It follows that $W = \rcvtimec c x P Q$ for some processes $P, Q$, 
hence $\rcvno{W, c} = 0$.
Further, $W' = \arcv c x P$, which leads to $\rcvno{W', c} = 1$;

\item the last Rule applied in the proof of $\conf{\Gamma}{W} \trans{c?v} W'$ is \rulename{SumRcv}; 
Then $W= P + Q$ for some processes $P,Q$ such that $\isrcv{\conf{\Gamma}{P}, c} = \ttrue$, 
and $\conf{\Gamma}{P} \trans{c?v} W'$.  By definition we have that $\rcvno{P+Q, c} = 0$; 
also, $\rcvno{P, c} = 0$, hence by inductive hypothesis $\rcvno{W', c} > 0$, as we wanted to prove; 
the symmetric case of Rule \rulename{SumRcv} is handled similarly.

\item the last rule applied in the proof of $\conf{\Gamma}{W} \trans{c?v} W'$ is Rule 
\rulename{Rec}; this case is analogous to the previous one, 

\item the last rule applied in the proof of $\conf{\Gamma}{W} \trans{c?v} W'$ is Rule \rulename{ResV}; 
then $W = \nu d.(W_1)$  and $W' = \nu d.(W_1')$ for some $d \neq c$, $W_1$ and $W_1'$ such that 
$\conf{\Gamma}[d\mapsto (\cdot, \cdot)]{W_1} \trans{c?v} W_1'$. In this case we have that 
$\rcvno{\nu d.(W_1), c} = \rcvno{W_1, c} > \rcvno{W_1', c} = \rcvno{\nu d.(W_1'), c}$, where 
the inequality $\rcvno{W_1, c} > \rcvno{W_1', c}$ follows from the inductive hypothesis, 

\item the last case to analyse is the one in which Rule \rulename{RcvPar} has been 
applied last in the proof of $\conf{\Gamma}{W} \trans{c?v} W'$. Then 
$W = W_1 | W_2$ for some $W_1, W_2$ such that $\conf{\Gamma}{W_1} \trans{c?v} 
W_1'$ and $\conf{\Gamma}{W_2} \trans{c?v} W_2'$. Further, since 
we are assuming that $\isrcv{\conf{\Gamma}{W_1 | W_2, c}} = \ttrue$, 
then either $\isrcv{\conf{\Gamma}{W_1}, c} = \ttrue$ or $\isrcv{\conf{\Gamma}{W_2}, c} = \ttrue$. 
Without loss of generality, suppose that $\isrcv{\conf{\Gamma}{W_1}, c} = \ttrue$. 
Note that in this case, if $\isrcv{\conf{\Gamma}{W_2}, c} = \ffalse$ then we know that 
$W_2' = W_2$, hence $\rcvno{W_2', c} = \rcvno{W_2, c}$. Otherwise, by inductive hypothesis 
it follows that $\rcvno{W_2', c} >\rcvno{W_2, c}$. In any case, we obtain that $\rcvno{W_2', c} \geq 
\rcvno{W_2, c}$. Also, by inductive hypothesis we have that $\rcvno{W_1', c} > \rcvno{W_1, c}$. 
By these two statements, and the definition of $\rcvno{\cdot, c}$, 
it follows that $\rcvno{W_1 | W_2, c} = \rcvno{W_1, c} + \rcvno{W_2, c} > \rcvno{W_1', c} + \rcvno{W_2, c} = 
\rcvno{W_1' | W_2', c}$.\qedhere
\end{itemize}

\begin{lemma}
\label{lem:time.auxiliary}
Suppose that $\conf{\Gamma}{W} \trans{\sigma} W'$;
\begin{enumerate}[label=(\roman*)]
\item \label{item:time.auxiliary.1} if $W = P + Q$ for some processes $P,Q$ then there exists two processes 
$P', Q'$ such that $\conf{\Gamma}{P} \trans{\sigma} P'$, $\conf{\Gamma}{Q} \trans{\sigma} Q'$ and 
$W' = P' + Q'$,
\item \label{item:time.auxiliary.2} if $W = W_1 | W_2$ for some $W_1, W_2$, then there exists two system terms 
$W_1', W_2'$ such that $W' = W_1' | W_2'$, $\conf{\Gamma}{W_1} \trans{\sigma} W_1'$ 
and $\conf{\gamma}{W_2} \trans{\sigma} W_2'$.
\end{enumerate}
\end{lemma}

\begin{proof}
Both statements can be proved by induction on the structure of $W$. We only provide the 
details for \eqref{item:time.auxiliary.1}, since the proof for \eqref{item:time.auxiliary.2} is identical in style.

\begin{itemize}
\item First note that if $W$ is a basic process, that is, it has
  either the form $\nil$, $\bcastc c e P$, $\matchb b P Q$, 
 $\rcvtimec  c x P Q$, $\tau.P$, $\fix X P$ or $\sigma.P$ then there is nothing to prove, as the assumption that 
$W = P + Q$ for some processes $P, Q$ is not valid;
\item suppose then that $W = P + Q$ for some processes $P,Q$, and that $\conf{\Gamma}{P+Q} \trans{\sigma} W'$. 
By inspecting the rules of the intensional semantics, it is clear that the last Rule applied in a proof of the transition 
above is \rulename{SumTime}. Thus, there exist processes $P_1, Q_1, P_1', Q_1'$ such that $P + Q = P_1 + Q_1$, 
$W' = P_1' + Q_1'$, $\conf{\Gamma}{P_1} \trans{\sigma} P_1'$ and $\conf{\Gamma}{Q_1} \trans{\sigma} Q_1'$. 
We need to show that there exist two processes $P', Q'$ such that $\conf{\Gamma}{P} \trans{\sigma} P'$, 
$\conf{\Gamma}{Q} \trans{\sigma} Q'$ and $P' + Q' = P_1' + Q_1'$. Note that the assumption 
$P + Q = P_1 + Q_1$ leads to three possible cases: 
\begin{enumerate}
\item there exists a process $R$ such that $P_1 = P + R$, $Q = R + Q_1$; In this case we can apply the inductive 
hypothesis to the system term $P_1$ (note that $P_1$ is a smaller term than $P + Q$, as $P + Q = P_1 + Q_1$). 
Thus the transition $\conf{\Gamma}{P_1} \trans{\sigma} P_1'$ ensures that there exist two system term $P', R'$ 
such that $\conf{\Gamma}{P} \trans{\sigma} P'$, $\conf{\Gamma}{R} \trans{\sigma} R'$ and $P_1' = P' + R'$. 
Further, by applying Rule \rulename{SumTime} to the transitions $\conf{\Gamma}{R} \trans{\sigma} R'$ 
and $\conf{\Gamma}{Q_1} \trans{\sigma} Q_1'$, we obtain $\conf{\Gamma}{R + Q_1} \trans{\sigma} \conf{\Gamma}{R' + Q_1'}$. 
By letting $Q' = R' + Q_1'$, we can rewrite this last transition as $\conf{\Gamma}{Q} \trans{\sigma} Q'$. 
Finally notice that we have $W = P_1' + Q_1' = (P' + R') + Q_1' = P' + (R' + Q_1') = P' + Q'$, as we wanted to prove, 
\item otherwise $P = P_1$ and $Q = Q_1$; this case is trivial, as it suffices to choose $P' = P_1', Q' = Q_1'$,
\item the last case possible is that there exists a process $R$ such that $P = P_1 + R$, $Q_1 = R + Q$; the 
proof here is symmetrical to the first case, as now it is necessary to apply the inductive hypothesis to $Q_1$, 
rather than to $P_1$,
\end{enumerate}

\item the last remaining cases are those in which either $W = \nu c.W_1$ or $W = W_1 | W_2$. Again, 
these cases invalidate the hypothesis that $W$ is a non-deterministic choice of processes, hence there 
is nothing to prove.\qedhere
\end{itemize}
\end{proof}

\paragraph{\textbf{Proof of Proposition \ref{prop:time-determinism}.}}
\label{proof:time-determinism}
We proceed by induction on the proof of the derivation $\confC \trans{\sigma} W_1$. 

\begin{itemize}
\item The last rule applied in the derivation $\confC \trans{\sigma} W_1$ 
is rule \rulename{EndRcv}. Then $\confC = \conf{\Gamma}{\arcv c x P}$ 
for some channel $c$, process $P$, channel environment $\Gamma$ 
for which $\Gamma \vdash_{\mathrm{t}} c: 1$ and $\Gamma \vdash_{\mathrm{v}} c: w$ for some 
closed value $w$. Also $W_1 =\{w/x\}P$. 
Suppose now that 
$\confC \trans{\sigma} W_2$ for some system term $W_2$. 
By inspecting the rules of the intensional semantics we have that 
the only rule which could have been applied to infer this transition 
is again Rule \rulename{EndRcv}. It follows that $W_2 = W_1 = \{w/x\}P$,

\item the cases where the last rule applied in the proof of $\confC \trans{\sigma} W_1$ 
is either \rulename{TimeNil}, \rulename{Sleep}, \rulename{ActRcv} or \rulename{Timeout} 
can be proved similarly to the previous one, 

\item if the last rule applied in the proof of $\confC \trans{\sigma} W_1$ 
is \rulename{SumTime}, then $\confC = \conf{\Gamma}{P + Q}$ 
for some processes $P,Q$. By Lemma \ref{lem:time.auxiliary}\eqref{item:time.auxiliary.1} we also know that 
$W_1 = P_1 + Q_1$ for some $P_1,Q_1$ such that 
$\conf{\Gamma}{P} \trans{\sigma} P_1, \conf{\Gamma}{Q} \trans{\sigma} Q_1$. 

Suppose that $\confC \trans{\sigma} W_2$ for some $W_2$. Then again, Lemma 
\ref{lem:time.auxiliary}\eqref{item:time.auxiliary.1} leads to 
$W_2 = P_2 + Q_2$ for some $P_2, Q_2$ such that 
$\conf{\Gamma}{P} \trans{\sigma} P_2$ and $\conf{\Gamma}{Q} \trans{\sigma} 
Q_2$. But by the inductive hypothesis we have that $P_1 = P_2$, $Q_1 = Q_2$. 
Hence $W_2 = P_2 + Q_2 = P_1 + Q_1 = W_1$,

\item if Rule \rulename{Rec} has been applied last, then $W = \fix X P$ for some process 
variable $X$ and process $P$; further, $\conf{\Gamma}{\{\fix X P/X\} P} \trans{\lambda} W_1$. 
Suppose now hat $\conf{\Gamma}{\fix X P} \trans{\sigma} W_2$ for some $W_2$; then again 
the last rule applied has been \rulename{Rec}, so that $conf{\Gamma}{\{\fix X P/X\} P} \trans{\lambda} W_2$. 
Now, by the inductive hypothesis, we get that $W_1 = W_2$, 

\item the case where \rulename{ResV} is the last one in the derivation $\confC \trans{\sigma} W_1$ 
is similar in style to the previous one, and is therefore left to the reader, 

\item the last case is the one in which the last rule applied for deriving $\confC \trans{\sigma} W_1$ 
is Rule \rulename{TimePar}; the proof in this case is analogous to the one where $\confC = \conf{\Gamma}{P+Q}$, 
using Lemma \ref{lem:time.auxiliary}\eqref{item:time.auxiliary.2} instead of \ref{lem:time.auxiliary}\eqref{item:time.auxiliary.1}.\qed
\end{itemize}

\paragraph{\textbf{Proof of Proposition \ref{prop:maximal-progress}.}}
\label{proof:maximal-progress}
By induction on the proof of the transition. We only supply the details for the most interesting cases.
\begin{itemize}
\item The last Rule applied in the proof of the derivation 
$\confC \trans{\sigma} W_1$ is Rule \rulename{TimeOut}. 
It follows that $\confC = \conf{\Gamma}{\rcvtimec c x P Q }$ for some $\Gamma$, 
channel $c$ and processes $P,Q$ such that $\Gamma \vdash c: \cfree$. 
By inspecting the rules of the intensional semantics we note that no Rule can 
be applied to obtain a transition of the form $\confC \trans{\sndac c v} W_2$, 
nor a transition of the form $\confC \trans{\tau} W_2$; for this last case, 
note in fact that a $\tau$-action can be inferred for a configuration of the form  
$\conf{\Gamma}{\rcvtimec c x P Q}$ only via Rule \rulename{RcvLate}, 
which however requires $\Gamma \vdash c: \cbusy$. This is in contrast with 
our assumption that $\Gamma \vdash c: \cfree$.

\item The last Rule applied in the proof of the transition 
$\confC \trans{\sigma} W_1$ is Rule \rulename{SumTime}. 
Then $\confC = \conf{\Gamma}{P+Q}$ for some $P,Q$ such that 
$\conf{\Gamma}{P} \trans{\sigma} P', \conf{\Gamma}{Q} \trans{\sigma} Q'$ 
and $W_1 = P' + Q'$. 

We show, by contradiction, that $\conf{\Gamma}{P+Q} \ntrans{\sndac c v}$ for any channel 
$c$ and value $v$, and $\conf{\Gamma}{P+Q} \ntrans{\tau}$. 
So suppose that $\conf{\Gamma}{P+Q} \trans{\lambda} W_2$  for some system term $W_2$ 
and action $\lambda \in \{\tau, \sndac c v \;|\; c \in \chanset, v \text{ closed value}\}$. Then the last rule applied in the proof of such a 
transition is either Rule \rulename{Sum} or its symmetric counterpart. 
In the first case we have that $\conf{\Gamma}{P} \trans{\lambda} W_2$, but this 
contradicts the inductive hypothesis; $\conf{\Gamma}{P} \trans{\sigma} P'$ implies 
$\conf{\Gamma}{P} \ntrans{\sndac c v}$. 
Similarly, in the second case $\conf{\Gamma}{Q} \trans{\lambda} W_2$, which 
contradicts the inductive hypothesis applied to the transition $\conf{\Gamma}{Q} \trans{\sigma} 
Q'$. Therefore $\conf{\Gamma}{P+Q} \ntrans{\lambda}$.\qed
\end{itemize}

\paragraph{\textbf{Proof of Proposition \ref{prop:exposure-consistency}.}}
\label{proof:exposure-consistency}
The proof is performed by induction on the structure of the proof of the derivation 
$\conf{\Gamma_1}{W} \trans{\lambda} W'$. Again, we only consider the most interesting 
cases:
\begin{itemize}
\item The last rule applied in the proof of the derivation $\conf{\Gamma_1}{W} \trans{\lambda} 
W'$ is Rule \rulename{Rcv}. Then $\lambda = \rcva c v$ for some channel $c$ 
and value $v$, $\Gamma_1 \vdash c: \cfree$, $W = 
\rcvtimec c x P Q$ for some $P,Q$ and $W' = \arcv c x P$. By Hypothesis we have that 
$\Gamma_2 \vdash c: \cfree$, so that $\conf{\Gamma_2}{\rcvtimec c x P Q} \trans{\rcva c v} 
\arcv c x P$.
\item The last Rule applied in the proof of $\conf{\Gamma_1}{W} \trans{\lambda} W'$ is 
Rule \rulename{RcvLate}. Then $\lambda = \tau$, $\Gamma_1 \vdash c: \cbusy$ for some channel $c$, 
$W = \rcvtimec c x P Q$ and $W' = \arcv c x {\{\err/x\}P}$. By hypothesis 
$\Gamma_2 \vdash c: \cbusy$, so that Rule \rulename{RcvLate} can be applied leading 
to $\conf{\Gamma_2}{\rcvtimec c x P Q} \trans{\lambda} \arcv c x {\{\err/x\}P}$.
\item The last rule applied in the proof of $\conf{\Gamma_1}{W} \trans{\lambda} W'$ is 
Rule \rulename{Then}. Then $W = \matchb b P Q$ for some $b$ such that 
$\interpr{b}_{\Gamma_1} = \ttrue$, $\lambda = \tau$ and 
$W' = \sigma.P$. Here it is necessary to make a case analysis 
on the form of the boolean expression $b$; the most interesting case, and the only one 
which we analyse, is $b = \expsd c$ for some channel $c$. 
Since $\interpr{b}_{\Gamma_1} = \ttrue$ then $\Gamma_1 \vdash c: \cbusy$. 
By hypothesis it follows that $\Gamma_2 \vdash c: \cbusy$, therefore 
$\interpr{b}_{\Gamma_2} = \ttrue$. Now we can apply Rule \rulename{Then} to 
infer $\conf{\Gamma_2}{\matchb b P Q} \trans{\tau} \sigma.P$.

\item The last rule applied in the proof of 
$\conf{\Gamma_1}{W} \trans{\lambda} W'$ is Rule \rulename{Sync}. 
It follows that $\lambda = \sndac c v$ for some channel $c$ and value $v$,  
$W = W_1 | W_2$ and $W' = W_1' | W_2'$ for some $W_1, W_2, W_1', W_2'$ 
such that $\conf{\Gamma_1}{W_1} \trans{\sndac c v} W_1'$, $\conf{\Gamma_2}{W_2} 
\trans{\sndac c v} W_2'$. Then by inductive hypothesis we have that 
$\conf{\Gamma_2}{W_1} \trans{\sndac c v} W_1'$ and $\conf{\Gamma_2}{W_2} 
\trans{\rcva c v} W_2'$. An application of Rule \rulename{Sync} gives 
$\conf{\Gamma_2}{W} \trans{\sndac c v} W'$.\qed
\end{itemize}

\paragraph{\textbf{Proof of Proposition \ref{prop:parallel-components} \eqref{par3}.}}
\label{proof:parallel-components}
Note that the proof of this statement uses Lemma, \ref{prop:parallel-components}\eqref{par2}, 
which can be proved independently. 
For the if implication,suppose that $\conf{\Gamma}{W_1} \trans{\sndac c v} W_1'$ 
and $\conf{\Gamma}{W_2} \trans{\rcva c v} W_2'$. Then, by an application of 
Rule \rulename{Sync} we obtain that $\conf{\Gamma}{W_1 | W_2} \trans{\sndac c v} 
W_1' | W_2'$. Similarly, if $\conf{\Gamma}{W_1} \trans{\rcva c v} W_1'$ and 
$W_2 \trans{\sndac c v} W_2'$, we can obtain the transition $\conf{\Gamma}{W_1 | W_2} 
\trans{\sndac c v} W_1' | W_2'$ using the symmetric counterpart of Rule \rulename{Sync}. 

For the only if implication, suppose that $\conf{\Gamma}{W_1 | W_2} \trans{\sndac c v} 
W'$. Note that we can rewrite $W_1 | W_2$ as $\prod_{i=1}^k P_k$ for some $k \geq 2$. 
We proceed by induction on $k$. 
\begin{itemize}
\item $k = 2$. Then $W_1 = P_1$, $W_2 = P_2$. The last rule applied in the derivation of 
$\conf{\Gamma}{P_1 | P_2} \trans{\sndac c v} W'$ is either Rule \rulename{sync} or 
its symmetric counterpart. 
In the first case we obtain that $\conf{\Gamma}{P_1} \trans{\sndac c v} P_1'$, 
$\conf{\Gamma}{P_2} \trans{\rcva c v} P_2'$ and $W' = P_1' | P_2'$, so that there is 
nothing to prove. The second case is analogous.

\item $k > 2$. Suppose that the statement is true for any index $i \leq k$. 
Again, the last rule applied in the proof of the transition $\conf{\Gamma}{W_1 | W_2} 
\trans{\sndac c v} W'$ is either Rule \rulename{Sync} or its symmetric counterpart. 
We consider only the first case, as the second one is treated similarly. 
If Rule \rulename{Sync} has been applied last, then there exist 
two system terms $W_a, W_b$ such that $W_1 | W_2 = W_a | W_b$ and 
$\conf{\Gamma}{W_a} \trans{\sndac c v} W_a'$, $\conf{\Gamma}{W_b} \trans{\sndac 
c v} W_b'$ and $W' = W_a' | W_b'$. 
Since $W_a | W_b = W_1 | W_2$, we have three possible cases: 

\begin{itemize}
\item $W_1 = W_a | W_x$, $W_b = W_x | W_2$ for some system term $W_x$. 
Then we can apply Proposition \ref{prop:parallel-components}
\eqref{par2}
%\footnote{Note that this requires that \eqref{par2} has been 
%proved before the current statement, \eqref{par3}} 
to the transition 
$\conf{\Gamma}{W_x | W_2} \trans{\rcva c v} W_b'$ to show 
that $\conf{\Gamma}{W_x} \trans{\rcva c v} W_x'$, 
$\conf{\Gamma}{W_2} \trans{\rcva c v} W_2'$ for some $W_x', W_2'$ 
such that $W_b' = W_x' | W_2'$. Now we can apply Rule \rulename{Sync} 
to the transitions $\conf{\Gamma}{W_a} \trans{\sndac c v} W_a'$ and 
$\conf{\Gamma}{W_x} \trans{\rcva c v} W_x'$ to infer 
$\conf{\Gamma}{W_1} \trans{\rcva c v} W_a' | W_x'$. 
Let $W_1' = W_a' | W_x'$. Then we have 
\[
W' = W_a' | W_b' = W_a' | W_x' | W_2' = W_1' | W_2'.
\]

\item $W_a = W_1, W_b = W_2$. In this case there is nothing to prove, 
as it suffices to choose 
$W_1' = W_a', W_2' = W_b'$ to obtain the result.

\item $W_a = W_1 | W_x$, $W_2 = W_x | W_b$ for some $W_x$. 
By the inductive hypothesis we obtain that either 
\begin{itemize}
\item $\conf{\Gamma}{W_1} \trans{\sndac c v} W_1'$, 
$\conf{\Gamma}{W_x} \trans{\rcva c v} W_x'$ for some $W_1', W_x'$ such 
that $W_a' = W_1' | W_x'$, or 

\item $\conf{\Gamma}{W_1} \trans{\rcva c v} W_1'$, 
$\conf{\Gamma}{W_x} \trans{\sndac c v} W_x'$ for some $W_1', W_x'$ such 
that $W_a' = W_1' | W_x'$.
\end{itemize}

We consider only the first case. In this case we can apply Rule \rulename{rcvPar} 
to the transitions $\conf{\Gamma}{W_x} \trans{\rcva c v} W_x'$ and 
$\conf{\Gamma}{W_b} \trans{\rcva c v} W_b'$ to obtain 
$\conf{\Gamma}{W_2} \trans{\rcva c v} W_x' | W_b'$. Let 
$W_2' = W_x' | W_b'$. Then we have proved that 
$\conf{\Gamma}{W_1} \trans{\sndac c v} W_1'$, 
$\conf{\Gamma}{W_2} \trans{\rcva c v} W_2'$; further 
we have that 
\[
W' = W_a' | W_b' = W_1' | W_x' | W_b' = W_1' | W_2'
\]
as we wanted to prove.\qed
\end{itemize}
\end{itemize}

\paragraph{\textbf{Proof of Lemma \ref{lem:taus-Gamma}.}}
\label{proof:taus-gamma}
We first prove that if $\conf{\Gamma}{W} \exttrans{\tau} \conf{\Gamma'}{W'}$ 
then $\Gamma \leq \Gamma'$. 
Note that such a transition could have been inferred in two different ways: 
\begin{itemize}
\item via an application of Rule \rulename{TauExt}, from which it follows that 
$\Gamma' = \gupd{\tau}{\Gamma} = \Gamma$, or 
\item via an application of Rule \rulename{Shh}, applied to a transition of 
the form $\conf{\Gamma}{W} \trans{c!v} W'$; it follows that $\Gamma' = \gupd{c!v}{\Gamma}$, 
from which we obtain that $\Gamma \leq \Gamma'$.
\end{itemize}
Now suppose that $\conf{\Gamma}{W} \extTrans{\tau} \conf{\Gamma'}{W'}$. By definition, 
there exists an integer $n \geq 0$ such that $\conf{\Gamma}{W}  = \conf{\Gamma_0}{W_0} 
\exttrans{\tau} \conf{\Gamma_1}{W_1} \exttrans{\tau} \cdots \exttrans{\tau} \conf{\Gamma_n}{W_n} = 
\conf{\Gamma'}{W'}$. By applying the result proved above to each step in this sequence, 
we obtain $\Gamma = \Gamma_0 \leq \Gamma_1 \leq \cdots \leq \Gamma_n = \Gamma'$, hence $\Gamma \leq \Gamma'$.
\hfill\qed

\begin{corollary}
\label{cor:iota.actions}
For any channel $c$, $\conf{\Gamma}{W} \extTrans{\iota(c)}$ 
implies $\conf{\Gamma}{W} \exttrans{\iota(c)}$.
\end{corollary}

\begin{proof}
By Definition, $\conf{\Gamma}{W} \extTrans{}
\conf{\Gamma'}{W'}\exttrans{\iota(c)}$ 
for some $\Gamma', W'$. Since, $\conf{\Gamma'}{W'} \exttrans{\iota(c)}$ 
we obtain that $\Gamma' \vdash c: \cfree$. Now Lemma \ref{lem:taus-Gamma} 
gives $\Gamma \leq \Gamma'$, hence $\Gamma \vdash c: \cfree$. Therefore we 
can apply Rule \rulename{Idle} of the extensional semantics and derive 
$\conf{\Gamma}{W} \exttrans{\iota(c)} \conf{\Gamma}{W}$.
\end{proof}

\paragraph{\textbf{Proof of Lemma \ref{lem:channel-exposure}.}}
\label{proof:channel-exposure}
Suppose $\conf{\Gamma_1}{W_1} \approx \conf{\Gamma_2}{W_2}$. 
If $\Gamma_1 \vdash c: \cfree$ then by definition of 
Rule \rulename{Idle} of Table~\ref{tab:extensional} it follows that
 $\conf{\Gamma_1}{W_1} 
\exttrans{\iota(c)}$. As $\conf {\Gamma_1} {W_1} \approx \conf 
{\Gamma_2} {W_2}$, 
it follows that $\conf{\Gamma_2}{W_2} \extTrans{\iota(c)}$. 
From Corollary \ref{cor:iota.actions} we have that $\conf{\Gamma_2}{W_2} 
\exttrans{\iota(c)}$, and by the definition of Rule 
\rulename{Idle} that $\Gamma_2 \vdash c: \cfree$. \hfill\qed

\paragraph{\textbf{Proof of Lemma \ref{lem:wf.preserved}.}}
\label{proof:wf.preserved}
We have to show that if $\confC$ is well-formed and 
$\confC \trans{\lambda} W'$, then $\confC' = \conf{\gupd{\lambda}{\Gamma}}{W'}$ 
is also well-formed. We provide the details of the most interesting cases 
of a rule induction on the proof of the aforementioned transition.
\begin{itemize}
\item The last rule applied is Rule \rulename{Rcv}. Then $\lambda = \rcva c v$ for 
some channel $c$ and closed value $v$. Further, $\confC = \conf{\Gamma}{\rcvtimec c x P Q}$, 
$W' = \arcv c x P$ and $\gupd{\rcva c v}{\Gamma} \vdash c: \cbusy$. 
The second equation in Definition \ref{def:wellformed} ensures that $\confC' \in \text{Wnets}$,
\item the last rule applied is Rule \rulename{EndRcv}; in this case $\lambda = \sigma$, 
$W = \arcv c x P$ for some $c$ such that $\Gamma \vdash c: \cbusy$, and $W' = \{w/x\}P$, 
where $w$ is the closed value such that $\Gamma \vdash_{\mathrm{v}} c: w$. It follows from the 
first equation in Definition \ref{def:wellformed} that $\confC' = \conf{\gupd{\sigma}{\Gamma}}{W'}$ is 
well formed,
\item the last rule applied is Rule \rulename{ActRcv}. In this case $W = W' = \arcv c x P$ for some 
$c$ such that $\Gamma \vdash_{\mathrm{t}} c: n$, where $n > 1$. To show that 
$\confC' = \conf{\gupd{\sigma}{\Gamma}}{\arcv c x P}$, it suffices to prove that 
$\gupd{\sigma}{\Gamma} \vdash c: \cbusy$; but this is true, since by Definition of 
$\gupd{\sigma}{\cdot}$ we have that $\gupd{\sigma}{\Gamma} \vdash_{\mathrm{t}} c: n -1$, 
and now $n -1 > 0$,

\item the last rule applied is Rule \rulename{Sync}. Then $\lambda = \sndac c v$, 
$W = W_1 | W_2$, $W' = W_1' | W_2'$ for some $W_1, W_2, W_1',W_2'$ such that 
$\conf{\Gamma}{W_1} \trans{\sndac c v} W_1'$, $\conf{\Gamma}{W_2} \trans{\rcva c v} 
W_2'$ and $W' = W_1' | W_2'$. By inductive hypothesis the configurations 
$\confC_1 = \conf{\gupd{\sndac c v}{\Gamma}}{W_1'}$ and 
$\confC_2 = \conf{\gupd{\rcva c v}{\Gamma}}{W_2'}$ are well formed, 
so by the third equation in Definition \ref{def:wellformed} we have that 
$\confC' \in \text{Wnets}$\footnote{Recall 
that $\gupd{\sndac c v}{\Gamma} = \gupd{\rcva c v}{\Gamma}$.}.\hfill\qed
\end{itemize}

\paragraph{\textbf{Proof of Proposition \ref{prop:noti2sigma}.}}
\label{proof:noti2sigma}
Let $\conf{\Gamma}{W}$ be a well-formed configuration. 
We give the details of the most important cases of a structural 
induction performed on the structure of a system term $W$. 

\begin{itemize}
\item $W = \bcastc c v P$, or $W = \tau.P$; this case is vacuous, since 
by definition of instantaneous reductions $\conf{\Gamma}{W} \red_{i}$,
\item $W = \sigma.P$; this case is trivial, since by applying Rule 
\rulename{Sleep} we infer $\conf{\Gamma}{W} \trans{\sigma} P$, hence 
$\conf{\Gamma}{W} \red_{\sigma} \conf{\gupd{\sigma}{\Gamma}}{P}$, 
\item $W = \arcv c x P$. By definition of well-formed networks 
we have that $\Gamma \vdash c: \cbusy$. Then there are 
two possible cases:
\begin{itemize}
\item $\Gamma \vdash_{\mathrm{t}} c: 1$ and $\Gamma \vdash_{\mathrm{v}} c: v$ for some 
value $v$. We can apply Rule \rulename{EndRcv} to infer the 
transition $\conf{\Gamma}{\arcv c x P} \trans{\sigma} \{v/x\}P$, 
which in turns gives the reduction $\conf{\Gamma}{\arcv c x P} \red_{\sigma} 
\conf{\gupd{\sigma}{\Gamma}}{\{v/x\}P}$, 
\item $\Gamma \vdash_{\mathrm{t}} c: n$ for some $n > 1$; in this case we can apply Rule 
\rulename{ActRcv} to infer $\conf{\Gamma}{\arcv c x P} \trans{\sigma} \arcv c x P$, 
leading to $\conf{\Gamma}{\arcv c x P} 
\red_{\sigma} \conf{\gupd{\sigma}{\Gamma}}{\arcv c x P}$.
 \end{itemize} 
\item $W = \fix X P$. Recall that in this case every occurrence of the process 
variable $X$ in $P$ is (time) guarded, so that we can apply the inductive 
hypothesis to the term  $\{\fix X P/X\}P$. Now suppose that 
$\conf{\Gamma}{\fix X P} \not\red_i$. Then it follows that 
$\conf{\Gamma}{\{\fix X P/X\}P} \not\red_i$, and by inductive 
hypothesis $\conf{\Gamma}{\{\fix X P/X\}P} \red_{\sigma}$. Now it is 
easy to show that $\conf{\Gamma}{\fix X P} \red_{\sigma}$.
\item $W = P + Q$. Suppose that $\conf{\Gamma}{P + Q} \not\red_i$. 
That is, $\conf{\Gamma}{P} \not\red_i$, $\conf{\Gamma}{Q} \not\red_i$, 
By inductive hypothesis we have that $\conf{\Gamma}{P} \trans{\sigma} P'$, 
$\conf{\Gamma}{Q} \trans{\sigma} Q'$ for some $P',Q'$. It follows from 
Rule \rulename{SumTime} that $\conf{\Gamma}{P + Q} \trans{\sigma} P' + Q'$, 
hence $\conf{\Gamma}{P+Q} \red_{\sigma} \conf{\gupd{\sigma}{\Gamma}}{P'+Q'}$.\qed
\end{itemize}

\begin{proposition}
\label{prop:welltimed}
%An environment $\rho$ is a partial mapping from process variables to closed terms. 
%Given a term $W$ and an environment $\rho$, we denote with $W\rho$ the 
%system term obtained by substituting each free occurrence of any process variable 
%$X$ with $\rho(X)$, if defined. 
%
For any channel environment $\Gamma$, (possibly open) process $P$ and process environment $\rho$ such that 
${P\rho}$ is closed, then $\conf{\Gamma}{P\rho}$ is well-timed.
\end{proposition}

\begin{proof}
%
%\paragraph{\textbf{Proof of Proposition \ref{prop:welltimed}.}}
We give the details of the most important cases of an induction performed on 
the structure of the process $W$. In the following we assume that 
$\rho$ is a process environment such that $W\rho$ is closed; recall that we are assuming that 
free occurrences of process variables are time guarded in $W$.
\begin{itemize}
\item $W = \rcvtimec c x P Q$. Then we have that $\conf{\Gamma}{(\rcvtimec c x P Q)\rho}
\not\red_i$; it follows that $\conf{\Gamma}{(\rcvtimec c x P Q)\rho}$ is well-timed. 
\item $W = X$ for some process variable $X$; this case is vacuous, since it violates the 
assumption that free occurrences of process variables are (time) guarded in $W$,
\item $W = \fix X P$ for some process $P$. Let $\rho'$ be the environment defined as 
$\rho[X \mapsto (\fix X P)\rho]$. By inductive hypothesis we have that $\conf{\Gamma}{P\rho'}$ 
is well-timed. Further, by definition $P\rho' = (\{\fix X P/X\}P)\rho$. Now note that  
$\conf{\Gamma}{(\fix X P)\rho} \red^h \confC'$ if and only if $\conf{\Gamma}{(\fix X P/X\}P)\rho} 
\red^h \confC'$. 
It follows that $\conf{\Gamma}{(\fix X P)\rho}$ is well-timed. 
\item $W = P + Q$. Suppose that both $(P+Q)\rho$ is closed; 
that is, both $P\rho$ and $Q\rho$ are closed. By inductive hypothesis they are well timed, 
meaning that there exists $k_P \geq 0$ such that whenever 
$\conf{\Gamma}{P\rho} \red^{h} \conf{\Gamma'}{P'}$ then 
$h \leq k_P$; similarly, there exists $k_Q \geq 0$ such that 
whenever $\conf{\Gamma}{Q\rho} \red^h \conf{\Gamma'}{Q'}$ for some $h$, then 
$h \leq k_Q$. Choose $k = \max(k_P,k_Q)$. It is easy to show that 
whenever $\conf{\Gamma}{(P+Q)\rho} \red^h \conf{\Gamma'}{W'}$ then either 
$\conf{\Gamma}{P\rho} \red^h \conf{\Gamma'}{W'}$, in which 
case $h \leq k_P \leq k$, or $\conf{\Gamma}{Q\rho} \red^h \conf{\Gamma'}{W'}$, 
in which case $h \leq k_P \leq k$. It follows that $\conf{\Gamma}{(P+Q)\rho}$ is 
well-timed.\qedhere
\end{itemize}
\end{proof}

\paragraph{\textbf{Proof of Proposition \ref{cor:wf2wt}.}}
\label{proof:wf2wt}
We give the proof for a fragment of the language where channel 
restriction is omitted. This limitation is needed only to avoid technical 
complications in the proof of the statement. In fact, when channel restriction is present, 
we need to introduce a structural congruence $\equiv$ between system terms; the main 
property required by this relation is that it preserves transitions of configurations, meaning 
that whenever $W_1 \equiv W_2$ and $\conf{\Gamma}{W_1} \trans{\lambda} W_1'$, 
then $\conf{\Gamma}{W_2} \trans{\lambda}{W_2'}$, with $W_2 \equiv W_2'$. 
Also, the relation $\equiv$ needs to be defined so that any system term $W$ 
can be rewritten in the form $\nu \tilde{c}.\left(\prod_{i=1}^n P_i\right)$. 
See \cite{phdthesis}, Definition \textbf{9.1.2} at Page 174, for the definition of 
the structural congruence .

Let us focus on the case in which channel restriction is not present in our language
First note that the result holds for any well-formed 
configuration of the form $\conf{\Gamma}{P}$, where $P$ is a closed process; in fact we have that, 
$\conf{\Gamma}{P} = \conf{\Gamma}{P\rho}$ for any process environment $\rho$, and 
the latter is well-timed by Proposition \ref{prop:welltimed}.

Otherwise, we can rewrite $\conf{\Gamma}{W}$ as $\conf{\Gamma}{\prod_{i=1}^n P_i}$, 
for some processes $P_1, \cdots, P_n$. Note that each configuration $\conf{\Gamma_i}{P_i}$ 
is well-formed, hence well-timed; by definition there exists an index $k_{P_i} \geq 0$ such 
that, whenever $\conf{\Gamma}{P_i} \red_i^{h} \conf{\Gamma'}{P_i'}$, then $h \leq k_{P_i}$. 
Now suppose that $\conf{\Gamma}{\prod_{i =1}^n P_i} \red_{i}^h \conf{\Gamma'}{\prod_{i=1}^n P'_i}$; 
we show that $h \leq \left(\sum_{i=1}^n k_{P_i}\right)$ by induction on $h$. 

The case $h=0$ is trivial; suppose then that $h > 0$, and the statement is valid 
for $h -1$; in this case we can rewrite the (weak) reduction above as 
$\conf{\Gamma}{\prod_{i=1}^n P_i} \red_{i} \conf{\Gamma''}{\prod_{i=1}^n P''_i} \red_i^{h-1} \conf{\Gamma'}{\prod_{i=1}^n P'_i}$, 
and by inductive hypothesis $h-1 \leq \sum_{i=1^n} k_{P''_i}$. 
Let us focus on why $\conf{\Gamma}{\prod_{i=1}^n P_i} \red_i \conf{\Gamma''}{\prod_{i=1}^n P''_i} \red_i^{h-1}$.
\begin{enumerate}[label=(\roman*)]
\item $\conf{\Gamma}{\prod_{i=1}^n P_i }\trans{\tau} {\prod_{i=1}^n P''_i}$, and $\Gamma'' = \Gamma$; 
in this case it is not difficult to note that there exists an index $j: 1\leq j \leq n$ such that $\conf{\Gamma}{P_j} \trans{\tau} 
P''_j$, and for any index $i \neq j, 1 \leq i \leq n$, $P''_i = P_i$. 
In this case we have that $k_{P''_j} \leq k_{P_j} - 1$ 

Without loss of generality, let $j = 1$. Then we have that
\begin{align*}
h - 1 & \leq & &  \sum_{i=1}^n k_{P''_{i}} & &  =\\
&=&  k_{P''_{1}}& + & \sum_{i=2}^n k_{P_i} & \leq\\
&\leq& (k_{P_1} - 1) & + &  \sum_{i=2}^n k_{P_i} &=\\
&=& & \left(\sum_{i=1}^n k_{P_i}\right) - 1 & &
\end{align*}
Hence $h \leq \left(\sum_{i=1}^n k_{P_i}\right)$, as we wanted to prove;

\item Otherwise $\conf{\Gamma}{\prod_{i=1}^n P_i }\trans{c!v} {\prod_{i=1}^n P''_i}$, and $\Gamma'' = \gupd{c!v|}{\Gamma}$. 
In this case we can partition the set $\{1,\cdots,n\}$ into three sets $\{l\}$, $I$ and $J$ such that 
\textbf{(a)} $\conf{\Gamma}{P_l} \trans{c!v} \conf{\Gamma''}{P''_l}$ and $P'' = \sigma^{\delta_v}.Q$ for some process $Q$, 
\textbf{(b)} for any $i \in I$, $\isrcv{\conf{\Gamma}{P_i}, c} =\ttrue$ and $P''_i = \arcv c x Q_i$ for some process $Q_i$, 
\textbf{(c)} for any $j \in J$, $\isrcv{\conf{\Gamma}{P_j}, c}= \ffalse$ and $P''_j = P_j$. 
Note that \textbf{(a)} implies that $k_{P''_l} = 0$ and $1 \leq k_{P_l}$, \textbf{(b)} implies that $k_{P''_{i}} = 0$ for any $i \in I$ 
and \textbf{(c)} implies that $k_{P''_j} = k_{P_j}$ for any $j \in J$.

Without loss of generality, suppose that $l = 1$, $I = \{2,\cdots,m\}$ for some $m \leq n$, and $J = \{m+1, \cdot n\}$. 
In this case we have 
\begin{align*}
h - 1 &\leq & & & \sum_{i=1}^n k_{P''_{i}} & & & =\\
&=& k_{P''_1} & + &\left(\sum_{i=2}^m k_{P''_i} \right) & + & \left( \sum_{i = m+1}^{n} k_{P''_i}\right) &=\\
&=& 0 &+& 0 &+& \sum_{i = m+1}^n k_{P_i} &\leq\\
&\leq& (k_{P_1} - 1) &+& 0 &+& \sum_{i= m+1}^n k_{P_i}& \leq\\
&\leq& & & \sum_{i=1}^n k_{P_i} & & &
\end{align*}
Again the last inequation gives $h \leq \left(\sum_{i=1}^n k_{P_i}\right)$.\hfill\qed
\end{enumerate}

\begin{lemma}
\label{lem:reduction.independence}
Let us say that a system term $T$ is behaviourally 
independent from $W$ if 
each channel name appearing free in $T$ does not appear 
free in $W$, and vice versa. 

If $T$ is independent from a configuration $W$, then 
whenever $\conf{\Gamma}{W | T} \red_i \confC$, then either 
\begin{enumerate}[label=(\roman*)]
\item $\confC = \conf{\Gamma'}{W | T'}$, and $\conf{\Gamma}{T} \red_i \conf{\Gamma'}{W'}$, or 
\item $\confC = \conf{\Gamma'}{W' | T}$, and $\conf{\Gamma}{T} \red_i \conf{\Gamma'}{W'}$. 
\end{enumerate}
\end{lemma}

\begin{proof}
Suppose that $T$ is a system term independent from a configuration $\conf{\Gamma}{W}$, and 
that $\conf{\Gamma}{W | T} \red_i \confC$. By the definition of instantaneous reductions, 
there are two possibilities: 
\begin{enumerate}
\item $\conf{\Gamma}{W | T} \trans{\tau} \widehat{W}$, and $\confC = \conf{\Gamma}{\widehat{W}}$. 
By Proposition \ref{prop:parallel-components}(1) then either $\widehat{W} = W' | T$, and $\conf{\Gamma}{W} 
\trans{\tau} W'$, or $\widehat{W} = W | T'$, and $\conf{\Gamma}{T} \trans{\tau} T'$; in the first case 
we obtain the reduction $\conf{\Gamma}{W | T} \red_i \conf{\Gamma}{W' | T}$, while in the second 
one we get $\conf{\Gamma}{W | T} \red_i \conf{\Gamma}{W | T'}$, 
\item the second possibility is that $\conf{\Gamma}{W | T} \trans{c!v} \widehat{W'}$, 
and $\confC = \conf{\Gamma'}{\widehat{W}}$, where $\Gamma' = \gupd{c!v}{\Gamma}$. In this case, by Proposition 
\ref{prop:parallel-components}\ref{par3} then $\widehat{W} = W' | T'$ and either 
\begin{enumerate}
\item $\conf{\Gamma}{W} \trans{c!v} W'$, $\conf{\Gamma}{T} 
\trans{c!v} T'$; the first transition is possible only if $c$ appears free in $W$, which by assumption 
gives that $c$ does not appear free in $T$; it follows that $\isrcv{\conf{\Gamma}{T}, c} = \ffalse$, 
and by Lemma \ref{lem:rcv-enabling} we obtain that $T' = T$. By converting the intensional 
transition in a reduction (recalling that $\Gamma' = \gupd{c!v}{\Gamma}$), we obtain that 
$\conf{\Gamma}{W | T} \red_i \conf{\Gamma'}{W' | T}$, 
\item or $\conf{\Gamma}{W} \trans{c?v} W'$, $\conf{\Gamma}{T} \trans{c!v} T'$; this case can be handled symmetrically 
to the previous one, and leads to $\conf{\Gamma}{W | T} \red_i \conf{\Gamma'}{W | T'}$.\qedhere
\end{enumerate}
\end{enumerate}
\end{proof}

\begin{lemma}
\label{lem:reduction.preservation}
Let $\conf{\Gamma_1}{W}$ be a configuration, and let 
$\Gamma_2$ be a channel environment such that, for any channel $c$ appearing free in $W$, 
$\Gamma_2(c) = \Gamma_1(c)$. Then if $\conf{\Gamma_1}{W} \red \conf{\Gamma_1'}{W'}$, 
there exists a channel environment $\Gamma_2'$ such that $\conf{\Gamma_2}{W} \red \conf{\Gamma_2'}{W_2'}$, 
and $\Gamma_1'(c) = \Gamma_2'(c)$ for any $c$ appearing free in $W$.
\end{lemma}

\begin{proof}[Outline of the proof]
The reduction $\conf{\Gamma_1}{W} \red_i \conf{\Gamma_1'}{W'}$ can be converted 
in a transition of the form $\conf{\Gamma_1}{W} \trans{\lambda} W'$, where $\lambda$ 
takes either the form $\tau$, $c!v$ or $\sigma$. Note here that if $\lambda$ takes the 
form $c!v$, then $c$ appears free in $W$. 
By performing an induction on the proof of the derivation of this transition we can infer 
a transition for the configuration $\conf{\Gamma_2}{W}$, 
namely $\conf{\Gamma_2}{W} \trans{\lambda} W'$. 
Also, by letting $\Gamma_2' = \gupd{\lambda}{\Gamma_2}$, 
we obtain the reduction $\conf{\Gamma_2}{W} \red \conf{\Gamma_2'}{W'}$. 
Now it remains to note that if $c$ appears free then, by hypothesis, 
$\Gamma_1(c) = \Gamma_2(c)$; hence $\Gamma_1'(c) = \gupd{\lambda}{\Gamma_1}(c) 
= \gupd{\lambda}{\Gamma_2}(c) = \Gamma_2'(c)$.
\end{proof}

\begin{corollary}[Independence of Computations]
\label{cor:computation.independence}
Let $\conf{\Gamma}{W}$ be a configuration, and let $T$ be a system term 
which only uses fresh channels. Then whenever $\conf{\Gamma}{W | T} \red^{\ast} 
\conf{\Gamma''}{\widehat{W}}$ it follows that 
$\widehat{W} = W' | T'$ for some $W', T'$ such that $\conf{\Gamma}{W'} \red^\ast 
\conf{\Gamma'}{W'}$, where $\Gamma'$ is such that $\Gamma'(c) = \Gamma''(c)$ 
for any $c$ appearing free in $W$.
\end{corollary}

\begin{proof}[Outline]
By induction on the number of derivations $k$ in a sequence of $k$ reductions, 
$\conf{\Gamma}{W | T} \red^k \conf{\Gamma''}{\widehat{W}}$; in the inductive 
step it is necessary to distinguish whether the first reduction of the sequence 
is instantaneous or timed. In the first case, the result follows from 
lemmas \ref{lem:reduction.independence} and \ref{lem:reduction.preservation}. 
In the second case, we need to recover the timed transitions for the individual 
components $\conf{\Gamma}{W}$ and $\conf{\Gamma}{T}$, then apply Lemma 
\ref{lem:reduction.preservation}.
\end{proof}

\paragraph{\textbf{Proof of Lemma \ref{lem:fresh} (Outline).}}
\label{proof:fresh}
\MHc{This is a variation on analogous results already given in the literature, for a number of different
process calculi. We show that the relation
\begin{align*}
& \Ss = \{ (\conf{\Gamma_1}{W_1} , \conf{\Gamma_2}{W_2})\;:\; & \\ 
& \conf{\Gamma'_1}{W_1 | T_1} \simeq \conf{\Gamma'_2}{W_2 | T_2} \text{ for some } T_1, T_2 \text{ independent from both } W_1, W_2 & \\
& \text{ and } \conf{\Gamma_1}(c) = \Gamma_1'(c), \Gamma_2(c) = \Gamma_2'(c) \text{ whenever } c \text{ appears free in } W\} &
\end{align*}
is barb preserving, reduction closed and contextual. Note that it is necessary to employ Corollary \ref{cor:computation.independence} 
to prove that $\Ss$ is reduction closed.\qed
}

\paragraph{\textbf{Proof of Proposition \ref{prop:reduction.isolation}:}}
\label{proof:reduction.isolation}
The two statements are proved separately. Let $\conf{\Gamma_1}{W_1}, \conf{\Gamma_2}{W_2}$ be 
well-formed, and suppose that $\conf{\Gamma_1}{W_1} \simeq \conf{\Gamma_2}{W_2}$.
\begin{enumerate}
\item Suppose that $\conf{\Gamma_1}{W_1} \red_i \conf{\Gamma_1'}{W_1'}$. We have two possible cases, 
according to the definition of $\red_i$:
\begin{enumerate}[label=(\roman*)]
\item       $\conf {\Gamma_1} {W_1} \trans{\tau}  {W'_1}$ and 
				$\Gamma'_1 = \gupd{\tau}{\Gamma_1} = \Gamma_1$, by an application of rule \rulename{TauExt}
\item       $\conf {\Gamma_1} {W_1} \trans{c!v}  {W'_1}$ and $\Gamma'_1 = \gupd{c!v}{\Gamma_1}  $, by an application of rule \rulename{Shh}.
\end{enumerate}
We consider the first case; the proof for the second case is virtually identical. 
  Let \eureka be a fresh channel; that is it does not appear free in $W_1$ and must satisfy 
  $\Gamma_1 \vdash \eureka : \cfree$. 
Let $\arb$ be a message which requires one  time unit 
to be transmitted, i.e.\ $\delta_{\arb} =1$. 
By an application of rules \rulename{TauPar} and \rulename{TauExt}
we derive 
\[
\conf {\Gamma_1} {W_1 | \bcastzeroc \eureka \arb} \exttrans{\tau}
\conf {\Gamma'_1}{W'_1 | \bcastzeroc \eureka \arb}
\]
with $\conf {\Gamma'_1}{W'_1 | \bcastzeroc \eureka \arb}  \Downarrow_{\eureka}$
and $\Gamma'_1 \vdash \eureka : \cfree$. By Definition~\ref{def:step}
this transition corresponds in the reduction semantics to 
\[
\conf {\Gamma_1} {W_1 | \bcastzeroc \eureka \arb} \red 
\conf {\Gamma'_1}{W'_1 | \bcastzeroc \eureka \arb}
\]
As  $\conf {\Gamma_1} {W_1} \simeq \conf {\Gamma_2} {W_2} $ and $\simeq$ is 
contextual, this step must be matched by a sequence of
reductions 
\begin{align}\label{eq:steps2}
   \conf {\Gamma_2} {W_2 | \bcastzeroc \eureka \arb} \red^*    \confC
\end{align}
such that $\conf  {\Gamma'_1}   {W'_1 | \bcastzeroc \eureka \arb} \simeq \confC$. 
Depending on whether the transmission at $\eureka$ is part 
of the sequence of reductions or not, the configuration $\confC$ must be
one of the following: 
\[
\begin{array}{rclcl}
\confC_1 & = &\conf {\Gamma'_2}{W'_2 | \bcastzeroc \eureka \arb} &
\mbox{ with } & \Gamma'_2\vdash \eureka : \cfree\\
\confC_2 & = & \conf {\Gamma'_2}{W'_2 | \sigma.\nil} & \mbox{ with } & \Gamma'_2\vdash \eureka : \cbusy\\
\confC_3 & = & \conf {\Gamma'_2}{W'_2 | \nil}& \mbox{ with } & \Gamma'_2\vdash \eureka : \cfree
\end{array}
\]
As $\eureka$ is a fresh channel (hence not appearing free in $W_2$, it follows that $\calC_3\not\Downarrow_{\eureka}$; 
therefore $\calC$ 
cannot be $\calC_3$. Since $\conf  {\Gamma'_1}   {W'_1 | \bcastzeroc \eureka \arb} \simeq \confC$ and 
$\Gamma'_1 \vdash \eureka : \cfree$, by Proposition~\ref{prop:exposure} 
(which can be applied since we are assuming that $\calC$ 
is well-formed, hence well-timed) 
it follows that $\calC$ cannot be $\calC_2$. 
So, the only possibility is $\calC = \calC_1$. By Lemma~\ref{lem:fresh}
it follows that $\conf {\Gamma'_1}{W'_1} \simeq \conf {\Gamma'_2}{W'_2}$. 
It remains to show that 
$\conf {\Gamma_2} {W_2} \red_{i}^{\ast} \conf {\Gamma'_2} {W'_2}$.

 To this end we can extract out from the reduction sequence (\ref{eq:steps2}) above a reduction
sequence 
\[
   \conf {\Gamma_2} {W_2 } \red^*    \conf {\Gamma'_2} {W'_2 }
\]
We show that each step in this sequence, say $ \conf {\Gamma} {W }
\red \conf {\Gamma'} {W '}$, corresponds to an instantaneous reduction, 
$ \conf {\Gamma} {W } \red_i \conf {\Gamma'}
{W' }$, from which the result follows.

Recall from Definition~\ref{def:step} that there are three possible ways to infer the reduction step 
 $ \conf {\Gamma} {W } \red   \conf {\Gamma'} {W '}$. If it is either 
(Internal), i.e.\ $\conf \Gamma W \trans{\tau} W'$, or a (Transmission), 
i.e.\ $\conf \Gamma W \trans{c!v} W'$,   then by definition 
$ \conf {\Gamma} {W } \red_i   \conf {\Gamma'} {W' }$ follows. Condition (ii), (Time), is not possible
because in the original sequence (\ref{eq:steps2}) above the testing component 
$\bcastzeroc \eureka \arb$ 
can not make a $\sigma$ move, hence it cannot perform a timed reduction $\red_{\sigma}$.

\item Suppose now that $\conf{\Gamma_1}{W_1} \red_{\sigma} \conf{\Gamma_1'}{W_1'}$. 
In this case we will use the testing context:
  \begin{align*}
    T = \delay \sigma {( \tau. \bcastzeroc \eureka \arb   
+  \bcastzeroc \cfail \no)  }
  \end{align*}
where  \eureka and $\cfail$ are fresh channels. 
Since $\conf{\Gamma_1}{W_1} \red_{\sigma} \conf{\Gamma_1'}{W_1'}$ we also have 
$\conf{\Gamma_1}{W_1 | T} \red_{\sigma}\red_{i} \confC_1$, 
where $\conf{C_1} = \conf{\Gamma_1'}{W' | \bcastzeroc \eureka \arb )}$. Note that, 
since $\cfail$ is a fresh channel, we have that $\confC_1 \Downarrow_{\eureka}$ and 
$\confC_1 \not\Downarrow_{\cfail}$.

The contextuality of $\simeq$ gives that $\conf{\Gamma_1}{W_1 | T} \simeq \conf{\Gamma_2}{W_2 | T}$, 
so that we must have the series of reduction steps
 \begin{align}\label{eq:more.sigma}
 \conf   {\Gamma_2} {W_2 |T }  \red^* \confC_2
 \end{align}
where $\confC_1 \simeq \confC_2$. Because $\confC_1 \Downarrow_{\eureka}$ and $\confC_1 \not\Downarrow_{\cfail}$, 
the same must be true of $\confC_2$. 
As $\Gamma'_1 \vdash \eureka : \cfree$, it follows that $\calC_2$  must take the form
$\conf {\Gamma'_2} {W'_2 |   \bcastzeroc \eureka \arb}$. By Lemma~\ref{lem:fresh} we have that 
$\conf {{\Gamma'_1}} {W'_1} \simeq \conf {\Gamma'_2} {W'_2} $. It remains to establish that
$\conf {\Gamma_2} {W_2}  \red_i^{\ast}\red_{\sigma}\red_i^{\ast}  \conf {\Gamma'_2} {W'_2} $. 

We proceed as in the previous proposition, by extracting out of
(\ref{eq:more.sigma}) the contributions from $\conf {\Gamma_2}{W_2}$;
we know that because of the presence of the time delay in $T$, one time unit 
needs to pass before the broadcast along $\eureka$ is enabled in $\conf{\Gamma_2}{W_2 | T}$; 
also, by maximal 
progress (Proposition~\ref{prop:maximal-progress}), we know that such a broadcast must be 
fired before time passes.  So \MHc{ (\ref{eq:more.sigma})
actually takes the form}
\begin{align*}
  \conf {\Gamma_2}{W_2 |T }  \MHc{\,\red_i^*\,} \conf{\Gamma'}{W' | \ldots} \red_{\sigma} 
           \conf{\Gamma''}{W'' | \ldots } \MHc{\,\red_i^*\,\,} \conf {\Gamma'_2}{W'_2 |  \bcastzeroc \eureka \arb} 
\end{align*}
\MHc{Each individual reduction step can now be projected on to the first component, giving the required}
\[%begin{align*}
  \conf {\Gamma_2}{W_2}  \red_i^{\ast} \conf{\Gamma}{W} \exttrans{\sigma} \conf{\Gamma'}{W'} 
                  \red_i^{\ast} \conf {\Gamma'_2}{W'_2} \eqno{\qEd}  
\]%end{align*}
\end{enumerate}

\paragraph{\textbf{Proof of Proposition \ref{prop:input.detection}.}}
\label{proof:input.detection}
The two implications are proved separately; 
first, let $\conf{\Gamma}{W}$ be a configuration such that $\conf{\Gamma}{W} \extTrans{c?v} \conf{\Gamma'}{W'}$; 
that is, $\conf{\Gamma}{W} \extTrans{} \conf{\Gamma^{\pre}}{W^{pre}} \exttrans{c?v} \conf{\Gamma^{\post}}{W^{\post}} 
\extTrans{}\conf{\Gamma'}{W'}$. 
Since $T_{c?v}$ does not contain any receiver, nor does $T^{\checkmark}_{c?v}$, 
we have the sequences of transitions $\conf{\Gamma}{W | T_{c?v}} \extTrans{} \conf{\Gamma^{\pre}}{W^{\pre} | T_{c?v}}$ and 
$\conf{\Gamma^{\post}}{W^{\post} | T^{\checkmark}_{c?v}} \extTrans{}\conf{\Gamma'}{W' | T^{\checkmark}_{c?v}}$.

Next we show that $\conf{\Gamma^{\pre}}{W^{\pre} | T_{c?v}} \exttrans{\tau} \conf{\Gamma^{\post}}{W^{\post}} | 
T^{\checkmark}_{c?v}$. Combined with the two (weak) transitions above, this gives the extensional transition 
$\conf{\Gamma}{W | T_{c?v}} \extTrans{} \conf{\Gamma}{W' | T^{\checkmark}_{c?v}}$, which can 
be rewritten as $\conf{\Gamma}{W | T_{c?v}} \red_i^{\ast} \conf{\Gamma}{W' | T^{\checkmark}_{c?v}}$. 

Consider then the transition $\conf{\Gamma^{\pre}}{W^{pre}} \exttrans{c?v} \conf{\Gamma^{\post}}{W^{\post}}$; 
this can only have been obtained by the intensional transition  $\conf{\Gamma^{\pre}}{W^{\pre}} \trans{c?v} W'$, and 
the equality $\Gamma^{\post} = \gupd{c?v}{\Gamma^{\pre}}$.
For the test $T_{c?v}$ we have the transition $\conf{\Gamma^{\pre}}{T_{c?v}} \trans{c!v} {T^{\checkmark}_{c?v}}$; 
Now we can combine the two transitions together, using Rule \rulename{sync}, and get $\conf{\Gamma^{\pre}}{W^{\pre} | T_{c?v}} 
\trans{c!v} W^{\post} | T^{\checkmark}_{c?v}$; also, we know that $\Gamma^{\post} = \gupd{c?v}{\Gamma^{\pre}} = 
\gupd{c!v}{\Gamma^{\pre}}$, hence we can infer the required transition 
$\conf{\Gamma^{\pre}}{W^{\pre} | T_{c?v}} \exttrans{\tau} \conf{\Gamma^{\post}}{W^{\post}} | 
T^{\checkmark}_{c?v}$.

For the other implication, suppose that $\conf{\Gamma}{W | T_{c?v}} \red_i^{\ast} \conf{\Gamma'}{W' | T^{\checkmark}_{c?v}}$. 
This is possible only if, at some point in the sequence, the test component $T_{c?v}$ fired the broadcast along channel 
$c$; in fact, we have that the broadcast along channel $\eureka$ is guarded by a broadcast action in $T_{c?v}$, while it 
is guarded by a delay of $\delta_v$ instants of time in $T^{\checkmark}_{c?v}$. Also, by Maximal Progress 
(Proposition \ref{prop:maximal-progress}) the broadcast performed by $T_{c?v}$ must happen before 
time elapses; formally, we have the sequence of reductions 
\[
\conf{\Gamma}{W | T_{c?v}} \red_i^{\ast} \conf{\Gamma^{\pre}}{W^{\pre} | T_{c?v}} \red_i \conf{\Gamma^{\post}}{W^{\text{\post}} | T_{c?v}^{\checkmark}}\red_i^{\ast} \conf{\Gamma'}{W' | T^{\checkmark}_{c?v}}
\] 
Now note that the sequence of instantaneous reductions  
\begin{equation}
\label{eq:input.sequence}
\conf{\Gamma}{W | T_{c?v}} \red_i^{\ast} \conf{\Gamma^{\pre}}{W^{\pre} | T_{c?v}}
\end{equation} 
induces the extensional transition $\conf{\Gamma}{W} \extTrans{} \conf{\Gamma^{\pre}}{W^{\pre}}$. 
This can be proved using the facts that, for any channel environment $\Gamma_x$ and channel $d$, 
whenever $\conf{\Gamma_x}{T_{c?v}} \exttrans{\tau} \conf{\Gamma_x'}{T'}$, then $T' = T_{c?v}$, 
and whenever $\conf{\Gamma_x}{T_{c?v}} \extTrans{\tau} \conf{\Gamma_x'}{T'}$ then 
$T' \neq T_{c?v}$.

Similarly, we can prove that the weak reduction 
\[
\conf{\Gamma^{\post}}{W^{\post} | T_{c?v}^{\checkmark}}\red_i^{\ast} \conf{\Gamma'}{W' | T^{\checkmark}_{c?v}}
\] 
induces the extensional transition $\conf{\Gamma^{\post}}{W^{\post}} \extTrans{} \conf{\Gamma'}{W'}$. 

It remains to show that we can infer the transition 
$\conf{\Gamma^{\pre}}{W^{\pre}} \exttrans{c?v} \conf{\Gamma^{\post}}{W^{\post}}$ from the reduction 
$\conf{\Gamma^{\pre}}{W^{\pre} | T_{c?v}} \red_i 
\conf{\Gamma^{\post}}{W^{\post} | T_{c?v}^{\checkmark}}$. 
Note that in $T_{c?v}$ we have a station which is ready to broadcast along channel $c$, while 
this is not true anymore in $T^{\checkmark}_{c?v}$.
By performing a case analysis on the intensional transition which could have led to the reduction above, 
we find that the only possible case is that 
$\conf{\Gamma^{\pre}}{W^{\pre} | T_{c?v}} \trans{c!v} 
W^{\post} | T_{c?v}^{\checkmark}$ and, more specifically, that 
$\conf{\Gamma^{\pre}}{W^{\pre}} \trans{c?v} W^{\post}$ and 
$\conf{\Gamma^{\pre}}{T_{c?v}} \trans{c!v} T^{\checkmark}_{c?v}$. 
Also, $\conf{\Gamma^{\post}} = \gupd{c!v}{\conf{\Gamma}^{\pre}}$. 
By an application of Rule \rulename{Input} in the extensional semantics, 
we get the required transition 
$\conf{\Gamma^{\pre}}{W^{\pre}} \exttrans{c?v} \conf{\Gamma^{\post}}{W^{\post}}$, 
which can be combined with the two weak transitions already derived, 
namely $\conf{\Gamma}{W} \extTrans{} \conf{\Gamma^{\pre}}{W^{\pre}}$ and 
$\conf{\Gamma^{\post}}{W^{\post}} \extTrans{} \conf{\Gamma'}{W'}$, 
to obtain $\conf{\Gamma}{W} \extTrans{c?v} \conf{\Gamma'}{W'}$.
\hfill\qed

\paragraph{\textbf{Proof of Proposition \ref{prop:exp.detection}.}}
\label{proof:exp.detection}
Suppose that $\conf{\Gamma}{W} \extTrans{\iota(c)} \conf{\Gamma'}{W'}$. 
This can be rewritten as $\conf{\Gamma}{W} \extTrans{} \conf{\Gamma^\pre}{W^\pre} 
\exttrans{\iota(c)} \conf{\Gamma^{\post}}{W^\post} \extTrans{}\conf{\Gamma'}{W'}$. 
Since the only rule of the extensional semantics that could have been used to 
derive $\conf{\Gamma^\pre}{W^\pre} \exttrans{\iota(c)} \conf{\Gamma^\post}{W^\post}$ 
is \rulename{Idle}, we obtain that $\conf{\Gamma^\pre}{W^\pre} = \conf{\Gamma^\post}{W^\post}$. 
Thus, we have $\conf{\Gamma}{W} \extTrans{} \conf{\Gamma^\pre}{W^\pre} = \conf{\Gamma^\post}{W^\post} 
\extTrans{}\conf{\Gamma'}{W'}$, or equivalently $\conf{\Gamma}{W} \extTrans{\conf{\Gamma'}{W'}}$. 
In terms of the reduction semantics, this can be rewritten as $\conf{\Gamma}{W} \red_i^{\ast} \conf{\Gamma'}{W'}$.

By Corollary \ref{cor:iota.actions} we know that $\conf{\Gamma}{W} \extTrans{\iota(c)}$ 
implies $\conf{\Gamma}{W} 
\exttrans{\iota(c)} \conf{\Gamma}{W}$; therefore $\Gamma \vdash c: \cfree$. 
Now it is easy to see that we have the reduction 
$\conf{\Gamma}{W | T_{\iota(c)}} \red_i \conf{\Gamma}{W |T^{\checkmark}_{\iota(c)}}  
\red_i^\ast \conf{\Gamma'}{W' | T_{\iota(c)}}$, 
where the first reduction has been obtained 
by letting the predicate $\expsd{c}$ be evaluated in $T_{\iota(c)}$, 
while the rest of the sequence can be derived using the facts that 
$\conf{\Gamma}{W} \red_i^{\ast} \conf{\Gamma'}{W'}$, and for 
any channel environment $\Gamma_x$ we have that 
$\conf{\Gamma_x}{T^{\checkmark}_{\iota(c)}} \not\red_i$, 
$\conf{\Gamma_x}{T^{\checkmark}_{\iota(c)}} \trans{c?v} T'$ implies 
$T' = T^{\checkmark}_{\iota(c)}$.

Conversely, suppose that $\conf{\Gamma}{W | T_{\iota(c)}} \red^i_{\ast} \conf{\Gamma'}{W' | T_{\iota(c)}^{\checkmark}}$. 
In this sequence of reductions, the evolution of the test component from $T_{\iota(c)}$ to $T^\checkmark_{\iota(c)}$ 
is possible only if eventually the exposure check on channel $c$ is evaluated to true. 
That is, we have the sequence of reductions  
\[\conf{\Gamma}{W | T_{\iota(c)}} \red_i^{\ast} \conf{\Gamma^\pre}{W^\pre | T_{\iota(c)}} 
\red_i \conf{\Gamma^\post}{W^\post | T^\checkmark_{\iota(c)}} \red_i^{\ast} \conf{\Gamma'}{W' | T^\checkmark_{\iota(c)}}
\] 
where $\Gamma^\pre \vdash c: \cfree$. 

Since the evaluation of the exposure check in the reduction 
$\conf{\Gamma^\pre}{W^\pre | T_{\iota(c)}} \red_i \conf{\Gamma^\post}{W^\post | T^\checkmark_{\iota(c)}}$ 
corresponds to a $\tau$-intensional transition which affects only the system term $T_{\iota(c)}$, 
that is $\conf{\Gamma^\pre}{T_{\iota(c)}} \trans{\tau} T^\checkmark_{\iota(c)}$, Proposition 
\ref{prop:parallel-components}(1) ensures that $W^\post = W^\pre$, and $\Gamma^\post = 
\gupd{\tau}{\Gamma^\pre} = \Gamma^\pre$. Using the facts that 
$\conf{\Gamma^\pre}{W^\pre} = \conf{\Gamma^\post}{W^\post}$ and $\Gamma^\pre \vdash c: \cfree$, 
we can apply Rule \rulename{Idle} of the extensional semantics and infer the transition 
$\conf{\Gamma^\pre}{W^\pre} \exttrans{\iota(c)} \conf{\Gamma^\post}{W^\post}$.

Next, note that for any configuration $\Gamma_x$, we have that 
$\conf{\Gamma_x}{T_{\iota(c)}} \trans{d?v} T'$ implies $T' = T_{\iota(c)}$, 
and $\conf{\Gamma_x}{T_\iota(c)} \red_i \conf{\Gamma'_x}{T'}$ implies 
$T' \neq T_{\iota(c)}$. Similar results hold for the system term $T^\checkmark_{\iota(c)}$. 
Using these facts, it is not difficult can derive the 
extensional transition $\conf{\Gamma}{W} \extTrans{} \conf{\Gamma^\pre}{W^\pre}$ 
from the sequence of reductions $\conf{\Gamma}{W | T_{\iota(c)}} \red_i^\ast\conf{\Gamma^\pre}{W^\pre}$, 
and the transition 
$\conf{\Gamma^\post}{W^\post} \extTrans{}\conf{\Gamma'}{W'}$ from 
the sequence of reductions $\conf{\Gamma^\post}{W^\post | T^\checkmark_{\iota(c)}} 
\red_i^\ast \conf{\Gamma'}{W' | T^\checkmark_{\iota(c)}}$.

Thus we have proved that $\conf{\Gamma}{W} \extTrans{} \conf{\Gamma^\pre}{W^\pre} \exttrans{\iota(c)} \conf{\Gamma^\post}{W^\post} 
\extTrans{} \conf{\Gamma'}{W'}$, or equivalently $\conf{\Gamma}{W} \extTrans{\iota(c)} \conf{\Gamma'}{W'}$.
\hfill\qed

\paragraph{\textbf{Proof of Proposition \ref{prop:delivery.detection}.}}
\label{proof:delivery.detection}
For any value $w$, let $T_w$ be the system term 
\[
T_w = \nu d:(0,\cdot).((\matchb {w=v} {\bcastzeroc d \arb} {\nil}) + \bcastzeroc \cfail \no | 
\sigma.\matchb{\expsd{d}}{\bcastzeroc \eureka \arb}{\nil}) 
\]

Suppose that $\conf{\Gamma}{W} \extTrans{\gamma(c,v)} \conf{\Gamma'}{W'}$. 
In particular, we have that $\conf{\Gamma}{W} \extTrans{} \conf{\Gamma^\pre}{W^\pre}  
\exttrans{\gamma(c,v)}
\conf{\Gamma^\post}{W^\post} \extTrans{} \conf{\Gamma'}{W'}$. 
From the transition $\conf{\Gamma^\pre}{W^\pre} \exttrans{\gamma(c,v)} \conf{\Gamma^\post}{W^\post}$ 
we get that $\Gamma^{\pre} = (1,v)$, and $\conf{\Gamma^\pre}{W^\pre} \trans{\sigma} W^\post$. 
In particular, note that $\Gamma^ \pre \vdash c: \cbusy$, hence $\conf{\Gamma}^{\pre}{W^{\pre} | T_{\iota(c,v)}}$ 
is well formed.
Note also that $\conf{\Gamma_x}{T_{\gamma(c,v)}} \not\red_i$ for any environment $\Gamma_x$ with 
$\Gamma_x \vdash c: \cbusy$, and that $\conf{\Gamma_x}{T_{\gamma(c,v)}} \trans{c?v} T'$ implies that 
$T' = T_{\gamma(c,v)}$. Also, since $\Gamma^\pre(c) = (1,v)$, we obtain the transition 
$\conf{\Gamma^\pre}{T_{\iota(c)}} \trans{\sigma} T_v$. 
Finally, note that, for any channel environment $\Gamma_x$ we also have the transition 
$\conf{\Gamma_x}{T_v} \trans{\tau} T^\checkmark_{\gamma(c,v)}$.
Using these facts, we can build the sequence of transitions 
\[
\conf{\Gamma}{W | T_{\gamma(c,v)}} \red_i^{\ast} \conf{\Gamma^\pre}{W^\pre | T_{\gamma(c,v)}} 
\red_{\sigma} \conf{\Gamma^\post}{W^\post | T_v} \red_i^{\ast} \conf{\Gamma'}{W' | T_v} 
\red_i \conf{\Gamma'}{W' | T^\checkmark_{\gamma(c,v)}}
\]

Now suppose that $\conf{\Gamma}{W}{T_{\gamma(c,v)}} \red_i^{\ast}\red_{\sigma}\red_{i}^{\ast} \conf{\Gamma}{W' | T_{\gamma(c,v)}}$; 
we need to show that $\conf{\Gamma}{W} \extTrans{\gamma(c,v)} \conf{\Gamma'}{W'}$. 
Note that, in order for the testing component $T_{\gamma(c,v)}$ to evolve into $T^\checkmark_{\gamma(c,v)}$, then 
\begin{enumerate}
\item when the first time instant passes, the test evolves into $T_w$ for some value $w$; this is because in $T^\checkmark_{\gamma(c,v)}$ 
the active receiver along channel $c$ has vanished, and  in CCCP active receivers along a channel $c$ can only disappear after a timed 
reduction has been performed, and only if the state of channel $c$ changes from exposed to idle,
\item at some point, in the remaining of the computation, the matching construct $[w=v]$ is evaluated in $T_w$, leading 
to the test component to evolve in $T^\checkmark_{\gamma(c,v)}$. Note that the matching construct $[w=v]$ cannot be 
evaluated to false, as this would cause the test component to evolve to a system term different from $T^\checkmark_{\gamma(c,v)}$. 
Therefore, $w = v$, and more specifically $T_{w} = T_v$. 
\item The evaluation of the matching construct $[v = v]$ to true is modelled as an $\tau$-intensional action, hence it 
does not affect the tested component $W$.
\end{enumerate} 
Formally, we have a sequence of reductions 
\begin{eqnarray*}
&\conf{\Gamma}{W | T_{\gamma(c,v)}} \red_i^{\ast} \conf{\Gamma^\pre}{W^\pre | T_{\gamma(c,v)}} &\red_{\sigma}\\
\red_{\sigma}& \conf{\Gamma^\post}{W^\post | T_v} \red_{i}^{\ast} \conf{\Gamma''}{W'' | T_v} &\red_i\\ 
\red_i & \conf{\Gamma''}{W'' | T^\checkmark_{\gamma(c,v)}}
\red_{i}^\ast\conf{\Gamma'}{W' | T^\checkmark_{\gamma(c,v)}}&
\end{eqnarray*}
where $\Gamma^\pre(c) = (1,v)$.

Let $T$ be either $T_{\gamma}{c}, T_v$ or $T^\checkmark_{\gamma{c,v}}$, and let $\Gamma_x$ be 
an arbitrary channel environment; note that we have 
that $\conf{\Gamma_x}{T} \trans{d?v} T'$ implies $T' = T$, and $\conf{\Gamma_x}{T} \red_i \conf{\Gamma'_x}{T'}$ 
implies that $T' \neq T$. Using these facts, it is not difficult to derive the transitions 
\begin{enumerate}[label=\({\alph*}]
\item $\conf{\Gamma}{W} \extTrans{} \conf{\Gamma^\pre}{W^\pre}$, 
\item $\conf{\Gamma^\post}{W^\post} \extTrans{}{\conf{\Gamma''}{W''}}$, 
\item $\conf{\Gamma''}{W''} \extTrans{} \conf{\Gamma'}{W'}$
\end{enumerate}
Thus, we only need to show that 
$\conf{\Gamma^\pre}{W^\pre} \exttrans{\gamma(c,v)} \conf{\Gamma^\post}{W^\post}$. 
The timed reduction $\conf{\Gamma^\pre}{W^\pre | T_{\gamma(c,v)}} 
\red_{\sigma} \conf{\Gamma^\post}{W^\post | T_v}$ can only be inferred if  
$\conf{\Gamma^\pre}{W^\pre} \trans{\sigma} W^\post$, 
$\conf{\Gamma^\pre}{T_{\gamma(c,v)}} \trans{\sigma} T_v$ and $\Gamma^\post = 
\gupd{\sigma}{\Gamma^\pre}$. Also, note that the only possibility for inferring 
the transition $\conf{\Gamma^\pre}{T_{\gamma(c,v)}} \trans{\sigma} T_v$ 
is by using an instance of Rule $\rulename{EndRcv}$ (where the channel environment 
contains value $v$ at channel $c$); therefore, we obtain that $\Gamma^\pre(c) = (1,v)$. 

We have proved that $\Gamma^\pre(c) = (1,v)$, $\conf{\Gamma^\pre}{W^\pre} \trans{\sigma} 
W^\post$ and $\Gamma^\post = \gupd{\sigma}{\Gamma^\pre}$; therefore, we can apply Rule \rulename{Deliver} 
to infer that $\conf{\Gamma^\pre}{W^\pre} \exttrans{\gamma(c,v)} \conf{\Gamma^\post}{W^\post}$, as 
we wanted to show. By combining this transition with the weak transitions listed in (a), (b), (c), above, 
we obtain the required $\conf{\Gamma}{W} \extTrans{\gamma(c,v)} \conf{\Gamma'}{W'}$.
\hfill\qed

%see for example Lemma 2.38 of \cite{dpibook}.

\end{document}

%% file: net.topology.tex
\begin{align*}
     \begin{tikzpicture}
          \node[state](s){$s$}; 
          \node[state](r)[right=of s]{$r$}; 
          \node 			(env)[right=of r]{};
 \path[to]
       (s) edge[thick] (r)
       (r) edge[thick] (env);
%   \begin{pgfonlayer}{background}
%    \node [background,fit=(s) (r)] {};
%    \end{pgfonlayer}
    \end{tikzpicture}
&&
    \begin{tikzpicture}
          \node[state](s){$s$}; 
          \node[state](r)[right=of s]{$r$};
          \node[state](e)[below=of r]{$e$}; 
          \node				(env)[right=of r]{};  
 \path[to]
       (s) edge[thick,text=d] (r)
       (e) edge[thick] (r)
       (r) edge[thick] (env);
%   \begin{pgfonlayer}{background}
%    \node [background,fit=(s) (r) (e)] {};
%    \end{pgfonlayer}
    \end{tikzpicture}
\\
{\mathcal N}_0
&&
{\mathcal N}_1
%
%\Gamma_M \with \Cloc{\tau.(c!\pc{v}.\Cnil \probc{0.81}) \Cnil}{m})
%&&
%\Gamma_N \with \Cloc{\tau.(c!\pc{v} \probc{0.9} \Cnil)}{m} \Cpar \Cloc{c?\pa{x}.(c!\pa{x} \probc{0.9} c!\pa{x})}{n}
\end{align*}

%%% Local Variables: 
%%% mode: latex
%%% TeX-master: "broadcast"
%%% End: 

%% file: tdma.tex
\pgfdeclarelayer{background}
\pgfdeclarelayer{foreground}
\pgfsetlayers{background,main,foreground}

% Define a few styles and constants
\tikzstyle{sensor}=[draw, fill=blue!20, text width=2em, 
    text centered, minimum height=2em]
\tikzstyle{receiver}=[draw, fill=red!20, text width=2em, 
    text centered, minimum height=2em]
\tikzstyle{ann} = [above, text width=5em]
\tikzstyle{naveqs} = [sensor, text width=2em, fill=red!20, 
    minimum height=1em, rounded corners]
\def\blockdist{1.0}
\def\edgedist{2.5}

\begin{tikzpicture}
    \node (s00) [sensor] {$!v_0^0$};
    \path (s00)+(\blockdist,0) node (s01) [sensor]{$\sigma$};
    \path (s01)+(\blockdist,0) node (s02) [sensor]{$!v_0^1$};
    \path (s02)+(\blockdist,0) node (s03) [sensor]{$\sigma$};
    \node (s0) [below of=s01]{$s_0$};
    \begin{pgfonlayer}{background}
    \path (s00.north west)+(-0.2,0.2) node (a) {};
        \path (s0.south -| s03.east)+(+0.2,-0.2) node (b) {};
        \path[fill=blue!10,rounded corners, draw=black!50]
            (a) rectangle (b);
    \end{pgfonlayer}
    \path (s02)+(4*\blockdist,0) node (s10) [sensor] {$\sigma$};
    \path (s10)+(\blockdist,0) node (s11) [sensor]{$!v_1^0$};
    \path (s11)+(\blockdist,0) node (s12) [sensor]{$\sigma$};
    \path (s12)+(\blockdist,0) node (s13) [sensor]{$!v_1^1$};
    \node (s1) [below of=s11]{$s_1$};
    \begin{pgfonlayer}{background}
    \path (s10.north west)+(-0.2,0.2) node (a) {};
        \path (s1.south -| s13.east)+(+0.2,-0.2) node (b) {};
        \path[fill=blue!10,rounded corners, draw=black!50]
            (a) rectangle (b);
    \end{pgfonlayer}
    \path (s00)+(0,-4*\blockdist) node (r00) [receiver] {$?x$};
    \path (r00)+(\blockdist,0) node (r01) [receiver]{$\sigma$};
    \path (r01)+(\blockdist,0) node (r02) [receiver]{$?y$};
    \path (r02)+(\blockdist,0) node (r03) [receiver]{$\sigma$};
    \node (r0) [below of=r01]{$r_0$};
    \begin{pgfonlayer}{background}
    \path (r00.north west)+(-0.2,0.2) node (a) {};
        \path (r0.south -| r03.east)+(+0.2,-0.2) node (b) {};
        \path[fill=red!10,rounded corners, draw=black!50]
            (a) rectangle (b);
    \end{pgfonlayer}
		\path (s10)+(0,-4*\blockdist) node (r10) [receiver] {$\sigma$};
    \path (r10)+(\blockdist,0) node (r11) [receiver]{$?x$};
    \path (r11)+(\blockdist,0) node (r12) [receiver]{$\sigma$};
    \path (r12)+(\blockdist,0) node (r13) [receiver]{$?y$};
    \node (r1) [below of=r11]{$r_1$};
    \begin{pgfonlayer}{background}
    \path (r10.north west)+(-0.2,0.2) node (a) {};
        \path (r1.south -| r13.east)+(+0.2,-0.2) node (b) {};
        \path[fill=red!10,rounded corners, draw=black!50]
            (a) rectangle (b);
    \end{pgfonlayer}
	\path [draw, ->] (s0.south)+(0,-0.2) -- node [left] {$d$} 
        (r01.north)+(0,0.2);
  \path [draw, ->] (s1.south)+(0,-0.2) -- node [left] {$d$} 
        (r11.north)+(0,0.2);
\end{tikzpicture}

%% file: net.topology.2.tex
\begin{tikzpicture}
          \node[state](s0){$s_0$};
          \node[state](s1)[above=of s0]{$s_1$}; 
          \node[state](r)[right=of s1]{$r$}; 
          \node[white](env)[right=of s0]{};
 \path[to]
       (s0) edge[thick] (env)
       (s1) edge[thick] (r)
       (r) edge[thick] (env);
%   \begin{pgfonlayer}{background}
%    \node [background,fit=(s) (r)] {};
%    \end{pgfonlayer}
    \end{tikzpicture}